\documentclass[acmsmall]{acmart}

\setcopyright{rightsretained}
\acmJournal{PACMPL}
\acmYear{2024} 
\acmVolume{8} 
\acmNumber{POPL} 
\acmArticle{84} 
\acmMonth{1} 
\acmDOI{10.1145/3632926}
\copyrightyear{2024}
\acmSubmissionID{popl24main-p582-p}
\received{2023-07-11}
\received[accepted]{2023-11-07}

\usepackage{mathtools}
\usepackage{relsize}
\usepackage{tikz}
\usepackage{proofzilla}
\usepackage{tikzit}

\usetikzlibrary{arrows.meta}

\tikzstyle{empty rect}=[fill=none, draw=black, shape=rectangle, minimum height=.6cm, minimum width=.6cm]
\tikzstyle{rect}=[fill=white, draw=black, shape=rectangle, minimum height=.6cm, minimum width=.6cm]
\tikzstyle{red rect}=[fill=white, draw=red, shape=rectangle, minimum height=.6cm, minimum width=.6cm]
\tikzstyle{condition}=[fill=white, draw=black, shape=circle, minimum size=0.6cm, inner sep=0pt]
\tikzstyle{token}=[draw=black, fill=black, shape=circle, inner sep=0pt, minimum size=0.2cm]
\tikzstyle{pt condition}=[fill=white, draw=black, shape=circle, minimum size=.8cm, inner sep=0pt]
\tikzstyle{big pt condition}=[fill=white, draw=black, shape=circle, minimum size=1.1cm, inner sep=0]
\tikzstyle{negative token}=[fill=white, draw=black, shape=circle, minimum size=0.2cm, inner sep=0pt, thick]
\tikzstyle{small rect}=[fill=white, draw=black, shape=rectangle]
\tikzstyle{BN node}=[fill=white, draw=black, shape=ellipse, align=center, minimum width=2.3cm, inner sep=.2em]
\tikzstyle{white rect}=[fill=white, draw=none, shape=rectangle]
\tikzstyle{circle basic}=[fill=white, draw=black, shape=circle, minimum size=.9cm, inner sep=0cm]

\tikzstyle{big arrow}=[->, {-{Latex[length=3mm,width=2mm]}}]
\tikzstyle{arrow}=[->]
\tikzstyle{dashed 0}=[-, dashed]
\tikzstyle{red dash dotted}=[-, dashdotted, draw=red]
\tikzstyle{red none}=[-, draw=red]
\tikzstyle{thick none}=[-, very thick]
\tikzstyle{dashed arrow}=[->, dashed]
\tikzstyle{residual arrow}=[->, {-{Latex[length=2mm,width=1mm]}}]
\tikzstyle{red residual arrow}=[->, {-{Latex[length=2mm,width=1mm]}}, draw=red]
\tikzstyle{plate}=[-, draw=gray, very thick]
\tikzstyle{dotted line}=[-, dotted]

\usepackage{graphicx}
\usepackage{bussproofs}
\usepackage{fancybox}
\usepackage{cmll}
\usepackage{stmaryrd}
\usepackage{multirow}

\usepackage{proof}
\usepackage{hhline}
\usepackage{xspace}
\usepackage{booktabs}
\usepackage{subcaption}
\usepackage{wrapfig}
\usepackage{array}
\usepackage{arydshln}
\usepackage{commath}
\usepackage{microtype}
\usepackage{xcolor}
\usepackage{cleveref}
\usepackage[notextcomp,notext,nomath,not1]{stix}

\newcommand{\tm}{t}
\newcommand{\tmtwo}{u}
\newcommand{\tmthree}{r}

\newcommand{\var}{x}
\newcommand{\vartwo}{y}

\newcommand{\la}[2]{\lambda #1.#2}

\newcommand{\val}{v}
\newcommand{\valtwo}{w}

\newcommand{\ctxholep}[1]{[#1]}
\newcommand{\ctxhole}{\ctxholep{\cdot}}

\newcommand{\evctx}{\mathsf{E}}

\newcommand{\nbvctxtwo}[1]{\nbvctxtwo{#1}}

\newcommand{\sctx}{\mathsf{S}}

\newcommand{\sctxp}[1]{\ctxholep{#1}\sctx}

\newcommand{\defeq}{:=}
\newcommand{\grameq}{::=}

\newcommand{\isub}[2]{\{#1\leftarrow#2\}}
\newcommand{\esub}[2]{[#1\leftarrow#2]}

\newcommand{\grammarpipe}{\mathrel{\big |}}

\newcommand{\ie}{\textit{i.e.}\xspace}
\newcommand{\eg}{\textit{e.g.}\xspace}
\newcommand{\ih}{\textit{i.h.}\xspace}

\newcommand{\red}[1]{{\color{red} {#1}}}
\newcommand{\blue}[1]{{\color{blue} {#1}}}

\newcounter{numberone}

\newcounter{numbertwo}

\newcommand{\tjudg}[3]{#1\vdash #2:#3}

\renewcommand{\int}[1]{\{#1\}}
\newcommand{\mset}[1]{[#1]}

\newcommand{\tmt}{\mathtt{t}}
\newcommand{\tmf}{\mathtt{f}}

\newcommand{\rabs}{\textsc{abs}}
\newcommand{\rvar}{\textsc{var}}

\newcommand{\Bool}{\mathbb{B}}

\newcommand{\coin}[1]{\mathtt{coin}_{#1}}

\newcommand{\pair}[2]{\langle#1,#2\rangle}
\newcommand{\letin}[3]{\mathtt{let}\;#1=#2\;\mathtt{in}\;#3}
\newcommand{\letp}[3]{\mathtt{letp}\;#1=#2\;\mathtt{in}\;#3}
\newcommand{\letpin}[4]{\mathtt{letp} \ \pair{#1}{#2} \ = \ #3 \ \mathtt{in} \ #4}
\newcommand{\leti}{\mathtt{let}\;}
\newcommand{\inl}{\;\mathtt{in}}

\newcommand{\dist}[1]{\{#1\}}
\renewcommand{\dist}{d}

\newcommand{\isZ}[1]{\mathtt{isZero}(#1)}

\newcommand{\pred}[1]{\mathtt{pred}(#1)}

\newcommand{\fix}[2]{\mathtt{fix}\;#1.#2}

\newcommand{\der}{\mathtt{der}\,}
\newcommand{\bang}{\oc}

\newcommand{\caseSym}{\scalebox{0.5}[1]{ $ \Rightarrow $ }}
\newcommand{\case}[2]{\texttt{case}\;#1\;\texttt{of}\; \{#2\}}

\newcommand{\rcond}{\textsc{s-cond}\xspace}

\newcommand{\rpair}{\textsc{s-pair}\xspace}
\newcommand{\rlet}{\textsc{s-let}\xspace}
\renewcommand{\rvar}{\textsc{s-var}\xspace}
\newcommand{\rletp}{\textsc{s-letp}\xspace}
\newcommand{\rbang}{\textsc{s-bang}\xspace}
\newcommand{\rder}{\textsc{s-der}\xspace}
\newcommand{\rapp}{\textsc{s-app}\xspace}
\renewcommand{\rabs}{\textsc{s-abs}\xspace}
\newcommand{\rsample}{\textsc{s-sample}\xspace}
\newcommand{\robs}{\textsc{s-obs}\xspace}
\newcommand{\rbool}{\textsc{s-bool}\xspace}
\newcommand{\icond}{\textsc{i-cond}\xspace}
\newcommand{\pcond}{\textsc{p-cond}\xspace}
\newcommand{\icoin}{\textsc{i-sample}\xspace}
\newcommand{\ipair}{\textsc{i-pair}\xspace}
\newcommand{\ilet}{\textsc{i-let}\xspace}
\newcommand{\plet}{\textsc{p-let}\xspace}
\newcommand{\ivar}{\textsc{i-var}\xspace}
\newcommand{\iletp}{\textsc{i-letp}\xspace}
\newcommand{\ibang}{\textsc{i-bang}\xspace}
\newcommand{\ider}{\textsc{i-der}\xspace}
\newcommand{\iapp}{\textsc{i-app}\xspace}
\newcommand{\iabs}{\textsc{i-abs}\xspace}
\newcommand{\isample}{\textsc{i-sample}\xspace}
\newcommand{\psample}{\textsc{p-sample}\xspace}
\newcommand{\iobs}{\textsc{i-obs}\xspace}

\newcommand{\exponential}{multiset\xspace}

\newcommand{\Names}{\textsf{Names}}

\renewenvironment{itemize}
{
	\begin{list}{\labelitemi}
		{\setlength{\itemsep}{0pt}
			\setlength{\topsep}{0pt}
			\setlength{\parsep}{0pt}
			\setlength{\partopsep}{0pt}
			\setlength{\leftmargin}{15pt}
			\setlength{\rightmargin}{0pt}
			\setlength{\itemindent}{0pt}
			\setlength{\labelsep}{5pt}
			\setlength{\labelwidth}{10pt}
	}}
	{
	\end{list} 
}

\renewenvironment{enumerate}
{
	\begin{list}{\arabic{numberone}.}
		{
			\usecounter{numberone}
			\setlength{\itemsep}{0pt}
			\setlength{\topsep}{0pt}
			\setlength{\parsep}{0pt}
			\setlength{\partopsep}{0pt}
			\setlength{\leftmargin}{15pt}
			\setlength{\rightmargin}{0pt}
			\setlength{\itemindent}{0pt}
			\setlength{\labelsep}{5pt}
			\setlength{\labelwidth}{15pt}
	}}
	{
	\end{list} 
}

\let\markeverypar\everypar
\newtoks\everypar
\everypar\markeverypar
\markeverypar{\the\everypar\looseness=-1\relax}

\newcommand{\RED}[1]{{ \color{red}{#1}}}

\newcommand{\pink}[1]{{\color{magenta}{#1}}}
\newcommand{\violet}[1]{{\color{violet}{#1}}}

\def\eg{{\it e.g.}\xspace}
\def\ie{{\it i.e.}\xspace}

\DeclareMathSymbol{\mhyphen}{\mathord}{AMSa}{"39}

\newcommand{\st}{\text{ s.t. }}
\newcommand{\tand }{\text{ and }}
\newcommand{\tor }{\texttt{ or }}

\newcommand{\BN}{Bayesian network\xspace}
\newcommand{\BNs}{Bayesian networks\xspace}

\newcommand{\Val}[1]{\mathsf{Val}(#1)} 

\newcommand{\bn}{\mathcal B}
\newcommand{\FProd}{\odot} 
\newcommand{\Fprod}{\odot}
\newcommand{\BigFProd}{\bigodot} 

\newcommand{\bX}{\mathbb X}
\newcommand{\bY}{\mathbb Y}
\newcommand{\bZ}{\mathbb Z}

\newcommand{\bE}{\mathbb E}

\newcommand{\bW}{\mathbb W}

\newcommand{\x}{ \mathtt{x}}
\newcommand{\y}{ \mathtt{y}}
\newcommand{\z}{\mathtt {z}}

\newcommand{\bx}{{\overline \x}}
\newcommand{\by}{{\overline \y}}
\newcommand{\bz}{{\overline \z}}

\newcommand{\xs}{\bx}
\newcommand{\ys}{\by}

\newcommand{\BigO}[1]{\mathcal{O}(#1)}
\newcommand{\Exp}[1]{\mathsf{exp}(#1)}

\newcommand{\sA}{A}
\newcommand{\sB}{B}
\newcommand{\sC}{C}
\newcommand{\sP}{P}
\newcommand{\sQ}{Q}

\newcommand{\sL}{L}
\newcommand{\sK}{K}
\newcommand{\cA}{\underline A}
\newcommand{\cB}{\underline B}

\newcommand{\cP}{\underline P}
\newcommand{\cQ}{\underline Q}

\newcommand{\cL}{\underline L}
\newcommand{\cK}{\underline K}
\newcommand{\cAp}[1]{\cA\ctxholep{#1}}
\newcommand{\cPp}[1]{\cP\ctxholep{#1}}
\newcommand{\cLp}[1]{\cL\ctxholep{#1}}
\newcommand{\cKp}[1]{\cK\ctxholep{#1}}
\newcommand{\cAnp}[2]{\cA_{#1}\ctxholep{#2}}

\newcommand{\cLnp}[2]{\cL_{#1}\ctxholep{#2}}
\newcommand{\cKnp}[2]{\cK_{#1}\ctxholep{#2}}
\newcommand{\cLamnp}[2]{\underline{\Lambda}_{#1}\ctxholep{#2}}
\newcommand{\cGamnp}[2]{\underline{\Gamma}_{#1}\ctxholep{#2}}

\newcommand{\One}{\mathsf 1} 

\newcommand{\bt}{\mathtt b}

\newcommand{\tb}{\bt}

\newcommand{\iA}{\at}
\newcommand{\iB}{\bt}

\newcommand{\iP}{\at^+}
\newcommand{\iQ}{\ct^+}

\newcommand{\nX}{ \mathtt X}
\newcommand{\nY}{ \mathtt Y}
\newcommand{\nZ}{ \mathtt Z}
\newcommand{\nW}{ \mathtt W}
\newcommand{\nD}{\mathtt{D}}
\newcommand{\nS}{\mathtt{S}}
\newcommand{\nR}{\mathtt{R}}

\newcommand{\R}{\mathtt{R}}

\newcommand{\X}{ \nX}
\newcommand{\Y}{ \nY}
\newcommand{\Z}{ \nZ}
\newcommand{\W}{ \nW}

\newcommand{\ovdash}[1]{\overset{\violet{#1}}{\vdash}}

\newcommand{\dem}{\triangleright}

\newcommand{\true}{\mathtt{t}}
\newcommand{\false}{\mathtt{f}}
\newcommand{\Real}{\mathbb{R}}
\newcommand{\bool}{\mathbb{B}}

\newcommand{\sample}[1]{\mathtt{sample}_{#1}}

\renewcommand{\pair}[2]{\langle #1, #2 \rangle}

\newcommand{\sem}[1]{\llbracket {#1} \rrbracket}
\newcommand{\den}[1]{\overline{\sem{#1}}}

\newcommand{\arrow}{\multimap}

\newcommand{\lam}{\lambda}

\newcommand{\tmu}{u}

\newcommand{\Nm}[1]{\mathsf{Nm}(#1)}

\newcommand{\gpipe}{\grammarpipe}

\newcommand{\PP}{\mathcal P}
\newcommand{\LL}{\mathcal L}
\newcommand{\subs}[2]{\{#1\leftarrow#2\}}

\newcommand{\hole}[1]{\llparenthesis #1\rrparenthesis}
\newcommand{\eshole}[1]{\holebag{#1}}
\renewcommand{\ctxholep}[1]{\hole{#1}}
\newcommand{\shole}[1]{\holebag{#1}}
\renewcommand{\sctxp}[1]{\shole{#1}\sctx}
\newcommand{\scp}{\sctxp}

\renewcommand{\ss}{\evctx}

\newcommand{\rsym}{\mathsf{r}}
\newcommand{\Rule}{\rsym}

\newcommand{\Root}[1]{\mapsto_{#1}}

\newcommand{\eslist}{\mathsf{S}}

\newcommand{\mult}{\mathbf{m}}
\newcommand{\meas}[1]{\mathtt{meas}(#1)}

\newcommand{\dB}{\mathsf{db}}
\newcommand{\dS}{\mathsf{dsub}}
\newcommand{\sval}{\mathsf{dsub}}
\newcommand{\dbang}{\mathsf{der \oc}}
\newcommand{\dpair}{\mathsf{pm}}  
\newcommand{\dom}{{\mathsf{dom}}}

\newcommand{\ft}[2]{  \overset{#2}{#1}  }

\renewcommand{\dmd}[1]{ {\color{blue} \, \diamond \, #1} }
\newcommand{\dft}[2]{~  {\color{blue} \, \diamond \, \overset{#2}{#1}} }
\newcommand{\ftone}{ \mathbf{1} }

\newcommand{\dmdC}[1]{~  {\color{violet} \, \diamond \, \scriptstyle{#1}} }
\newcommand{\dftC}[2]{~  {\color{violet} \, \diamond \, \scriptstyle{#2}} }

\newcommand{\Proj}[3][]{#2\vert_{#3}^{#1}} 

\newcommand{\Low}{\lambda_{\texttt{low}}}
\newcommand{\low}{\mathtt{low}}

\newcommand{\Bang}{ \lambda_{!} }

\newcommand{\iBang}{\texttt {iTypes}}

\newcommand{\larrow}{\multimap}

\newcommand{\Up}[1]{(#1)^-}
\newcommand{\Down}[1]{(#1)^+}
\newcommand{\up}[1]{{#1}^\uparrow}
\newcommand{\down}[1]{{#1}^\downarrow}

\renewcommand{\iA}{\sA}
\renewcommand{\iB}{\sB}

\renewcommand{\iP}{\sP}
\renewcommand{\iQ}{\sQ}
\newcommand{\iL}{\sL}
\newcommand{\iK}{\sK}

\newcommand{\cpts}{CPT's\xspace}

\newcommand{\pax}{probabilistic axiom\xspace}
\newcommand{\paxs}{probabilistic axioms\xspace}

\newcommand{\bv}{\mathtt {b}}
\newcommand{\bvs}{\mathtt{\overline {b}}}

\newcommand{\Cpts}[1]{\mathsf{Cpts}(#1)}
\newcommand{\flow}[1]{\mathsf{flow}(#1)}
\newcommand{\pattern}[1]{\langle #1\rangle}

\newemptyvertex{node}{}
\defedgetype{flow}{draw=blue}{}
\defedgetype{dirflow}{-latex,draw=blue}{}
\defedgetype{flowtwo}{draw=red}{}
\defedgetype{dirflowtwo}{-latex,draw=red}{}

\renewcommand{\defeq}{\triangleq}
\renewcommand{\coin}[1]{\mathtt{bernoulli}_{#1}}

\newcommand{\cpt}{CPT\xspace}

\newcommand{\condSym}{\textnormal{\ttfamily\bfseries c}}

\newcommand{\freevar}[1]{\scriptstyle{ {\langle} #1 {\rangle}}}

\newcommand{\CPT}[1]{\condSym {\freevar{#1}}}
\newcommand{\CPTN}[2]{{\condSym^{#1}\!\freevar{#2}}}

\newcommand{\cc}{\condSym}

\newcommand{\rv}{r.v.\xspace}
\newcommand{\rvs}{r.v.s\xspace}

\newcommand{\name}{name\xspace}
\newcommand{\names}{names\xspace}

\newcommand{\linear}{low-level\xspace}

\renewcommand{\bool}{\mathsf{B}}
\renewcommand{\Bool}{\bool}

\usepackage{amsthm}

    \newtheorem{theorem}{Theorem}[section]
\newtheorem*{theorem*}{Theorem} %
\newtheorem*{thm*}{Theorem}
\newtheorem*{prop*}{Proposition}
\newtheorem*{lemma*}{Lemma}
\newtheorem*{example*}{Example}
\newtheorem*{definition*}{Definition}
\newtheorem*{property*}{Property}
\newtheorem*{remark*}{Remark}

\newtheorem{thm}{Theorem}[section]

\newcommand{\version}{1}
\newcommand{\commento}{0}

\newcommand{\SLV}[2]{\ifthenelse{\equal{\version}{0}}{#1}{#2}}
\newcommand{\condinc}[2]{\ifthenelse{\equal{\commento}{0}}{#1}{{\color{orange}{#2}}}}

\newcommand{\Appendix}{Appendix\xspace}

\graphicspath{ {../images/} }

\usepackage[normalem]{ulem}
\usepackage[
textwidth=1.5cm,textsize=footnotesize]{todonotes} 

\usepackage{tcolorbox}
\newtcolorbox{mybox}
{
	colback=white,
	boxsep=0pt,left=4pt,right=4pt,top=4pt,bottom=4pt
}
\newtcolorbox{myboxC} 
{
	colback=blue!05!white,
	boxsep=0pt,left=4pt,right=4pt,top=4pt,bottom=4pt
}

\Crefname{section}{Sect.}{Sections}
\Crefname{theorem}{Thm.}{Thm.}
\Crefname{thm}{Thm.}{Thm.}
\Crefname{proposition}{Prop.}{Prop.}
\Crefname{prop}{Prop.}{Prop.}
\Crefname{definition}{Def.}{Def.}
\Crefname{Def}{Def.}{Def.}
\Crefname{figure}{Fig.}{Figs.}
\Crefname{equation}{Eq.}{Eqss.}

\AtEndPreamble{%
	\theoremstyle{acmdefinition}
	\newtheorem{remark}[theorem]{Remark}}

\AtEndPreamble{%
	\theoremstyle{acmdefinition}
	\newtheorem{notation}[theorem]{Notation}}

\AtEndPreamble{%
	\theoremstyle{acmdefinition}
	\newtheorem{fact}[theorem]{Fact}}

\newcommand\todoc[2][]{\todo[color=pink!20,#1]{#2}} 

\newcommand{\CF}[1]{\violet{**CF: #1}}

\begin{document}

\title{Higher Order Bayesian Networks, Exactly}
\SLV{}{\subtitle{{\large Extended Version}}}

\author{Claudia Faggian}
\orcid{0009-0009-8875-3595}
\affiliation{%
	\institution{IRIF, CNRS, Université Paris Cité}
	\country{France}
}
\email{faggian@irif.fr}

\author{Daniele Pautasso}
\orcid{0009-0008-8865-7942}
\affiliation{%
	\institution{University of Turin}
	\country{Italy}
}
\email{daniele.pautasso@unito.it}

\author{Gabriele Vanoni}
\orcid{0000-0001-8762-8674}
\affiliation{%
	\institution{IRIF, CNRS, Université Paris Cité}
	\country{France}
}
\email{gabriele.vanoni@irif.fr}

\begin{abstract}
Bayesian networks 
are  graphical \emph{first-order} probabilistic models that allow for a compact representation of large probability distributions, and for  efficient inference, both exact and approximate. We introduce  a \emph{higher-order} programming language---in the idealized form of a $\lambda$-calculus---which we prove \emph{sound and complete} w.r.t. \BNs: each \BN can be encoded as a term, and conversely each (possibly higher-order and recursive) program of ground type \emph{compiles} into a \BN. 

The  language allows for  the specification of recursive 
probability models and hierarchical structures. 
Moreover, we provide a \emph{compositional} and \emph{cost-aware} semantics which is based on factors, the standard mathematical tool  used in  Bayesian inference. 
Our results rely  on advanced  techniques rooted into linear logic, intersection types, rewriting theory,  and Girard's geometry of interaction, which are here combined in a novel way.
  
\end{abstract}



\maketitle





\section{Introduction}
This paper is  a  foundational study, taking a cost-aware approach to  the semantics of higher-order probabilistic programming languages.
Probabilistic models play a crucial role in several fields such as machine learning,  cognitive science, and applied statistics, with applications spanning from finance to biology. A prominent example of such models are Bayesian networks (BNs)~\cite{Pearl88},  a  (first-order, static) graphical formalism  
able to represent  complex  systems 
in a   \emph{compact} way and  enabling  \emph{efficient} inference algorithms.  
BNs decompose  
 large joint distributions into smaller \emph{factors}. These are used in inference algorithms, both exact  (such as message passing and variable elimination) and approximate (sampling-based).
Despite   their  significant strengths, the task of modeling using \BNs is comparable to the task of programming using logical circuits.

\paragraph{Probabilistic Programming Languages} A different approach is taken by \emph{probabilistic programming languages} (PPLs), where statistical models are specified as programs. The fundamental idea behind PPLs is to separate the model description---the program---from the computation of the probability distribution specified by the program---the inference task. This separation aims at making stochastic modeling as accessible as possible, hiding the underlying inference engines, which typically encompass various sampling methods such as importance sampling, Markov Chain Monte Carlo,  and Gibbs sampling. In this paper, we specifically focus on \emph{functional} PPLs, which allow for first-class higher-order functions and compositional semantics (\eg Church \cite{GoodmanMRBT08}, Anglican \cite{Wood14} and Venture \cite{mansinghka2014venture}).


\paragraph{Bayesian Networks and  PPLs}
 Bayesian networks and PPLs are closely interconnected.  
 On the one hand, Bayesian networks can be easily represented as simple, first-order, probabilistic programs~\cite{BUGS, KollerMP97}. On the other hand, several \emph{first-order} PPLs (such as BUGS, a widely used declarative  language, and Infer.NET) compile programs into a graphical model, specifically a Bayesian network, which is then utilized for performing inference tasks. For a detailed tutorial and an analysis of the involved subtleties, we refer to \cite{abs-1809-10756} (Ch.3).
 
%

\paragraph{Towards a Foundational Understanding.}


The research community is devoting considerable effort to establish a solid foundational understanding of functional PPLs. 
	A central   concept is  compositionality, which  is the key to reason in a modular (and scalable) way about programs.
 Such an understanding is crucial for the development of robust formal methods to  facilitate the analysis of probabilistic programs, the construction of Bayesian models, and the verification of inference correctness. Pioneering works by \citet{JacobsZ16, JacobsZ} have paved the way for a logical and semantical comprehension of Bayesian networks and inference from a categorical perspective. The majority of foundational papers, \eg \cite{HeunenKSY17, ScibiorKVSYCOMH18, DahlqvistSDG18, DBLP:journals/pacmpl/VakarKS19, Stein2021CompositionalSF}, have adopted an approach based on  category theory, which has yielded remarkable insights, enabling denotational proofs of correctness and  compositionality principles. 
This line of research however  
 does not 
take into account the  raison d'\^etre  of \BNs,  namely the space and time efficiency of inference.
 In the literature, compositional semantics and efficiency are typically explored as \emph{separate} entities. This dichotomy stands in stark contrast to  Bayesian networks, where the representation, its semantics (the defined joint distribution), and the inference algorithms are deeply \emph{intertwined}.

\paragraph{This paper.} Our   foundational investigation  adopts a cost-aware perspective.
We introduce a semantical framework that integrates the \emph{efficiency} of Bayesian networks with the \emph{expressiveness} of higher-order functional programming, and the \emph{compositional} nature of type systems. 
We adopt an idealized  functional PPL, namely an untyped $\lambda$-calculus enriched with probabilistic primitives.
%
Our language \emph{faithfully} encompasses conventional Bayesian networks (expressed here by first-order terms in normal form),   and comes equipped   with    a semantics and  formal methods which are \emph{resources-sensitive}. 
The core of our approach lies in semantical techniques, including rewriting theory and  type systems. These  form the groundwork for a compilation  scheme that translates higher-order terms into Bayesian networks, which serve as a low-level language.
We list below our main contributions:
\begin{itemize}
	\item \emph{A higher-order language for BNs.} We introduce  a probabilistic call-by-push-value $\lambda$-calculus that  we prove 
	\emph{sound and complete for BNs}:  not only any \BN can be encoded as a term (which is standard), but also---conversely---any   higher-order (possibly recursive)  program of ground type will eventually reduce to a first-order  normal form, corresponding to a \BN (\Cref{thm:compileI} and \Cref{prop:BN}). Our language also supports the encoding of advanced stochastic models, including \emph{template} \BNs, and  the specification of \emph{recursive} probability models. Notably, the operational semantics we define  corresponds to the process of unrolling the template into an actual BN, in line with the intended semantics of the templates.

	\item \emph{A factor-based semantics.} We endow each term of ground type with a \emph{factor}-based semantics. Factors, indeed, are the mathematical structure which   give     semantics to BNs. 
	Technically, computing this factor-based semantics requires tracking the generation and sharing of random variables---this task is achieved using non-idempotent intersection types.
	This technique allows us 
	 to extract a \BN $\bn_t$  from the  type derivation of a ground term $\tm$, and to  prove that the semantics (in the sense of  \BNs theory) of $\bn_\tm$  coincides with the semantics of $\tm$. 
	 
	\item \emph{The factor semantics is compositional.}
	The main technical achievement of the paper is to establish the compositionality of the factor semantics for typed terms  (\Cref{sec:term-semantics} and \Cref{sec:proof_of_connexity}).
	 This is particularly noteworthy since operations on factors do not exhibit this property, in general. The design of the type system plays a crucial role in ensuring compositionality, thereby enabling the modular reasoning that one expects in a high-level programming language.
	
	
	\item \emph{The factor semantics is resource-sensitive.}
	Our semantics takes resource consumption into account---its computational complexity (both in terms of space and time) is similar to that for 	Bayesian networks. 
	 Additionally, the type system offers a precise estimate of the  cost associated with computing the semantics of a term (i.e., the cost of calculating the \emph{exact} distribution defined by the term).
	
	\item \emph{Proof techniques.} 
			The proof of our results incorporates sophisticated techniques that have their foundation in linear logic, leveraging  significant advancements made over the past 15 years. Specifically, we employ intersection types, rewriting theory, linear logic, and concepts inspired by Girard's geometry of interaction in a novel and synergistic manner.
\end{itemize}

\vskip 8pt
\SLV{Proofs and more examples are available in the Technical Report \cite{long}.}
{Proofs and more examples are  in the Appendix.}


\newcommand{\Dry}{\texttt{Dry}}
\newcommand{\Rain}{\texttt{Rain}}
\newcommand{\Sprinkler}{\texttt{Sprinkler}}
\newcommand{\Wet}{\texttt{Wet}}

\newcommand{\dry}{\textsf{dry}}
\newcommand{\rain}{\textsf{rain}}
\newcommand{\sprinkler}{\textsf{sprinkler}}
\newcommand{\wet}{\textsf{wet}}

\newcommand{\bias}{\textsf{bias}}
\newcommand{\ctoss}{\textsf{coin}}

\section{About Bayesian Networks, Sharing,  and the $\lambda$-Calculus, Informally}\label{sec:informal}

\paragraph{Bayesian Reasoning.} 
In the Bayesian interpretation, probabilities describe \emph{degrees of belief} in events,
and  inference allows for  reasoning under uncertainty. 
One challenge in Bayesian reasoning is how to represent joint probability distributions. Indeed, these can quickly become  very large:
in general, a probability distribution over $n$ boolean variables requires storing $2^n$ values.
\BNs are able to express a joint probability distribution over several variables
in a compact way (\emph{factorized representation}), allowing for \emph{efficient inference},
without ever needing to reconstruct the full joint distribution. 

\paragraph{An Example of Bayesian Network.}

\begin{figure}
	\scalebox{.75}{ \begin{tikzpicture}
	\begin{pgfonlayer}{nodelayer}
		\node [style=BN node] (1) at (0, -3.5) {\large \textsf{WetGrass} \\ $\nW$};
		\node [style=BN node] (2) at (-3, 0) {\large \textsf{Sprinkler} \\ $\nS$};
		\node [style=BN node] (3) at (0, 3.5) { \large \textsf{DrySeason} \\ $\nD$};
		\node [style=BN node] (4) at (3, 0) {\large \textsf{Rain} \\ $\nR$};
		\node [style=white rect] (5) at (11.1, 3) {
            \small
            $\begin{array}{|cc|c|}
	            \hline	
	            \nD  &  \nR  & \text{Pr}(\nR|\nD) \\
	            \hline
	            \true	& \true   & 0.2 \\
	            \true	& \false  & 0.8 \\
	            \hline
	            \false  & \true   & 0.75 \\
	            \false	& \false  & 0.25 \\
	            \hline
            \end{array}$ 
        };
		\node [style=white rect] (6) at (-10.3, 4) {
            \small
            $\begin{array}{|c|c|}
	            \hline	
	            \nD & \text{Pr}(\nD) \\
	            \hline
	            \true	& 0.6 \\
	            \false	& 0.4 \\
	            \hline
            \end{array}$ 
        };
		\node [style=white rect] (7) at (-11, -1) {
            \small
            $\begin{array}{|cc|c|}
	            \hline	
	            \nD  &  \nS  & \text{Pr}(\nS | \nD) \\
	            \hline
	            \true	& \true   & 0.8 \\
	            \true	& \false  & 0.2 \\
	            \hline
	            \false  & \true   & 0.1 \\
	            \false	& \false  & 0.9 \\
	            \hline
            \end{array}$ 
        };
		\node [style=white rect] (8) at (12, -3) {
            \small
            $\begin{array}{|ccc|c|}
	            \hline	
	            \nS  & \nR  & \nW  & \text{Pr}(\nW|\nS,\nR) \\
	            \hline
	            \true	& \true  & \true	& 0.99 \\
	            \true	& \true  & \false	& 0.01 \\
	            \hline
	            \false  & \true  & \true	& 0.7 \\
	            \false	& \true  & \false	& 0.3 \\
	            \hline
	            \true	& \false & \true	& 0.9 \\
	            \true	& \false & \false	& 0.1 \\
	            \hline
	            \false  & \false & \true	& 0.01 \\
	            \false  & \false & \false	& 0.99 \\
	            \hline
            \end{array}$ 
        };
	\end{pgfonlayer}
	\begin{pgfonlayer}{edgelayer}
		\draw [style=big arrow] (3) to (4);
		\draw [style=big arrow] (3) to (2);
		\draw [style=big arrow] (2) to (1);
		\draw [style=big arrow] (4) to (1);
		\draw [style=dotted line] (6) to (3);
		\draw [style=dotted line] (8) to (1);
		\draw [style=dotted line] (2) to (7);
		\draw [style=dotted line] (5) to (4);
	\end{pgfonlayer}
\end{tikzpicture} }
	\caption{An example of \BN (from \cite{DarwicheBook}).}
	\label{fig:BNrain}
\end{figure}

Let us start with an informal example. We want
to model  the  fact that the lawn being  \texttt{Wet} in the morning
may depend on either \texttt{Rain} or the \texttt{Sprinkler} being on.
In turn, both \texttt{Rain} and the regulation of the  \texttt{Sprinkler}
depend on being or not in the \texttt{Dry Season}. 
The dependencies between these  four variables (shortened to $ \nD,\nS,\nR,\nW $)
are represented as arrows in \Cref{fig:BNrain}, while (for each variable)
the strength of the dependencies is quantified by a \emph{conditional probability table} (\cpt).
Assume we wonder: did it rain last night? 
Given that we are in \texttt{DrySeason}, our \emph{prior} belief is that \texttt{Rain}
happens with probability $ 0.2 $. However, if we observe that the lawn is \texttt{Wet},
our confidence increases.  The updated belief is called \emph{posterior}.
The model in \Cref{fig:BNrain} allows us to infer the posterior probability of \texttt{Rain},
given the \emph{evidence} that the lawn is \texttt{Wet}, \ie $\Pr(\R=\true | \W=\true)$. 
Posteriors  are typical queries which can be answered by \emph{Bayesian inference},
which has at its core Bayes conditioning, often condensed in the following informal formula:
\[
\textsf{Posterior} =  \textsf{Prior} \times \textsf{ Likelihood} \div \textsf{Evidence}
\]
Concretely, in our example:
\begin{align*}
	\Pr(\Rain=\true|\Wet=\true)  = \dfrac{\Pr(\Wet=\true|\Rain=\true) \Pr(\Rain=\true)}{\Pr(\Wet=\true)} 
			 =  \dfrac{\Pr(\Rain=\true,\Wet=\true)}{\Pr(\Wet=\true)} 
\end{align*}

So, to compute the posterior $\Pr(\mathtt{Rain}=\true|\mathtt{Wet}=\true)$, we have to  compute the \emph{marginal}  $\Pr({\Rain}=\true,{\Wet}=\true)$, which can be  obtained by 
summing out the other variables  from the joint probability. \emph{Marginalization} is illustrated in 
\Cref{fig:marginalization} (all rows which agree on the value of \Rain\ and \Wet\ are merged into a single row,  summing up the  probabilities). We remark that the joint distribution over $\nD,\nS,\nR, \nW$ has $2^4$ entries, 
although here we only display a subset of them.



\definecolor{myred}{RGB}{200, 20, 20}
\definecolor{mygreen}{RGB}{90, 150, 60}
\newcommand{\myred}[1]{{\color{myred}#1}}
\newcommand{\mygreen}[1]{{\color{mygreen}#1}}
\begin{figure}[t]
	\footnotesize	
	\begin{tabular}{ccc}
		Joint distribution over $\nD,\nS,\nR,\nW$  &  &  Marginal over $\nR,\nW$ \\
		$\begin{array}{|c|c|c|c|c|}
			\hline
			\nD & \nS &	\nR & \W  & \Pr(d,s,r,w)  \\
			\hline
			\true& \true& \mygreen{\true}  & \mygreen{\true} & \mygreen{0.09504}   \\
			\true & \false& \mygreen{\true}  & \mygreen{\true} &\mygreen{0.0168} \\
			\false& \true& \mygreen{\true}  & \mygreen{\true} & \mygreen{0.0297} \\
			\false& \false& \mygreen{\true}  & \mygreen{\true} & \mygreen{0.189} \\
			\true& \true& \myred{\true}  & \myred{\false} &   \myred{0.00096} \\
			\true & \false& \myred{\true}  & \myred{\false} & \myred{0.0072}\\
			\false& \true& \myred{\true}  & \myred{\false} & \myred{0.0003} \\
			\false& \false& \myred{\true}  & \myred{\false} & \myred{0.081} \\
			\multicolumn{4}{|c|}{\dots} & \dots\\
			\hline
		\end{array}$
		& 
		$\begin{array}{c}
			\Longrightarrow \\
				\text{summing  out  $\nD,\nS$} \\[10pt] 
		\end{array}$
		& 
		$\begin{array}{c}
			\begin{array}{|c|c|c|}
				\hline
				\nR  &  \nW &  \Pr(r,w) \\
				\hline
				\mygreen{\true}  & \mygreen{\true} & \mygreen{0.33} \\
				\myred{\true}  & \myred{\false} &   \myred{0.09} \\
				\multicolumn{2}{|c|}{\dots} & \dots\\
				\hline
			\end{array}\\[10mm]
			{{\Pr(r,w) = \mathlarger\sum\limits_{d,s\in\set{\true,\false}}\Pr(d,s,r,w)} }	\\	
		\end{array}$
	\end{tabular}
	\caption{Joint distribution corresponding to the \BN in \Cref{fig:BNrain}, and \emph{marginalization}.}
	\label{fig:marginalization}
\end{figure}
The  marginal probability  $\Pr(\Wet=\true)$ of the evidence is computed in a similar way (yielding $0.69$).
Then, we obtain the posterior by normalizing: $\Pr(\Rain=\true|\Wet=\true)=0.33  / 0.69 = 0.48$.
Further evidence (for example, the Sprinkler is broken) would  once again update our belief.
In practice, the   numerator of Bayes theorem (the unnormalized marginal) often suffices, since the actual posterior is proportional to it: 
\[\textsf{Posterior} \propto  \textsf{Prior} \times \textsf{Likehood }\]
Summing up, the key step in \emph{exact inference} is the  computation of marginals.


%

%
%
%

\paragraph{Bayesian Networks as Terms.}
In probabilistic programming, it is standard to describe a \BN with a term (see \eg \cite{GordonHNR14} for a brief tutorial). 
We can encode our initial example into a rather standard probabilistic call-by-value $\lambda$-calculus, as follows:
\begin{equation}
\begin{array}{lll}
	\leti  \textsf{dry}&= \coin{0.6} \inl \\
	\leti  \textsf{rain}&=  \case{\pattern{\dry}} {\tmt\caseSym\coin{0.8};\tmf\caseSym\coin{0.1}} \inl \\
	\leti  \textsf{sprinkler} & = \case{\pattern{\dry}}{\tmt\caseSym\coin{0.2}; \tmf\caseSym\coin{0.75}}   \inl  \\
	\leti \textsf{wet} & =  \texttt{case}\;  \pattern{\rain, \sprinkler} 
	\;\texttt{of}\;\{\pair{\tmt}\tmt\caseSym\coin{0.99}; \pair\tmt\tmf\caseSym\coin{0.7}; \\  &\hspace{4.1cm} \pair\tmf\tmt\caseSym\coin{0.9}; \pair\tmf\tmf\caseSym	\coin {0.01}\}\\
	\texttt{in }  \wet &
\end{array}
\label{fig:BNterm}
\end{equation}
The idea is that any \cpt can be encoded with a \texttt{case} construct, together with a probabilistic primitive $\texttt{sample}_d$, which returns a value sampled from a (countably supported) probability distribution $d$. In this example we sample from a Bernoulli distribution.
This way, all Bayesian networks can be encoded in a very basic fragment of the simply typed call-by-value $\lambda$-calculus (this will be discussed in the next sections).  What if we allow for a richer, non linear, $\lambda$-calculus?

\paragraph{Beyond Ground Bayesian Networks.} 
Standard Bayesian networks describe probabilistic models in an intuitive  and compact way, 
but have some  inherent limitations. They are a first-order model, lacking \emph{modularity and compositionality}.
A well-established way to increase the expressive power of Bayesian networks  
 are \emph{templates}\footnote{We refer to  \cite{KollerBook}, Ch. 6,  for a detailed presentation and pointers to the vast literature.},
which allow  for the description of \emph{hierarchical}  models and  for taking into account 
\emph{temporality}. Dynamic \BNs \cite{DBN} (in \Cref{fig:HMM}) are an instance  extensively used in 
real-world applications---\eg mobile robotics. A more structured approach is indeed
essential when models become large and  complex.  While standard Bayesian networks 
have a precise  mathematical definition, rooted in graph theory and statistics, 
templates are  a more informal notion. We show that techniques from functional programming language theory
can provide a neat and mathematically sound framework, founded on the theory of lambda calculus,
  exactly matching the intended semantics of templates.

\paragraph{Repeated Coin Tosses.} 
Let us consider the following  experiment, where repetition is involved.
\begin{enumerate}
	\item Sample a bias $r_i$ from a discrete distribution 
	(for simplicity, let us assume there are only two possible choices: $r_1$ or $r_2$).
	\item Toss $m$ times a coin of bias $r_i$. 
	\item Return the results of the $m$ (biased) coin tosses.
\end{enumerate}

It is standard to graphically describe such an experiment by means of the 
\emph{plate} notation \cite{BUGS, Buntine94}
(see \Cref{fig:template}), a graphical meta-formalism for  representing 
models with \emph{repeated structures} and \emph{shared parameters}.
A rectangular plate grouping random variables  indicates multiple 
copies of the sub-graph.
A number ($m$) is  drawn to represent the number of repetitions.  Unrolling 
the plate $m$ times  defines a ground  Bayesian network. Please notice that the intended \BN is the unrolled one.
In \Cref{fig:template},
we show the template which models our experiment (a), and the ground \BN resulting from its unrolling (b).

\begin{figure}[t]
	\centering
	\begin{minipage}[b]{0.42\linewidth}
		\centering
		\begin{subfigure}[b]{.3\textwidth}
		\centering
		\scalebox{.8}{ \begin{tikzpicture}
	\begin{pgfonlayer}{nodelayer}
		\node [style=none] (0) at (-1.75, 0.25) {};
		\node [style=none] (1) at (1.75, 0.25) {};
		\node [style=none] (2) at (-1.75, -3) {};
		\node [style=none] (3) at (1.75, -3) {};
		\node [style=circle basic] (4) at (0, -1.7) { $\textsf{coin}$ };
		\node [style=circle basic] (5) at (0, 1.75) { $\textsf{bias}$ };
		\node [style=none] (6) at (-1.125, -0.2) {\Large \color{darkgray} $m$};
		\node [style=none] (14) at (0, 3.75) {};
	\end{pgfonlayer}
	\begin{pgfonlayer}{edgelayer}
		\draw [style=plate] (0.center) to (2.center);
		\draw [style=plate] (2.center) to (3.center);
		\draw [style=plate] (0.center) to (1.center);
		\draw [style=plate] (1.center) to (3.center);
		\draw [style=big arrow] (5) to (4);
	\end{pgfonlayer}
\end{tikzpicture} }
		\caption{Plate}
		\end{subfigure}
		\begin{subfigure}[b]{.68\textwidth}
		\centering
		\scalebox{.8}{ \begin{tikzpicture}
	\begin{pgfonlayer}{nodelayer}
		\node [style=none] (2) at (-3.25, -4) {};
		\node [style=none] (3) at (3.25, -4) {};
		\node [style=circle basic] (4) at (-3.25, -2.75) {$\textsf{coin}_1$};
		\node [style=circle basic] (5) at (0, 0.75) {$ \textsf{bias}$ };
		\node [style=circle basic] (6) at (-0.75, -2.75) { $\textsf{coin}_2$};
		\node [style=circle basic] (7) at (3.25, -2.75) {$\textsf{coin}_m$};
		\node [style=none] (8) at (1.2, -2.75) {\LARGE ...};
		\node [style=none] (16) at (0, 2.75) {};
	\end{pgfonlayer}
	\begin{pgfonlayer}{edgelayer}
		\draw [style=big arrow] (5) to (4);
		\draw [style=big arrow] (5) to (6);
		\draw [style=big arrow] (5) to (7);
	\end{pgfonlayer}
\end{tikzpicture} }
		\caption{Unrolling $m$ times}
		\end{subfigure}
	\caption{Modeling $m$ coin tosses.}
	\label{fig:template}
	\end{minipage}
	\hfill
	\vline
	\hfill
	\begin{minipage}[b]{0.55\linewidth}
		\centering
		\begin{subfigure}[b]{.46\textwidth}
		\centering
		\scalebox{.8}{ \begin{tikzpicture}
	\begin{pgfonlayer}{nodelayer}
		\node [style=condition] (3) at (-2, 1.75) {$ \; \mathtt{S} \; $};
		\node [style=condition] (4) at (2, 1.75) {$\; \mathtt{S}' \;  $};
		\node [style=condition] (5) at (2, -1.75) {$\; \mathtt{O}' \; $};
		\node [style=none] (8) at (-3.75, 3) {};
		\node [style=none] (9) at (-0.25, 3) {};
		\node [style=none] (10) at (-3.75, -3) {};
		\node [style=none] (11) at (-0.25, -3) {};
		\node [style=none] (13) at (0.25, 3) {};
		\node [style=none] (14) at (3.75, 3) {};
		\node [style=none] (15) at (0.25, -3) {};
		\node [style=none] (16) at (3.75, -3) {};
		\node [style=none] (17) at (-2, 3.75) {time $t=0$};
		\node [style=none] (18) at (2, 3.75) {time $t+1$};
	\end{pgfonlayer}
	\begin{pgfonlayer}{edgelayer}
		\draw [style=big arrow] (3) to (4);
		\draw [style=big arrow] (4) to (5);
		\draw [style=plate] (8.center) to (9.center);
		\draw [style=plate] (9.center) to (11.center);
		\draw [style=plate] (8.center) to (10.center);
		\draw [style=plate] (10.center) to (11.center);
		\draw [style=plate] (13.center) to (14.center);
		\draw [style=plate] (14.center) to (16.center);
		\draw [style=plate] (13.center) to (15.center);
		\draw [style=plate] (15.center) to (16.center);
	\end{pgfonlayer}
\end{tikzpicture} }
		\caption{Template}
		\end{subfigure}
		\begin{subfigure}[b]{.46\textwidth}
		\centering
		\scalebox{.8}{ \begin{tikzpicture}
	\begin{pgfonlayer}{nodelayer}
		\node [style=condition] (3) at (-3, 1.75) {$\mathtt{S}_0$};
		\node [style=condition] (4) at (0, 1.75) {$\mathtt{S}_1$};
		\node [style=condition] (5) at (0, -1.75) {$\mathtt{O}_1$};
		\node [style=none] (8) at (-3.75, 3) {};
		\node [style=none] (10) at (-3.75, -3) {};
		\node [style=none] (14) at (3.75, 3) {};
		\node [style=none] (16) at (3.75, -3) {};
		\node [style=none] (17) at (-2, 3.75) {};
		\node [style=none] (18) at (2, 3.75) {};
		\node [style=condition] (19) at (3, 1.75) {$\mathtt{S}_2$};
		\node [style=condition] (20) at (3, -1.75) {$\mathtt{O}_2$};
	\end{pgfonlayer}
	\begin{pgfonlayer}{edgelayer}
		\draw [style=big arrow] (3) to (4);
		\draw [style=big arrow] (4) to (5);
		\draw [style=big arrow] (19) to (20);
		\draw [style=big arrow] (4) to (19);
	\end{pgfonlayer}
\end{tikzpicture} }
		\caption{Unrolling (2 times)}
		\end{subfigure}
	\caption{Discrete Time Dynamic \BN.}
	\label{fig:HMM}
	\end{minipage}
\end{figure}


\paragraph{Repetitions, Sharing, and PPL}
How can we describe  this experiment as a $\lambda$-term? Let us assume  $m=2$.
It is tempting to encode the template in the following way:
\[
\begin{array}{ll}
	\leti & \bias = \sample d \inl  \\
	\leti & \ctoss =  \case{\pattern{\bias}}{r_1\caseSym\coin{r_1}; r_2\caseSym\coin{r_2}} \\
	\texttt{in} & \langle {\ctoss},{\ctoss}\rangle
\end{array}
\]
Unfortunately, the call-by-value policy makes sure that all the instances of  \ctoss\ have the \emph{same} shared value, so the possible outcomes are  only $\pattern{ \true, \true}$ and $\pattern{\false,\false}$. Indeed, since only values can be substituted for variables,  all expressions, even the probabilistic ones, have to be evaluated before being substituted. Switching the evaluation order to call-by-name does not solve the problem: now \emph{all} the probabilistic primitives are copied \emph{before} being evaluated. This means that all   coin tosses are  independent, this way not belonging necessarily to the same coin: some could have bias $r_1$ and some others $r_2$. We need a finer evaluation mechanism that allows the programmer to say when values have to be shared, and  when instead we want to actually duplicate unevaluated expressions.

\paragraph{Call-by-Push-Value.} 
More than 20 years ago \citet{Levy99} introduced call-by-push-value
as a subsuming paradigm and functional/imperative synthesis,
refining the computational $\lambda$-calculus by \citet{Moggi89}.
The slogan was: ``a value \emph{is}, a computation \emph{does}'',
as this language tries to unify call-by-name and call-by-value,
providing two new primitives. The former \emph{thunks} a computation inside a value, and the latter \emph{forces} the evaluation of a value as a computation. Similar ideas have  been  independently developed in  the linear logic community, using Girard's translations.  There,  linear $\lambda$-calculi use the $\oc$ to thunk, and $\textsf{der}$(eliction) to force  \cite{BentonW96, Simpson05, MelliesT10, EggerMS14, Ehrhard16}.
\newcommand{\mycoin}{\CPT{\bias}}
\begin{figure}[t]
	\small
	\[
	\begin{array}{ll}
			\leti & {\bias}=\sample d \inl\\
			\leti & \ctoss = \, \oc \left( 	\mycoin  \right) \inl\\
			\leti & {y_1 } = \der \ctoss \inl\\
			\leti & {y_2 } = \der \ctoss \\
			\texttt{in} & \langle y_1,y_2\rangle
	\end{array}
	\ \to \
	\begin{array}{ll}
			\leti & \bias=\sample d \inl\\
			\leti & {y_1 } = \der \oc(\mycoin) \inl\\
			\leti & {y_2 } = \der \oc(\mycoin) \\	
			\texttt{in} & \langle y_1,y_2\rangle
	\end{array}
	\ \to^* \
	\begin{array}{ll}
				\leti & \bias=\sample d \inl\\
				\leti & {y_1 } = \mycoin \inl\\
				\leti & {y_2 } = \mycoin \\
				\texttt{in} & \langle y_1,y_2 \rangle
	\end{array} 
	\]
	\caption{The $\lambda$-term correctly modeling two coin tosses, and its reduction to normal form. $\CPT{\bias}$ stands for the conditional expression $\case{\pattern{\bias}}{r_1\caseSym\coin{r_1}; r_2\caseSym\coin{r_2}}$. The result stored in \textsf{bias} is correctly shared, while \textsf{coin} is copied before performing the toss, thus giving two independent and identically distributed values to  $y_1$ and $y_2$.}
	\label{fig:BNreduction}
\end{figure}
Having all of this in mind, we encode our experiment as the leftmost term of \Cref{fig:BNreduction}. 

\paragraph{Operational Semantics.}
The  simplest operational semantics for probabilistic programs is in terms of \emph{sampled} values. 
A standard approach in probabilistic $ \lambda $-calculi (including \cite{EhrhardT19}) is to
give the operational semantics via  Markov chains: sequential evaluation produces
distributions over  \emph{execution paths}.
Here, we follow a different route, because we want to model the unrolling of a
higher-order term into a ground Bayesian Network, as shown in \Cref{fig:BNreduction}.
By firing all the redexes but the probabilistic choices, the term $\tm$ reduces to a \emph{normal form}  that represents a standard, ground Bayesian network, exactly matching the unrolling of the template in \Cref{fig:template}.

\newcommand{\holebag}[1]{\langle\mkern-5.5mu\langle #1\rangle\mkern-5.5mu\rangle}

\newcommand{\BB}{\set{\true,\false}}

\newcommand{\observe}[2]{\mathtt{obs}(#1=#2)}
\newcommand{\obs}{\observe}
\newcommand{\obst}[1]{\observe{#1}{\true}}
\newcommand{\obsf}[1]{\observe{#1}{\false}}
\newcommand{\obsb}[1]{\observe{#1}{\bv}}

\newcommand{\Xb}{\X^{\bv}}
\newcommand{\Xt}{\X^{\true}}
\newcommand{\Xf}{\X^{\false}}
\newcommand{\Yb}{\Y^{\bv}}
\newcommand{\Yt}{\Y^{\true}}
\newcommand{\Yf}{\Y^{\false}}

\newcommand{\XX}{\mathcal{X}}
\newcommand{\YY}{\mathcal{Y}}
\newcommand{\ZZ}{\mathcal{Z}}

\section{Preliminaries on Calculus and Types}\label{sec:calculus}
This section presents the probabilistic programming language we are going to use throughout this paper. 
The language includes  constructs for describing \emph{sampling} and \emph{conditioning}. 
We have already argued why we opted for a language that is able to thunk computations and force values. 

\subsection{Syntax}\label{sec:syntax}
Our language, dubbed $ \Bang$-calculus, is a fragment of \cite{EhrhardT19} probabilistic call-by-push-value.
For ease of  presentation, in this paper we limit  \emph{ground types} to  booleans. This way, all random variables are assumed  binary, and as a consequence, we only sample from Bernoulli distributions. Generalizing the language to   \emph{discrete} \rvs is straightforward.

\paragraph{Terms.}
Let $\mathcal{V}$ be a countable set of variables. $\Bang$-terms are defined by the following grammar:
%
\[\begin{array}{llcl}
	\textsc{Terms} &	\tm,\tmu &\grameq & \val \gpipe 
    \sample\dist  \gpipe \case{\val}{\val_i\caseSym\sample {d_i}} \gpipe	
	\observe{x}{\bv}  \\[2pt]
	&&&		\letin{ x }{\tmu}{\tm}  \gpipe \letp{\pair x y}{\val}{\tm}\gpipe  \tm\val \gpipe \lam x. \tm  \gpipe \der \val  \\[2pt]
	\textsc{Values} &	\val,\valtwo & \grameq &  x\in\mathcal{V} \gpipe \pair{\val}{\valtwo} \gpipe \oc t \gpipe \bv \\[2pt]
	\textsc{Booleans} & \bv& \grameq&\true \gpipe \false \\
\end{array} \]
Following \cite{Ehrhard16,EhrhardT19}, we use Linear Logic inspired notations: $\oc t$ corresponds to $\mathtt{thunk}(t)  $ and $\der t$ to $ \mathtt{force}(t) $. 
The probabilistic primitive  $\sample d$ samples a boolean value from a (Bernoulli) distribution $d$. 
The \texttt{case} construct is just a generalized if/then/else---please notice that the \texttt{case} expression  is  restricted, because we reserve it to the encoding of  \cpts, as we have informally described in \Cref{sec:informal}. 
Observed data (the \emph{evidence}) are  specified syntactically using
an \emph{observe} construct, written $\texttt{obs}$---for example $\obs {\wet}{\true}$; we will give several examples of its use in \Cref{sec:evidence}.

\emph{Free} and \emph{bound variables} are defined as 
usual: $\la\var\tm$ binds $\var$ in $\tm$, and the same for $\leti$ and \texttt{letp}. A term is \emph{closed} when 
there are no free occurrences of variables in it.
Terms are considered modulo $\alpha$-equivalence, and capture-avoiding (meta-level) substitution of 
all the free occurrences of $\var$ for $\tmtwo$ in $\tm$ is noted 
$\tm\isub\var\tmtwo$. 
\paragraph{Syntactic Sugar.} 
The grammar of the calculus is  rather restricted,  reminiscent of  A-normal forms (and similarly to \cite{Levy99}). It is standard to recover general constructs as follows:
\[
\begin{array}{rcl}
	\tm\tmtwo & \defeq & \letin{z}{\tmtwo}{\tm z}\\
	\pair{\tmu_1}{\tmu_2} & \defeq & \letin{z_1}{\tmu_1}{\letin{z_2}{\tmu_2}{\pair{z_1}{z_2}}}\\
	\der \tmu & \defeq & \letin{z}{\tmu}{\der z}\\
	\letp{\pair {y_1} {y_2}}{\tmtwo}{\tm} & \defeq & \letin{z}{\tmtwo}{\letp{\pair  {y_1} {y_2}}{z}{\tm}}\\
	\texttt{case}\;\tmu\;\texttt{of}\; \{\val_i\caseSym\tm_i\} & \defeq & \letin{z}{\tmu}{\texttt{case}\;z\;\texttt{of}\; \{\val_i\caseSym\tm_i\}}\\
	\obsb{\tmu}	 & \defeq & \letin{z}{\tmu}{\obsb z}
\end{array}
\]
%
{\begin{notation}We often write $\pattern{v_1, \dots, v_n}$ for a $n$-tuple,  ignoring the tree order. In particular, we write $\bvs$ for tuples of booleans $\pattern{\bv_1,\ldots,\bv_n}$. 
	\end{notation}
}
\paragraph{(Call-by-Push-Value) Simple Types} In the actual technical development of this paper, we  will use intersection types. However, we prefer to first give the intuitions about typing in the  more familiar setting of simple types.
The ground types are (tensors of) booleans.
Following \citet{Levy99} and \citet{EhrhardT19}, we then define by mutual induction two kinds of types: \emph{positive} types and \emph{general} types. Only \emph{positive} types can be assigned to variables in the type environment and can appear in the left hand side of an arrow.
\[\begin{array}{lrcl}
	\textsc{Ground Types}& \sL, \sK & \grameq & \Bool \gpipe  \sL\otimes \sK
	\\[2pt]
	\textsc{Positive Types}&	 \sP,\sQ & \grameq & L \gpipe \oc \sA  \\[2pt]
	\textsc{Types}&\sA,\sB &\grameq&  \sP \gpipe \sP \larrow \sA   
\end{array}\]
%
\begin{figure}[t]
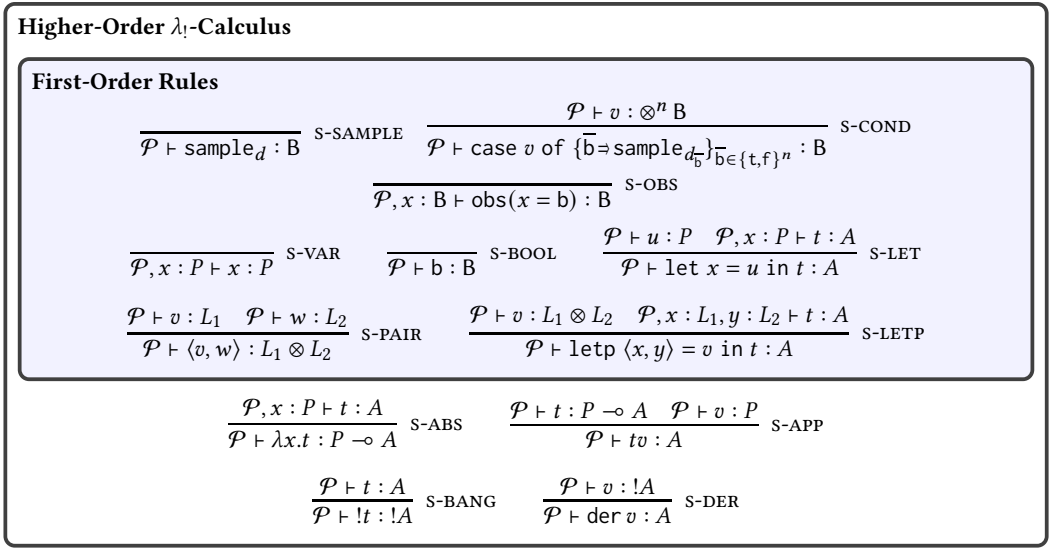

	\small
	\begin{mybox}
		\textbf{Higher-Order $\Bang$-Calculus}
		\begin{myboxC}
			\textbf{First-Order Rules} 
			\begin{gather*}
			\infer[\rsample]{ \PP \vdash \sample{\dist} : \bool }{}  
			\quad  
			\infer[\rcond]{ \PP \vdash \texttt{case}\;\val\;\texttt{of}\; \{\bvs\caseSym\sample{\dist_{\bvs}}\}_{\bvs\in\{\tmt,\tmf\}^n} : \Bool}
			{\PP \vdash \val : \otimes^n \,\Bool}\\
			\infer[\robs]{ \PP, x:\Bool  \vdash \obsb{x} : \bool }{}  \\
			\infer[\rvar]{ \PP, x : \sP \vdash x : \sP}{} \qquad \infer[\rbool]{\tjudg{\PP}{\bv}{\Bool}}{} 
			\qquad
			\infer[\rlet]{\PP \vdash \letin{x}{\tmu}{\tm} : \sA }{ \PP \vdash \tmu : \sP & \PP, x : \sP \vdash \tm : \sA }
			\\[4pt]
			\infer[\rpair]{ \PP  \vdash \pair {\val} {\valtwo} :  \sL_1 \otimes \sL_2}
			{ \PP \vdash \val :  \sL_1 & \PP \vdash \valtwo: \sL_2} 
			\qquad 
			\infer[\rletp]{\PP  \vdash \letp{\pair x y }{\val}{\tm}: \sA }
			{ \PP \vdash \val : \sL_1 \otimes \sL_2 &
				\PP, x : \sL_1, y : \sL_2  \vdash \tm : \sA }
			\end{gather*}
		\end{myboxC}
		\vskip -.4cm
		\begin{gather*}
		\infer[\rabs]{\PP \vdash \lambda x.\tm :  \sP \larrow \sA }{\PP, x : \sP \vdash \tm : \sA}
		\qquad
		\infer[\rapp]{\PP  \vdash \tm\val : \sA }{\PP \vdash \tm   : \sP \larrow \sA & \PP \vdash \val : \sP}
		\\[4pt]
		\infer[\rbang]{ \PP \vdash \oc \tm :  \oc \sA}{ \PP \vdash \tm: \sA  }
		\qquad
		\infer[\rder]{\PP \vdash \der \val :  \sA}{\PP \vdash \val : \oc \sA}
		\end{gather*}
	\end{mybox}
	\caption{The simply typed $\Bang$-calculus.}
	\label{fig:simple_types}
\end{figure}
%
The typing rules are in \Cref{fig:simple_types}, where a context $\PP$ is a sequence of assignments of positive types $\sP$ to variables $x$. As usual,  a judgment $ \PP  \vdash  t: \iA$  indicates that $\tm$ has type $\iA$ given typing context $\PP$. 
 We write   $\pi \dem  \PP  \vdash  t: \iA$   to indicate that  $\pi$ is a type derivation  of the given judgment. 
  All the rules in \Cref{fig:simple_types} are  standard but $\rcond$. 
Notice that  in rule $\rcond$ by $\{\bvs\caseSym\sample{\dist_{\bvs}}\}_{\bvs\in\{\tmt,\tmf\}^n}$ 
we mean that for each possible $n$-tuple of booleans $\bvs\in\{\tmt,\tmf\}^n$ there is a corresponding \texttt{sample} clause.
For the sake of brevity, from now on we will often shorten a $\mathtt{case}$ expression depending on $n$ variables $x_1,\dots,x_n$ as follows:
\[ \CPT{x_1,\dots,x_n} \defeq \case{\langle x_1,\dots,x_n \rangle}{\bvs \caseSym \sample{d_\bvs}}_{\bvs\in\{\tmt,\tmf\}^n} \]

\begin{remark}[Additive Contexts]
The reader familiar with Linear Logic and  calculi based on it (such as \cite{BentonBPH93})
may be surprised by the fact that here (as in \cite{Levy99, Ehrhard16}) the context is managed additively. This is because the only types which are allowed in a context are positive, hence  either of the form $\oc \sA$, or  booleans, coded by additives.\footnote{Positive types can be contracted and weakened. This is clear for types of the form $!\sA$, but holds also for ground types. Indeed a boolean type corresponds to the additive formulas $\One \oplus \One$. Notice that $\bot \& \bot = (\One \oplus \One)^\perp$ can be  weakened and contracted.} Notice that proper linear  types, such as  $\sA\larrow \sB$, are not allowed in the context (if allowed, their management would be multiplicative).


\end{remark}

\paragraph{The Higher-Order and the  Low-Level Language.}\label{sec:lowlevel}
It is standard  to encode  a ground \BN with a simple let-term of ground type.
The reader can easily realize that every ground \BN can  be described  in 
a simple, first-order fragment of the $\Bang$-calculus, as we have done in the examples in \Cref{sec:informal}.
In particular, there is no need for the modalities $\oc$ and $\mathtt{der}$, which are instead the key to implement higher-order behaviors. Abstraction and application are not necessary, as well. 
We refer to such a fragment as $\Low$-calculus. Formally, the grammar for $\Low$-terms is:
\[\begin{array}{llcl}
	\textsc{Low-level Terms}  &	\tm,\tmu &\grameq & \val \gpipe 
    \sample\dist  \gpipe \case{\val}{\bvs\caseSym\sample {d_\bvs}} \gpipe  \obsb{x}\gpipe\\[2pt]
&&&	 	\letin{ x }{\tmu}{\tm}  \gpipe	\letp{\pair x y}{\val}{\tm} \\[4pt]
	\textsc{Low-level Values} &	\val,\valtwo & \grameq &  x\in\mathcal{V} \gpipe \pair{\val}{\valtwo} 
	\gpipe \true \gpipe \false  
\end{array} \]
It is easy to check that a $\Low$-term is typable with a ground context if and only if it has ground type and is typable  with  first-order rules (those highlighted in \Cref{fig:simple_types}), only.
\condinc{}{
\begin{lemma}\label{lem:low_types}
Assume  $t$ is a \linear term and  $\LL$ a context of  ground types.
\begin{center}
	$\pi \dem \LL \vdash t: \sA \quad \Rightarrow  \quad \pi \dem_{\low} \LL \vdash t: \sA $, ~ and ~ the type  $ \sA$ is \emph{ground}.
\end{center}
\end{lemma}
\CF{NB :		if 	$ \pi \dem_{\low} \LL \vdash t: \sA $ (with $\LL$ a ground context), then   the type  $ \sA$ is \emph{ground}.\\
			non e' una tautologia perche' $\dem_{\low} x:!(A\larrow A)\vdash x:!(A\larrow A)$ (e abbiamo \rlet). 
		Ma non occorre  un enunciato formale.}
}
Going back to  our introductory intuitions, we see the $\Bang$-calculus as the target high-level language in which the statistical  model is designed by the programmer, and the $\Low$-calculus as the low-level language, closer to ground \BNs. \emph{Compiling} $\Bang$-terms into $\Low$-terms is taken care of by \emph{semantical} tools. The first of such tools is the reduction relation, which we introduce next. 

\subsection{Operational Semantics}\label{sec:operational}
In  \Cref{sec:informal} we have anticipated  that  the operational semantics of our calculus  formalizes the unrolling of a template into a ground \BN, which is its intended meaning. Formally,
every $\Bang$-term of ground type compiles (\ie, rewrites) into a $\Low$-term.

\paragraph{Root Rules.}
Since the  goal is to produce a term describing a Bayesian network, here reduction  does not fire probabilistic redexes, \ie we do not actually sample from distributions. As a consequence,  a  term  of shape $\sample{\dist}$ never reduces to a value.
This feature of our language forces us to opt for a notion of reduction, dubbed \emph{reduction  at a distance} \cite{Milner07, AccattoliK10}, which is a bit more sophisticated
than usual, and reminiscent of reduction on graphs, such as proof-nets and  bigraphs. Precisely, our  reduction is similar to that in  \cite{BucciarelliKRV20, ArrialGK23}.
%
Reduction is called \emph{at a distance} because in some of the rules the interacting parts of a redex can be separated by an arbitrarily long (possibly empty) list of \texttt{let} constructs---\ie they are \emph{distant}. Formally, we need the notion of 
\emph{substitution list}, \ie a sequence of nested \texttt{let} constructors:
\[\begin{array}{r@{\hspace{.5cm}}rlll}
	\textsc{Substitution Lists} & \sctx & \grameq &  \holebag{\cdot} \mid   \letin{\var}{\tmtwo}\sctx \gpipe \letp{\pair\var\vartwo}\val\sctx
\end{array}
\]

We are  now able to define the  rewriting rules 
which are the base  of our reduction relation. We call the term on the  left-hand side a \emph{redex}. 
\[\begin{array}{rll@{\hspace{0.5cm}}rll}
	\multicolumn{6}{c}{\textsc{Root Rules}}\\
	\holebag {\la\var\tm} \sctx\,\val & \mapsto_{\dB}   &  \holebag {\tm\isub\var\val}\sctx &
	\letin\var{ \holebag {\val}\sctx}\tm& \mapsto_{\dS}& \holebag{\tm\isub\var\val} \sctx \\[2pt]
	\der  !\tm  &\mapsto_{\dbang}& \tm &	\letp  {\pair x y}{\pair {\val}{\valtwo}}\tm  &\mapsto_{\dpair}&  
	 {  \tm \isub x {\val}  \isub y {\valtwo}}		\\
\end{array}
\]
$\holebag {\tm}\sctx$ stands for the term obtained from $\sctx$ by replacing the  hole $\holebag{\cdot}$ with $\tm$ (possibly capturing the free variables of $\tm$). The rule $ \mapsto_\dB $ fires a (possibly distant) beta-redex. The rule $\mapsto_\dS$  fires a (possibly distant) \texttt{let},
provided that its argument is a value. The rule  $\mapsto_\dbang$  defrosts a frozen term. The rule $\mapsto_\dpair $  performs pattern matching with pairs.
We set $\mapsto\,\defeq\, \mapsto_{\dB}\cup\mapsto_{\dS} \cup\mapsto_{\dbang}\cup \mapsto_{\dpair}$.



 \paragraph{Reduction.}
  A \emph{reduction step} $\to$ is the   closure  of   $\Root{}$  under evaluation context.  Reduction $\to_{\Rule}$ (for $\Rule\in \set{\dB, \sval, \dbang,\dpair}$) is defined similarly.
\emph{Evaluation contexts}, which are terms containing exactly 
one occurrence of a special symbol---the \emph{hole} $\ctxhole$---   
are defined as follows:
\[\begin{array}{lrcl}
	\textsc{Evaluation Contexts} & \ss & \grameq & \ctxhole \gpipe \ss\val \gpipe 
	\letin {x}{\ss}\tm \gpipe \letin{x}{\tmu}\ss
\end{array}\]
$\ss\hole{\tm}$ stands for the term obtained from $\ss$ by replacing the  hole $\ctxhole$ with $\tm$ (possibly capturing the free variables of $\tm$).
As it is standard  with programming languages, we adopt a weak notion of reduction, which here means that we do not reduce inside the scope of a $!$ (a thunk), nor in the scope of a $\lambda$.
Please notice that given a term of shape $\letin x u t$, reduction can be performed inside either $u$ or $t$. This is essential to make possible a reduction  such as the one in \Cref{fig:BNreduction}. As a consequence, 
reduction  is \emph{not} deterministic. However,  the choice of redex is irrelevant, in the following sense.
\begin{proposition}[Confluence]\label{diamond}\mbox{}
	\begin{enumerate}
		\item 	The reduction $\to$ is confluent. 
		\item Every normalizing term is strongly normalizing.
		\item 	All maximal   reduction sequences from a  term $t$  have the same length.
	\end{enumerate}
\end{proposition}
\begin{proof}
 Consequences of a diamond-like property,  essentially  as in \cite{BucciarelliKRV20}. 
\end{proof}

\paragraph{Progress and BN Normal Forms.} We say that a term $\tm$ is in normal form if no reduction applies ($\tm \not \to$).
It is well-known that simply typed $\lambda$-calculi are strongly normalizing. Here we prove that  every $\Bang$-term of ground type reduces to a normal form which is a  \linear term, that we dub \emph{BN normal form}.  The idea---which we will make formal  in \Cref{sec:terms_as_BN}---is that  normal forms of ground type directly correspond to ground \BNs, hence the name.

\begin{proposition}[Progress]\label{lem:progressSimple}   Let  $\tm$ be a $\Bang$-term in normal form and   $\pi \dem  \LL\vdash \tm:\sL$  a  type derivation, where all types in the  context $\LL$ are ground.  Then $\tm$ is  a   $\Low$-term,
and $\pi$ contains first order rules, only. 
\end{proposition}

\begin{proof}By induction on the structure of the derivation $\pi$. \SLV{}{(In the \Appendix).} \end{proof} 

%
%
%
%
%
\newcommand{\gnf}{\mathit{n}}
\newcommand{\glp}{s}
\newcommand{\gprob}{p}
\newcommand{\gconst}{k}
The proposition above allows us to describe the shape of BN normal forms, \ie the normal forms  of ground type. 
If we restrict our attention to closed terms, the BN normal forms are the subset of closed low-level terms generated by the following set of productions\footnote{Please notice that the converse is not true: not all terms generated by this grammar are typable.}:
	\[
		\gnf  \grameq  \letin{x}{\glp}\gnf\   \gpipe v \gpipe \glp \qquad\qquad
		\glp  \grameq  \letin{x}{\glp}\glp  \gpipe \sample{\dist} \gpipe \CPT{x_1,\dots,x_n}\  \gpipe \observe{x}{\bv}
	\]		
Please notice that BN normal forms are not values, in general. This is because we do not reduce probabilistic primitives.

\section{The Intersection Type System} \label{sec:itypesSEC}
Simple type systems 
guarantee termination, but are  poor in expressiveness. In this work we want to specify rich behaviors, such as recursion, and this is
why we switch to the \emph{untyped} $\lambda$-calculus. 
However, this is not enough to obtain the results we want, and we still do need a form of typing.
Since we are interested in giving a Bayesian network \emph{semantics} to terms, we need to \emph{keep track} of  the random variables defined by a term, and  to 
handle the fact that recursive programs generate a finite but \emph{unbounded number} of random variables.  
Random variables are associated to the sampling and conditional primitives, so  we need to take into account how many times these primitives are duplicated during the reduction. 
A quite natural option (especially when using a calculus inspired by linear logic, as we do) is to consider  a type system that explicitly takes into account how many times a sub-term is copied during the evaluation. \emph{Non-idempotent intersection types} for the \emph{untyped $\lambda$-calculus} do precisely that. This approach has several advantages:
\begin{itemize}
	\item \emph{Tracking random variables:} ground types now correspond precisely to random variables, as we discuss in \Cref{sec:named_types} below.
	\item \emph{Distinguishing copies:} non-idempotent intersection types intrinsically take into account how many times a (sub-)term is copied. This is of fundamental importance for us, since we want to keep track of all the random variables generated by our typed program. Think for example of a single (thunked) $\sample\dist$ instruction that is  duplicated several times during execution. Non-idempotent intersection type derivation duplicate \emph{in advance} each use of a sub-term, and give each a (possibly different) type. 
	\item \emph{Enforcing termination:} Bayesian networks are representations  for \emph{finite} probabilistic models. 
	Therefore  we are interested in terminating programs, only. Intersection types provide such a guarantee, still allowing for complex behaviors such as general recursion.
	\item \emph{Ruling out ill-formed terms:} as a by-product, intersection  types act as a traditional type system discarding terms which ``go wrong''. 
\end{itemize}

We highlight that we could have proceeded differently. We could have considered a \emph{typed} \texttt{PCF}-like language, as in \cite{EhrhardT19}. Still---for the reasons described above---we would have needed to add, on top of its type system, an intersection type system, like in \cite{DBLP:journals/corr/abs-1104-0193,Ehrhard16}. However, dealing with two layers of types is heavy, and we preferred keeping the syntax as simple as possible. This way, from now on we will consider just the intersection type system for the (untyped) $\Bang$-calculus which we have presented in \Cref{sec:syntax}.
 In \Cref{sec:PCF} we sketch how the (two-layered) type system for the call-by-push-value PCF would look like.

\begin{remark}[Intersection Types and Turing Completeness]\label{rem:Turing}
Type systems  ensure safety and desirable properties such as termination,  or deadlock-freeness. \emph{Intersection types} for the untyped $\lambda$-calculus \cite{CoppoDezani1978, CDV81}  bring this idea to its extreme consequence. Intersection types not only \emph{guarantee} termination, but also \emph{characterize} it, providing a \emph{compositional} presentation of \emph{all and only} the terminating programs.  They  can indeed be  seen as \emph{a semantical tool} for higher-order languages.  
Being the untyped $\lambda$-calculus Turing-complete, the price to pay is that intersection type systems are inherently undecidable. This is not typically considered an issue because such  systems have a semantic nature: they are used to give denotational models. However, please notice that \emph{the first-order fragment} of our system  \emph{is decidable}.
\end{remark}

\subsection{ Towards Bayesian Networks: Named Types  }\label{sec:named_types}

Before  presenting the intersection type system, we need to introduce one more ingredient, namely \emph{named} booleans, which we use  to track random variables. In this subsection we give some simple examples to convey the intuitions, which we then formalize in \Cref{sec:itypes}.

Describing a \BN (or a marginal distribution) by means of a term of type $\otimes^n\Bool$ is a standard, easy task. However, retrieving a Bayesian network from a term of type $\otimes^n\Bool$ is less immediate, even with \linear terms.
Does a \linear term of ground type  define a marginal distribution? If yes, over how many random variables? And what is the underlying \BN, if any?

\begin{example}\label{ex:names1} Let us consider the term $t$, where 
	$\CPT y \defeq \case{y}{\bv \caseSym \sample{d_{\bv} }}$.
	\[	t \defeq \letin x {\coin {0.2}}  
	{\letin  y x  {\letin{z}{\CPT y}{\pair x y}}  } : \Bool \otimes \Bool\]
	We know that the term $t$ defines two  \rvs (say $\X$ and $\Y$), because there are two probabilistic constructs. We also know that the output is  a probability distribution over    tuples in $\Bool \otimes \Bool$. It is however not obvious which variables are involved in the final marginal distribution. 
	We can track which random variables are involved in the term by \emph{naming} the booleans and assigning a \emph{distinct name} to the subject of each \pax (we will formalize this in \Cref{sec:itypes}).
	{\small 	\[\infer{\vdash \letin x {\sample d}  
			{\letin  y x  {\letin{z}{\CPT y}{\pair x y}}  }:{\bool_\X}\otimes {\bool_\X}}
		{ {\vdash \sample d:{\bool_\X}}   
			& \infer{   x:{\bool_\X}  \vdash   \letin  y x  {\letin{z}{\CPT y}{\pair x y}} :{\bool_\X}\otimes {\bool_\X}}
			{x:{\bool_\X}\vdash x:{\bool_\X} & \infer{{x:{\bool_\X}, y:{\bool_\X}\vdash {\letin{z}{\CPT y}{\pair x y}}}  }
				{y:\bool_\X \vdash \CPT y:\bool_\Y & \infer{x:{\bool_\X}, y:{\bool_\X}, z:\bool_\Y  \vdash \pair x y : {\bool_\X}\otimes {\bool_\X}}{ x:{\bool_\X} \vdash x:{\bool_\X}  &  y:{\bool_\X} \vdash y:{\bool_\X}} }    }}
		\]}
	We now realize that the marginal distribution defined by the term $t$ is in fact $\Pr(\X=\bv, \X=\bv)$, for  $\bv \in \set{\true, \false}$. This is a \emph{redundant} version of  $\Pr(\X=\bv)$, the  distribution over  the single variable $\X$. 
	The probabilities associated to  the tuple of values in  $\Bool_\X \otimes \Bool_\X$ are  indeed  $\pattern{\true, \true} \mapsto 0.2, \pattern{\false, \false}\mapsto 0.8 $. Necessarily, the  tuples  $\pattern{\true, \false}$ and $\pattern{\false,\true}$ have probability $0$.
	Please contrast the term $t$ above with the term $\letin x {\coin{ 0.2}}  {\letin  y {\CPT x} { \pattern{x,y}}}: \bool_\X \times \Bool_\Y$, which instead defines a distribution over \emph{two} variables.
\end{example}

In the following  we formalize these ideas, exploiting named types to associate a \BN to a ground term. 
We  write  only the  names  $\X,\Y$, but please think of them as named booleans $\Bool_\X, \Bool_\Y$.

\subsection{The Type System} \label{sec:itypes}

In this section, we introduce the type system, which will be the  base for the factor semantics in \Cref{sec:term-semantics}.
In order to focus on the main ideas, we postpone the   treatment of the evidence, namely of the construct $\observe\var{\texttt{b}}$ to \Cref{sec:evidence}. This is because while being conceptually easy, it requires some fine tuning of the type system, that would make its description harder to understand.

\paragraph{The Type System.}
The grammar of  intersection types is derived from the one for simple types, and for this reason we keep the same meta-variables. Indeed, we shall not use simple types anymore in the rest of the paper, so no confusion can occur.
We assume a countable set  $\Names=\{\X,\Y,\Z \dots\}$ of symbols, called \emph{names}, which play the role of \emph{atomic types}.
The grammar of types  includes ground  types, multisets of types (the proper intersection types), and functional types.
\[\begin{array}{lrcl}
	\textsc{Ground Types}&	 \iK, \iL & \grameq& \X\in \Names  \gpipe \iK\otimes \iL  \\[2pt]
	\textsc{Positive Types}&	 \iP,\iQ & \grameq & L \gpipe\, \mset{A_1,\ldots,A_n}  \\[2pt]
	\textsc{Types}&\iA,\iB &\grameq&  \sP \gpipe \iP \larrow \iA   
\end{array}\]
Here $[\dots]$ denotes the multiset constructor. Please notice that the empty multiset $[]$ is a positive type, as well. Two changes are present w.r.t. the definition of simple types:
\begin{enumerate}
	\item \emph{Named Booleans :} the type $\bool$ of booleans has been substituted by a countable set of names. This means that each \emph{use} of a boolean variable has now a distinct type (name). Morally, different names correspond to different random variables.
	\item \emph{Multisets for Thunks :}  types  of shape $\oc A$  have been replaced by multisets of types.
	As usual in non-idempotent intersection types, the idea is that every type inside a multiset corresponds to a single use of the typed term (see  \Cref{sec:types_use} for an example).
\end{enumerate}
The typing rules are in \Cref{fig:iTypes}; typing contexts, which we separate in 
 \emph{ground contexts} (denoted by $\Lambda$) and multiset contexts (denoted by $\Gamma$ or $\Delta$), are defined in the next paragraph.
Given a type $\iA$, we denote by $\Nm \iA$  the \emph{set} of names which appear in $\iA$.
	The definition  extends to typing contexts (\eg $ \Nm\Lambda $) and type  derivations ($\Nm{\pi}$). 
%
\paragraph{Contexts.}  
A \emph{typing context} $\Sigma$ is a (total) map from variables to
positive types such that only finitely many variables are not
mapped to the empty multiset $[]$. 
The \emph{domain} of $\Sigma$ is
the set $\dom(\Sigma) \defeq \{x\mid \Sigma(x)\neq []\}$. 
A context $\Sigma$ is \emph{empty} if $\dom(\Sigma)=\emptyset$.
A typing context $ \Sigma $ is denoted by $x_1 :\iP_1, \ldots , x_n
:\iP_n$ if $\dom(\Sigma) \subseteq \{x_1, \ldots, x_n\}$ and $
\Sigma(x_i) = \iP_i $ for all $1\leq i \leq n$.
Given two typing
contexts $ \Sigma_1$ and $ \Sigma_2$ such that $ \dom(\Sigma_1) \cap
\dom(\Sigma_2) = \emptyset $, the typing context $\Sigma_1,\Sigma_2$ is
defined as $(\Sigma_1,\Sigma_2)(x) \defeq  \Sigma_1(x)$ if $x \in
\dom(\Sigma_1)$, $(\Sigma_1,\Sigma_2)(x) \defeq  \Sigma_2(x)$ if $x \in
\dom(\Sigma_2)$, and $ (\Sigma_1,\Sigma_2)(x) \defeq  []$ otherwise.
Observe that $\Sigma, x :[]$ is equal to $\Sigma$.
Given a context $\Sigma$, it is convenient to partition it into a ground and a \exponential context.  We call \emph{ground context}  (denoted by  $\Lambda$)  the restriction of $\Sigma$ to the variables which are mapped to ground types, and we call \emph{\exponential context} (denoted by  $\Gamma$ or $\Delta$) its complement, \ie the restriction of $\Sigma$ to the variables which are mapped to multisets.
Multiset union $\uplus$ is extended to \exponential contexts point-wise, i.e.  $(\Gamma \uplus
\Delta)(x) \defeq  \Gamma(x) \uplus \Delta(x)$, for each variable $x$.

\begin{remark}
	Please notice that in \Cref{fig:iTypes} we have operated a few simplifications in the presentation of the type system. In particular, we do not type boolean constants anymore---this simplifies the $\icond$ rule, that now is an axiom. Moreover,  notice that tuples always have the multiset context empty.
	As we aforementioned, we postpone the   treatment of  $\observe\var{\texttt{b}}$ to \Cref{sec:evidence}. 
\end{remark}

\paragraph{Type Derivations.} We write   $\pi \dem  \Lambda, \Gamma  \vdash  t: \iA$   to indicate that  $\pi$ is a type derivation (using the full type system  in \Cref{fig:iTypes}), proving  that $\tm$ has type $\iA$ given typing context $\Lambda$ (ground) and $\Gamma$ (multiset). 
We write   $\pi  \dem_\low \Lambda  \vdash  t: \iL$ for a   derivation $\pi$  which  uses first order rules, restricted to  ground contexts, only.

\begin{figure}[t]
	\small
	\begin{mybox}		
		\textbf{Higher-Order Calculus}
		\begin{myboxC}
			\textbf{First-Order Rules}
			\begin{gather*}
				\infer[\isample]{ \Lambda \vdash {\sample{d}} : \X}{ \X \not \in \Nm{\Lambda} } 
				\qquad
				\infer[\icond]{ \Lambda, y_1 : \Y_1, \dots, y_n : \Y_n \vdash \CPT{y_1,\dots, y_n} : \X }
				{ \X \not \in \{\Y_1,\dots,\Y_n\} \tand \X \not \in \Nm{\Lambda} } 
				\\[4pt]	
				\infer[\ivar]{ \Lambda, x : \iP \vdash x : \iP}{}   	
				\qquad
				\infer[\ilet]{ \Lambda, \Gamma_1\uplus\Gamma_2  \vdash \letin{x}{\tmu}{\tm} : \iA }
				{ \Lambda, \Gamma_1 \vdash \tmu : \iP  &  \Lambda, \Gamma_2, x : \iP \vdash \tm : \iA }
				\\[4pt]
				\infer[\ipair]{ \Lambda \vdash \pair{v}{w} :  \iL_1 \otimes \iL_2 }
				{ \Lambda \vdash v : \iL_1  &  \Lambda  \vdash w :\iL_2 } 
				\qquad
				\infer[\iletp]{ \Lambda,   \Gamma  \vdash \letp{\pair x y }{\val}{t} :  \iA }
				{ \Lambda \vdash \val : \iL_1 \otimes \iL_2 
					& \Lambda, x : \iL_1, y : \iL_2, \Gamma \vdash t : \iA 
				}
			\end{gather*}
		\end{myboxC}
		\vskip -.4cm
		\begin{gather*}
			\infer[\iabs]{ \Lambda, \Gamma \vdash  \lam x.t : \iP \larrow \iA }{ \Lambda, \Gamma, x : \iP \vdash t :  \iA}
			\qquad
			\infer[\iapp]{ \Lambda, \Gamma_1\uplus\Gamma_2 \vdash tv : \iA }
			{ \Lambda, \Gamma_1 \vdash t : \iP \larrow \iA  &  \Lambda, \Gamma_2 \vdash v : \iP}
			\\[4pt]
			\infer[\ibang]{  \Lambda,   \biguplus_i\Gamma_i \vdash \oc \tm : \mset{\iA_1,\dots,\iA_n} }{ \big( \Lambda, \Gamma_i \vdash \tm: \iA_i \big)_{i = 1}^n  }
			\qquad
			\infer[\ider]{ \Lambda,   \Gamma \vdash \der \val :  \iA}
			{ \Lambda, \Gamma \vdash v : \mset \iA}
		\end{gather*}
	\end{mybox}
	\vspace{-8pt}
	\caption{The intersection type system $\iBang$. }
	\vspace{-8pt}
	\label{fig:iTypes}
\end{figure}

\paragraph{Naming Condition.}  We call \emph{main names} those which type  the subject of an
\icond or \icoin rule.  Given a type derivation $\pi$, we assume  that all the  main  names  are \emph{pairwise distinct}. So, each \icond or \icoin rule is uniquely identified by a name $\X$. This requirement is easy to implement. Indeed, it is a sort of Barendregt convention, but for types. One could consider the name introduced by a probabilistic axiom as the address of the axiom in the  type derivation.

\paragraph{The First-Order Fragment.} 
Notice that the first-order fragment   in \Cref{fig:iTypes} is the same as the first-order fragment of   simple types (\Cref{fig:simple_types}), the only difference being that now  the booleans are named.  Clearly this fragment is \emph{decidable}, since any first-order simply typed derivation 
of $\LL \vdash t: \otimes^n\Bool$ (where every type in the context  $\LL$ is ground) can easily be 
named, by assigning a distinct name to the subject of every \paxs, and to every occurrence of boolean type in $\LL$. 


\subsection{Properties of the Type System}\label{sec:TS_properties} 
The intersection type system satisfies all the  properties one would expect---proofs are in the \Appendix. First, types are stable under reduction and expansion (the latter property  not holding for simple types).
\begin{proposition}[Subject Reduction/Expansion]\label{prop:subsconv}
Let $\tm$ be a $\Bang$-term such that $\tm\to\tmtwo$. Then $\tjudg{\Sigma}{\tm}{\iA}$ if and only if $~\tjudg{\Sigma}{\tmtwo}{\iA}$.
\end{proposition}
Subject reduction can be strengthened, showing that there is a measure that decreases along each reduction sequence. This  gives a combinatorial proof  that typable terms are strongly normalizing.
\begin{theorem}\label{thm:S-soundness} Let $\tm$ be a $\Bang$-term. If $\tm$ is typable, then $\tm$ is strongly normalizing.
\end{theorem}

Crucially, the progress lemma, stated for simple types, still holds for intersection types. This means that every (higher-order, possibly recursive) $\Bang$-term which is typable with ground type in a ground context, eventually reduces to a \RED{BN} normal form.
\begin{theorem}[Compiling into the low-level]\label{thm:compileI}   Let   $\tm$ be a $\Bang$-term such that 
$\pi \dem \Lambda\vdash \tm:\iL$. Then $\tm\to^*\tmu$, where $\tmu$ is  a   $\Low$-term in  normal form  (a BN  normal form).
\end{theorem}
Notice that the  normal form $\tmu$ has necessarily  a type derivation $\pi' \dem_\low \Lambda\vdash \tmtwo:\iL$ (using  first-order rules only).
Finally, we highlight that the type system is syntax driven. As a consequence: 
	\begin{proposition}[Unique Derivation]\label{prop:uniqueness} Let  $\tm$ be a $\Bang$-term, and $\Lambda$  a ground context. Then there exists at most one type derivation $\pi$ such that
			$\pi \dem \tjudg{\Lambda}{\tm}{\iL}$.
	\end{proposition}
This means that for each $\Bang$-term $\tm$ and ground context $\Lambda$, there exists at most one ground type $\iL$ such that  $\tjudg{\Lambda}{\tm}{\iL}$ admits a type derivation. In particular, the type derivation for a term in BN normal form is uniquely determined. 
By subject expansion 
we have:
	\begin{theorem}\label{cor:unique}
	 Let $\tm$ be a closed $\Bang$-term $\tm$. If $\tm$ reduces to BN normal form, then a type derivation $\pi\dem\, \vdash \tm:\iL$ \emph{exists} and is \emph{unique}.
	\end{theorem}

	\subsection{Putting the Calculus and the Type System  at Work}\label{sec:types_use}
	 We illustrate with some examples   the expressiveness of the calculus, and the use of the type system.
	\paragraph{Expressiveness.}
	Since  we have access to the full  expressiveness of the untyped $\lambda$-calculus (\Cref{rem:Turing}),  we  can use  a  standard  encoding (in its call-by-push-value flavor)  of integers, arithmetic, if/then/else, and fixed point combinators.
	\begin{example}[Encoding Recursive Behavior] 
		We are now able to  encode   Dynamic Bayesian networks, such us the one depicted in \Cref{fig:HMM} (from \cite{KollerBook}, Ch.6). 
		The idea behind this model is that a system evolves with time in a stochastic way. At each time step, one random variable $\mathtt{S}_{i+1}$, which depends only on the previous state $\mathtt{S}_{i}$, represents the new state, while the observation $\mathtt{O}_{i+1}$ depends only on the current state $\mathtt{S}_{i+1}$. A typical query in these kinds of models is
		\[
		\Pr(\mathtt{S}_n=\mathtt{s}\mid \mathtt{O}_1=\mathtt{o}_1,\ldots,\mathtt{O}_n=\mathtt{o}_n)
		\]
		which intuitively means: after $n$ time steps, what is the probability of being in a certain state, knowing all the observations?
		We can write the template of \Cref{fig:HMM}  as follows:
		
		\[
		\arraycolsep=1.6pt
		\begin{array}{ll}
			\tm &\defeq \la{n}\letin{s_0}{\coin{p}}{\tmtwo \; n \; {s_0}} \\
			\tmtwo &\defeq \texttt{fix } \oc( \, \lambda x. \lambda n. \lambda s. \texttt{if} \; \isZ{n} 
			\begin{array}[t]{lll}
				\texttt{then} & s  \\
				\texttt{else} &\leti s' = \CPTN{\mathtt{S}}{s} \, \texttt{in} \\
				&\leti o' = \CPTN{\mathtt{O}}{s'} \, \texttt{in} \\
				&\leti m = \pred{n} \, \texttt{in} \\
				&\leti r = (\der x) \; m \; s' \\
				&\texttt{in } \pair{o'}{r} \, )
			\end{array}
		\end{array}
		\]
		
		Then $\tm\underline{n}$---where $\underline{n}$ is an  encoding of the integer $n$---represents the template that is to be  unrolled $n$ times. 
		Operationally, we have exactly that the fixed point operator is unfolded $n$ times generating the unrolled Bayesian network. From the point of view of the type system, we have that the type of $\tm\underline{n}$ \emph{depends} on the integer $n$:
		\[
		\tjudg{}{\tm\underline{n}}{\mathtt{O}_1\otimes\cdots\otimes\mathtt{O}_n\otimes\mathtt{S}_n}
		\]
	\end{example}
	
	
	
\paragraph{The Intersection Types, in use.} The term which encodes  the BN of \Cref{fig:BNrain} is easily typed in the first-order fragment---the reader can find  the  derivation in  \Cref{ex:flow}. 
Here we give an example of type derivation where   \emph{multiple copies} are involved,  so that we  need to use  the \emph{multiset type}.
First, observe that a  type of shape $[\iA_1, \dots, \iA_n]$ can be thought of as an informative refinement of the thunk type $\oc \iA$. Sub-terms are typed several times, once for each copy that will be produced during the reduction. 
This feature is crucial to handle  random variables. 
We stress that   the inability to deal with an unbounded number of random variables is 
the key issue which limits to first-order 
the compilation of programs into BNs, or  data flow analysis (see \eg \cite{abs-1809-10756,GorinovaGSV22}).
\begin{example}[Multiple Coin Tosses] \label{ex_coins}Consider  a   term $\tmu$ similar to that in \Cref{fig:BNreduction},  modeling two tosses of the same (biased) coin:
	$ 	\tmu ~\defeq ~ \letin{x}{\sample{d}}{\letin{y}{\oc(\CPTN{}{x})}{\pair{x}{\der y, \der y}}}$.  For readability, here we use some syntactic sugar.
	The (unique) type derivation for $\tmu$ is
	\vspace*{-4pt}
	\begin{prooftree}\footnotesize
		\def\extraVskip{3pt}
		\def\ScoreOverhang{1pt}
		\AxiomC{$ \vdash \sample{d} : \nX $}
			\AxiomC{$x : \nX  \vdash \CPTN{}{x} : \nY_1 $}
				\AxiomC{$x : \nX  \vdash \CPTN{}{x} : \nY_2 $}
			\insertBetweenHyps{\hskip .2cm}
			\BinaryInfC{$x : \nX \vdash \oc(\CPTN{}{x}) : [\nY_1,\nY_2] $}
			\AxiomC{$x : \nX \vdash x : \nX$}
				\AxiomC{$x : \nX, y : [\nY_1]  \vdash y : [\nY_1] $}
				\UnaryInfC{$x : \nX, y : [\nY_1]  \vdash \der y : \nY_1$}
					\AxiomC{$ x : \nX, y : [\nY_2] \vdash y : [\nY_2] $}
					\UnaryInfC{$x : \nX, y : [\nY_2] \vdash \der y : \nY_2 $}
				\insertBetweenHyps{\hskip .3cm}
				\BinaryInfC{$x : \nX, y : [\nY_1,\nY_2] \vdash \pair{\der y}{\der y} : \nY_1 \otimes \nY_2 $}
				\insertBetweenHyps{\hskip .1cm}
			\BinaryInfC{$x : \nX, y : [\nY_1,\nY_2] \vdash \pair{x}{\der y, \der y} : \nX \otimes \nY_1 \otimes \nY_2$}
		\BinaryInfC{$x : \nX \vdash \letin{y}{\oc(\CPTN{}{x})}{\pair{x}{\der y, \der y}} : \nX \otimes \nY_1 \otimes \nY_2$}
		\insertBetweenHyps{\hskip -.4cm}
		\BinaryInfC{$ \pi \dem \vdash  \tmu : \nX \otimes \nY_1 \otimes \nY_2 $}
	\end{prooftree}
\end{example}


\renewcommand{\name}{\rv}
\renewcommand{\names}{\rvs}

\newcommand{\Pa}[1]{\mathsf{Pa}(#1)}
\newcommand{\aphi}{\mathcal{T}}

\section{The Semantics of  Bayesian Networks}
In this section we formally define  the semantics of \BNs.
First, let us  briefly revise the language of Bayesian modeling. For more details, we refer to \cite{DarwicheHandbook} for a concise presentation, and to standard texts for an exhaustive treatment~\cite{Pearl88,DarwicheBook,NeapolitanBook}.
\condinc{}{Following  \cite{DarwicheBook}, in this paper we are only concerned with  random variables   whose set of values is finite, so \emph{discrete random variables}.}

\subsection{Random Variables}\label{sec:rvs}
Bayesian methods provide a formalism for reasoning about partial beliefs under
conditions of uncertainty. Since we   cannot determine for
certain  the state of some features of interest,  we settle for determining how \emph{likely} it is that a particular feature is in a particular state. 
Random variables represent features  of the system being modeled.   
For the purpose of 
modeling,  a random variable  can be seen as a \emph{name} for an atomic proposition (\eg "Wet") which assumes values from a set of states (\eg $\{\true,\false\}$).  The system  is  modeled as \emph{a joint} probability distribution on all possible values of the variables of interest -- an element in the sample space represents a possible state of the system.

\begin{example} The  canonical sample space sketched in \Cref{fig:marginalization}  consists of $2^4$ tuples (only some entries are displayed); to each tuple  $\bx$ 
	is associated a probability.
The event $(R=\true) $  contains $2^3$ tuples,  the event $(\nR=\true,\nW=\true)$ contains $2^2$ tuples, and has probability $0.33$.
\end{example}

\begin{remark}
Notice that in Bayesian modeling, random variables are identified first, and only implicitly become functions on a sample space. 
We refer to the excellent textbook by Neapolitan \cite{NeapolitanBook} (Ch.~1) for a formal treatment relating the  notion of random variable as used in Bayesian inference, with the classical definition of function on a sample space.  

\condinc{}{Given a  set of  variables $\bX$,
it is convenient and standard  
to assume as \emph{canonical sample space}  the set $\Omega=\Val \bX$.
We can then see each  variable $X_i$ as a function from $\Omega$ to $\Val {X_i}$, projecting $ (x_1, \dots, x_n) $ into $x_i$.
Hence, a random variable in the classical sense.
}

\end{remark}

 Given a countable set $\Names$, we associate to each name $\X\in\Names$ a set of values, denoted by  $\Val \X$ (typically $\Val{\X}=\{\true, \false\}$). From now on, we silently identify a name  $\X$ with the pair $(\X, \Val \X)$, which effectively defines a  \emph{random variable} (\rv).
A  finite set of names $\bX=\{\X_1, \dots, \X_n\}$ defines a "compound" \rv whose value set $\Val \bX$   is the Cartesian  product 
$  \Val {\X_1} \times \dots \times \Val {\X_n} $.

\begin{notation}\label{notation_BN}
The metavariables $\bX,\bY,\bZ$  range over finite \emph{sets of names} (\rvs). 
As standard,  a lowercase letter   $\x$  denotes a generic value  $  \x\in \Val \nX $, and $\xs$  denotes a tuple in the cartesian product 
$\Val{\bX}\defeq  \Val {\nX_1} \times \dots \times \Val {\nX_n} $. Moreover, we use juxtaposition as a tuple constructor, \eg if $\x\in\Val{\X}$ and $\ys\in\Val{\Y_1}\times\cdots\times\Val{\Y_n}$, then $\x\ys\in\Val{\X}\times\Val{\Y_1}\times\cdots\times\Val{\Y_n}$.
Given a subset $\bY\subseteq\bX$, we denote by $\Proj{\bx}{\bY}$ the restriction of $\bx$ to ${\bY}$ (so, $\Proj{\bx}{\bY} \in \Val \bY$ ). 
Given two sets of names $\bX$ and $\bY$, we say that $\bx\in\Val{\bX}$ and $\by\in\Val{\bY}$ \emph{agree on the common names} ($\bx\sim \by$ for short) whenever $\Proj{\bx}{\bX\cap\bY}=\Proj{\by}{\bX\cap\bY}$.

\end{notation}

\subsection{Factors, Sum and Product Operations}
Inference algorithms  rely on basic operations on a class of functions 
known as \emph{factors}, which generalize    the notions of  probability distribution and of conditional distribution. Factors will be the key ingredients also in our semantics.

\begin{definition}[Factor]
	A \emph{factor} $\ft\phi{\bX}$  over a set of  names (\rvs) $\bX$
	is a function $\ft\phi{\bX}:\Val{\bX}\to\mathbb{R}_{\geq 0}$
	mapping each  tuple  $\bx\in \bX$  to a non-negative real.
\end{definition}
When $\bX$ is clear from the context, we simply write  $\phi$ (omitting the superscript $\bX$); then $\Nm{\phi}$  denotes $\bX$. Letters $\phi,\psi$ range over factors. 
Please notice that in the literature about BNs, $\phi(\bx)$ is often written $\phi_\bx$. We adopt this  convenient notation when making explicit calculations.

\begin{example} Factors generalize familiar concepts from probability theory.
\begin{itemize}
	\item A \emph{joint probability distribution} over the set $\bX$ is a factor $\phi$ which maps each tuple  $\xs\in \Val \bX$ to
	a probability $\phi(\xs)$ such that $\sum_{\xs\in \Val{\bX}} \phi(\xs) =1$.
	\item A \emph{conditional probability table} (\cpt)  for $\X$ given $\bY$  is a factor $\ft\phi {\{\X\}\cup\bY}$ which maps each tuple $\x \ys\in\Val{\{\X\}\cup\bY}$ to a probability $\phi(\x \ys)$ such that for each $\ys\in\bY$, $\sum_{\x\in \Val{\X}} \phi(\x \ys) =1$.\footnote{Please notice that here, and in the following definition of sum out, we slightly abuse the notation. In fact, as standard, we consider the tuples as sequences indexed by the set of random variables. 
	Every time, we present the tuples ordered in the most convenient way for a compact  definition.}
\end{itemize}
\end{example}
Factors come with two important operations: sum (out) and product. 
%
Summing out a name (\rv)  $\Z$ from a factor  means that we are removing  $\Z$,
 thus obtaining a  smaller  factor. As depicted in \Cref{fig:marginalization}, intuitively we do so  by merging all tuples  which agree on all the other variables but $\Z$.
\begin{definition}[Sum Out]
	The \emph{sum out} of $\bZ\subseteq\bX$ from  $\ft\phi{\bX}$ is a factor $\sum_\bZ \phi$ over  $ \bY \defeq \bX  - \bZ $, defined as: 
	\[\left(\sum_\bZ \phi\right)(\by)\defeq\sum\limits_{\bz\in \Val{\bZ}}\phi(\bz\, \by)\]
\end{definition}
Multiplication of factors is defined in such a way  that only ``coherent`` pairs are  multiplied.
\begin{definition}[Product]
	The \emph{product} of $\ft{\phi_1}{\bX}$ and $\ft{\phi_2}{\bY}$  is a factor $\phi_1\FProd \phi_2$ over $\bZ \defeq \bX\cup \bY$, defined as:
	\[(\phi_1\FProd \phi_2)(\bz) \defeq \phi_1(\bx) \phi_2(\by) \quad  \mbox{where $\bx= \Proj{\bz}{\bX}$ and $\by= \Proj{\bz}{\bY}$.}\]
\condinc{}{\RED{where $\bx\sim \bz$ and $\by\sim \bz$, \ie $\bx$ and $\by$ 
agree with $\bz$ on  the common variables (see  \Cref{notation_BN}).}}
\end{definition}
We denote $n$-ary products by $\BigFProd_{n} \phi_n $. We denote by $\ftone_\bY\defeq\ft\ftone\bY$ the factor over the set of names $\bY$, sending every tuple of $\Val{\bY}$ to $1$. Observe that $\ft{\phi}{\bX} \odot \ft{\ftone}{\bY} = \ft{\phi}{\bX}$ if $\bY \subseteq \bX$.
Factors over an empty set of variables are allowed, and called trivial. In particular, we write $ \ftone_\emptyset\defeq\ft \ftone \emptyset$ for 
the trivial factor assigning $1$ to the empty tuple.
%
Product and summation  are both commutative, product is associative, and---crucially---they distribute \emph{under suitable conditions}: 
\begin{equation} \label{eq:sum_prod}
	\text{If }\bZ\cap \Nm{\phi_1}=\emptyset \text{ then } 
	\sum_\bZ ( \phi_1 \FProd \phi_2)  = \phi_1 \FProd \left( \sum_\bZ  \phi_2 \right)  
\end{equation}
This distributivity is  the key property  on which  exact inference algorithms rely.
CPT's being factors, they admit sum and product operations. Please notice that the result of such operations is not necessarily a \cpt, but, of course, it is a factor. 
\begin{remark}[Cost of  Operations on Factors]\label{rem:factors_cost}
	Summing out any number of variables from a factor $\phi$ demands $\BigO{\Exp{w}}$ time and space, where $w$ is the number of variables over which $\phi$ is defined. 
Multiplying $k$ factors demands 
$\BigO{k\cdot \Exp{w}}$
 time and space, where $w$ is the number of variables in the resulting factor.
\end{remark}
\subsection{The Semantics of Bayesian Networks}\label{sec:BN}

Bayesian networks are graph-theoretic objects able to represent large probability distributions compactly, via 
 a \emph{factorized representation}. Inference algorithms then implement    \emph{ factorized  computations.}

\begin{definition}
	A \emph{\BN} $\bn$ over the set of \rvs $\bX$ is a pair $(\mathcal G,\aphi)$ where:
	\begin{itemize}
		\item $\mathcal G$ is a directed acyclic graph (DAG) over the set of nodes $\bX$. 
		\item 	$\aphi $   assigns, to each variable $\X\in \bX$ a \emph{conditional probability table (a \cpt)}, which is a  factor $ \phi^\X$ over variables $\{\X\}\cup \Pa \X$, where $\Pa \X$  denotes the set of  parents of $\X$ in $\mathcal G$. 
	\end{itemize}
\end{definition}
The graph structure and  independence assumptions on it yield the correctness of  the semantics.
\begin{theorem}[\citet{Pearl86}]\label{thm:pearl}
	A \BN $\bn$ over the set of \rvs $\bX$ defines a unique probability distribution over $\bX$ (its \emph{semantics}):
	\[\sem{\bn} \defeq \BigFProd_{\X\in \bX} \phi^\X.\]
\end{theorem}
Please notice also that, given a \BN $\bn$ defining a probability distribution over $\bX$, 
 the  \emph{marginal distribution} of $\sem \bn$ over a subset $\bY \subseteq \bX$ is defined
by $\sum_{\bX-\bY} \sem \bn$.


\newcommand{\compose}{\mathit{Compose}}

\newcommand{\tcoin}{\texttt{coin}}
\newcommand{\tcond}{\texttt{cond}}
\newcommand{\tvar}{\texttt{var}}
\newcommand{\tlet}{\texttt{let}}
\newcommand{\tpair}{\texttt{pair}}

\newcommand{\iMin}{\texttt{iMin}}

\section{The Semantics of Typed Terms} \label{sec:term-semantics}


This section contains the main result of this paper, namely the fact that we can endow terms of ground type with a factor-based semantics that reflects the probabilistic behavior of the term. Moreover, we prove that the semantics  is compositional, in the sense that it can be computed following the structure of intersection type derivations.

\paragraph{Semantics of the Probabilistic Axioms.}\label{sec:axioms}
The intuition guiding the definition of our semantics is the very definition of the semantics of a BN, as we have just seen in  \Cref{sec:BN}. In general, this depends only on the \cpts which are assigned to each random variable. We apply the same principle in the realm of (intersection) typed $\Bang$-terms. The idea is that we can associate a CPT to each probabilistic axiom in a type derivation, and then we define the semantics of the typed term as their product (as factors), summing out the names not occurring in the final type judgment to obtain the marginal which is specified by the term.  
We start by formally defining the factor which is associated to a probabilistic axiom.  	We give an \emph{example} of how  the following definitions  work in \Cref{fig:paxs}.
\begin{definition}[Probabilistic Axioms] \label{def:CPTpax}
	Recall that   we  identify each name $\X$ with the pair   $(\X, \Val{\X})$,  effectively defining a \rv (see \Cref{sec:rvs}).
\begin{itemize}
	\item $\isample$.
	We associate to the axiom $ \Lambda \vdash {\sample{d}} : \nX $ the  factor $\phi$ over the \name  $\set{\X}$ such that 
	$\phi (\x)\defeq\dist(\x)$  for each $\x\in \Val{\X}$.

	\item $\icond$. Let us consider the following instance of the $\icond$ axiom.
	\[\infer[\icond]{ \Lambda, y_1:\Y_1, \dots y_n:\Y_n \vdash   \case {\pattern{y_1, \dots, y_n}} {\bvs\Rightarrow\sample{\dist_{\bvs}}}_{\bvs\in\{\true,\false\}^n}:\X }
	{   \quad&    \X\not\in  \{\Y_1, \dots, \Y_n\} \tand  \X\not\in \Nm{\Lambda}}\]
	We associate to this axiom the factor $\phi$ over the set of \rvs $\{\Y_1, \dots, \Y_n,\X\}$, such that $\phi(\bvs\x)\defeq \dist_{\bvs}(\x)$,
	for each $\bvs\x\in \Val{\Y_1} \times \dots \times  \Val{\Y_n} \times \Val{\X}$\footnote{Please notice that we are a bit informal here, because we  assume that $\Y_1, \dots, \Y_n$ are pairwise distinct. This is not an obligation, so the actual definition is more involved. This is, however, just a technical point that does not affect the meaning of the definition. For the sake of completeness, we provide the technically precise definition in the \Appendix.}.
	
\end{itemize}
\end{definition}
%
{\small 
	\begin{figure}[b]\[\begin{array}{c|c}
		\begin{array}{c}
		 \infer[\isample]{\Lambda\vdash   \coin {0.2}:\R} {}\\[4pt]
		\phi_1\defeq
		\begin{array}{|c|c|}
			\hline	
			\R  & \Pr (\R) \\
			\hline	
			\true	& 0.2 \\
			\false	& 0.8 \\
			\hline
		\end{array}\end{array} &
		%
		\begin{array}{c}
	\infer[\icond]{\Lambda, 	r: \R \vdash \texttt{case}\; r
	\;\texttt{of}\;\{\tmt\caseSym\coin{0.7}; 	\tmf\caseSym	\coin{0.01}\}:\W}{} \\[4pt]
		\phi_2 \defeq 
		\begin{array}{|cc|c|}
			\hline	
			\R  & \W  & \Pr(\W | \R) \\
			\hline	
			\true  & \true	& 0.7 \\
			\true  & \false	& 0.3 \\
			\hline
			\false & \true	& 0.01 \\
			\false & \false	& 0.99 \\
			\hline
		\end{array}\end{array}
	\end{array}\]
	\caption{The factor associated to this  $\isample$ axiom is $\phi_1$.  
		  The factor associated to this   $\icond$  axiom is $\phi_2$. 
		  For example,  $\phi_2(\true\false)\defeq \coin{0.7}(\false)$.}
	\label{fig:paxs}
\end{figure}}

\paragraph{Semantics of a Type Derivation.}\label{sec:derivations}
Once given  the interpretation of  the \paxs, it is straightforward to extend the interpretation to  any type derivation of ground type. We multiply all the factors associated to the axioms, and then sum out  all the names not appearing in the conclusion.
\begin{definition}[Semantics of a Type Derivation]\label{def:sem}\label{def:core}
	Let $\pi\dem J$ be a type derivation,  $\Cpts{\pi}$ the set of factors associated to its probabilistic axioms, and $\bX=\bigcup_{\phi\in \Cpts{\pi}} \Nm{\phi}$. Then the semantics of $\pi$ is:
	\[
	\sem{\pi} \defeq \sum_{\bX-\Nm{J}} \left(\BigFProd_{\phi\in \Cpts{\pi}} \phi \right)
\]
\end{definition}
\begin{remark}\label{ex:var_ax}
	 If $\pi$ is a type derivation such that $\Cpts{\pi}=\emptyset$, then $\sem{\pi}={\ftone}_{\emptyset}$. Notice, in particular, that this is the case for  each \ivar axiom. 
\end{remark}

\paragraph{On Compositionality.} One immediately notices that the above  definition of semantics  does  \emph{not} look compositional. Indeed, the semantics of a type derivation $\pi$ is computed looking \emph{globally} at $\pi$,  in particular at its \paxs. We ask ourselves, and we answer in the positive, if it is possible to give a more local,  modular way, of computing the very same semantics. Informally, given a type derivation $\pi$ obtained from the composition of $\pi_1,\ldots,\pi_n$, such as
\[ \pi\quad\defeq\quad\infer {J }{ \pi_1 & \dots &\pi_n } \]
we would like to obtain $\sem{\pi}$ from $\sem{\pi_1},\ldots,\sem{\pi_n}$. In particular, following the pattern of Definition~\ref{def:sem}, we could write:
\begin{equation}\label{eq:compositionality}
	\sem{\pi} = \sum_{\bZ}  \left( {\BigFProd_i \sem{\pi_i}}\right)\quad
	\text{ where } \bZ= \bigcup_i \Nm{\sem{\pi_i}} - \Nm{J}
\end{equation}
In words, composition  is obtained by \emph{first} performing the product  $ \BigFProd_i\sem {\pi_i} $---which yields a factor  over the names $\bigcup_i \Nm{\sem {\pi_i}}$--- and \emph{then} marginalizing, by summing out the names which do not appear in the conclusion $J$. Such a notion of composition is yet a variant of the pervasive paradigm
\begin{center}
	\emph{	composition = parallel composition + hiding.}
\end{center}
The problem now is that the equation~\ref{eq:compositionality} is not a priori true. This is because sum out and product do not distribute in general, but only under suitable conditions. The type system design is crucial to guarantee that the factors semantics is indeed compositional.  We illustrate this fact with an example.
\begin{example}[Types and Compositionality.] \label{ex:non-comp} 
	Consider the following  derivation, which is \emph{non well-typed}
	because the names introduced by the first and  third \pax are  not distinct. 
	Below,  $\cc^1,\cc^3$ are \texttt{sample} terms, and $\CPTN2{z}$, $ \CPTN4{z'} $ 
	are \texttt{case} expressions. The terms 
	$ t,  u$ are  those typed by  $\pi_1,   \rho_1 $, respectively.
	The premise $ \rho_2$ is the obvious derivation of 
	$x : \X, y : Y \vdash \pair{x}{y} : \X \otimes \Y $. 
\begin{center}
	\small	
	\begin{prooftree}
	\AxiomC{$ \vdash \cc^1{} : \Z \dmd{ \ft{\phi^1}{\{ \Z\}} }$}
		\AxiomC{$  z : \Z \vdash \CPTN2{z} :  \X \dmd{ \ft{\phi^2}{\{\Z, \X\}} }$}
	\insertBetweenHyps{\hskip .2cm}
	\BinaryInfC{$\pi_1 \dem  \letin{z}{\cc^1{}}{\CPTN2{z}} : \X $}
		\AxiomC{$ x : \X \vdash \cc^3{} : \Z \dmd{ \ft{\phi^3}{\{\Z\}} }$}
			\AxiomC{$ x : \X, z': \Z \vdash \CPTN4{z'} :  \Y \dmd{ \ft{\phi^4}{\{\Z, \Y\}} }$}
		\insertBetweenHyps{\hskip .2cm}
		\BinaryInfC{$  \rho_1 \dem  x : \X \vdash \letin{z'}{\cc^3{}}{\CPTN4{z'}}  : \Y $}
			\AxiomC{$ \rho_2 $}
		\insertBetweenHyps{\hskip .2cm}
		\BinaryInfC{$ \pi_2 \dem x : \X \vdash \letin{y}{ u}{\pair{x}{y}} : \X \otimes \Y$}
	\BinaryInfC{$\pi \dem \vdash \letin{x}{ t}{\letin{y}{ u}{\pair{x}{y}}} : \X \otimes \Y $}
	\end{prooftree}
\end{center}
	We annotate each of the four \paxs with the corresponding \cpt.
	It is easy to check that compositionality (\Cref{eq:compositionality}) does not hold, because
	$\sem{\pi} = \sum_{\Z}(\phi^1 \FProd \phi^2 \FProd \phi^3\Fprod\phi^4) ~\not=~
	(\sum_{\Z} \phi^1 \Fprod\phi^2)\Fprod (\sum_{\Z} \ft{\phi^3\FProd \phi^4}{}) =  \sem {\pi_1}\FProd \sem{\pi_2}.$ 
	 Observe  that $\sem{\pi}$ is a factor over $\set{\X,\Y}$,  $\sem{\pi_1}$ over $\set{\X}$, and $\sem{\pi_1}$ over $\set{\Y}$. So for example  we have:	
	{\small 	
	\[
		{\sem \pi}_{\true\true} = \phi^1_\true \cdot \phi^2_{\true\true} \cdot \phi^3_\true \cdot\phi^4_{\true\true} +
		\phi^1_\false \cdot \phi^2_{\false\true} \cdot \phi^3_\false \cdot\phi^4_{\false\true}
		\qquad
		\sem{\pi_1}_\true = \phi^1_\true \cdot  \phi^2_{\true\true}  +  \phi^1_\false \cdot \phi^2_{\false\true}
		\qquad
		\sem{\pi_2}_\true = \phi^3_\true \cdot  \phi^4_{\true\true}  +  \phi^3_\false \cdot \phi^4_{\false\true}
	\]
	}
	Therefore
{\small $
	\sem{\pi}_{\true\true} \not=  
	\left(\phi^1_\true \cdot  \phi^2_{\true\true}  +  \phi^1_\false \cdot \phi^2_{\false\true} \right)
	\cdot 
	\left(   \phi^3_\true \cdot  \phi^4_{\true\true}  +  \phi^3_\false \cdot \phi^4_{\false\true}   \right)=
	\sem{\pi_1}_\true \cdot \sem{\pi_2}_\true = 
	(\sem{\pi_1}\Fprod\sem{\pi_2})_{\true\true}.
	$}
\end{example}

\paragraph{Proving Compositionality.}\label{sec:compositionality}
We are ready to prove that $\sem \pi$ can be compositionally defined for every typed derivation $\pi \dem \Lambda \vdash t: \iL$.
The crucial property---guaranteed by the type system---is that $\pi$ is well-formed, in the technical sense given below. Such a property is the key ingredient in the proof of compositionality, because it ensures the distributivity of the sum  over the product.

\begin{definition}[Well-Formedness] \label{def:connection} \
	\begin{enumerate}
		\item The type derivations $\pi_1\dem J_1, \dots, \pi_n \dem J_n$ are \emph{compatible} if 
		\[	\mbox{ exists } j \st \Z\in \Nm {\pi_j} \tand \Z\not \in \Nm {J_j}\qquad\Rightarrow \qquad \Z\not \in \Nm{\pi_i} \mbox{ for any } i\not = j.\]
		\item 	A type derivation $\pi$ is \emph{well-formed} if for every rule in $\pi$, its premises are compatible. 
	\end{enumerate}
\end{definition}

\begin{lemma}[Well-Formed Derivations]\label{lem:connexity}
	Every type derivation $\pi\dem \Lambda \vdash \tm:\iL$  is well-formed.
\end{lemma}
We postpone to \Cref{sec:flow} the discussion of the proof,
 which relies on a fine analysis of the flow of the computation. 
  Using this result, we are able to prove that the semantics---that we have defined in a global way---can  indeed be computed compositionally, validating \Cref{eq:compositionality}.
%
\begin{proposition}[Compositionality]\label{lem:comp}
	Let  $\pi \dem J$ be the following type derivation:
	\[ \infer{\pi \dem J}{ \pi_1\dem J_1 & \dots &\pi_n\dem J_n } \]
	If $\pi_1,\ldots,\pi_n$ are compatible, then:
	\[\sem{\pi} = \sum_{\bZ}  \left( {\BigFProd_i \sem{\pi_i}}\right)\quad
	\text{ where } \bZ\defeq \bigcup_i \Nm{\sem{\pi_i}} - \Nm{J}\]
\end{proposition}
\begin{proof}
	Wlog,  consider $n=2$.	Let us set $\Phi_{\pi}\defeq\BigFProd_{\phi\in \Cpts{\pi}} \phi$ 
	for each derivation $\pi$.
	 By  \Cref{def:sem}, we can write $\sem{\pi_i}=\sum_{\bW_i} \Phi_{\pi_i}$,
	where $\bW_i=\Nm{\Phi_{\pi_i}} - \Nm{J_i}$ ($i \in \{1,2\}$). 
	Crucially, since $\pi_1$ and $\pi_2$ are compatible, $\bW_1 \cap \Nm{{\pi_2}} = \emptyset$ and $\bW_2 \cap \Nm{{\pi_1}} = \emptyset$.
	Hence, by \Cref{eq:sum_prod}, sum and product distribute, and we have:
	\begin{align}\label{eq:binary_rule}
		\sem{\pi_1} \Fprod \sem{\pi_2} = & \sum_{\bW_1} \Phi_{\pi_1} \FProd \sum_{\bW_2} \Phi_{\pi_2}
		= \sum_{\bW_1} \sum_{\bW_2}  ( \Phi_{\pi_1} \FProd  \Phi_{\pi_2}) = \sum_{\bW_1} \sum_{\bW_2} \Phi_\pi
	\end{align}
	where we used the fact that $\Cpts{\pi} = \Cpts{\pi_1} \cup \Cpts{\pi_2}$.
	Now let $\bY_i\defeq\Nm{\sem{\pi_i}}$; since $ \Nm{\Phi_{\pi_i}} =  \bY_i \uplus \bW_i$ and,
	by the compatibility of $\pi_1$ and $\pi_2$, $\bW_i \cap \Nm{\Phi_{\pi_j}} =\emptyset$ for $i\not=j$, we have
	\begin{equation}\label{eq:names_prod}
		\Nm{\Phi_\pi} = (\bY_1\cup \bY_2) \uplus (\bW_1\uplus \bW_2).
	\end{equation}
	Let $\bZ \defeq (\bY_1\cup\bY_2) - \Nm{J}$; by \Cref{eq:names_prod} and compatibility, 
	$\Nm{\Phi_\pi} - \Nm{J} = \bZ \uplus \bW_1 \uplus \bW_2 $. Therefore
	\[ \sum_\bZ	\big(\sem{\pi_1} \Fprod \sem{\pi_2} \big) = 
	\sum_\bZ \sum_{\bW_1} \sum_{\bW_2} \Phi_\pi = \sum_{\Nm{\Phi_\pi}-\Nm J} \Phi_{\pi} ~\defeq ~\sem{\pi} \] 
	where we  sum out $\bZ$ from both sides of \Cref{eq:binary_rule}.
\end{proof}
 

\begin{figure} [t]
	\small
	\begin{mybox}
		\textbf{Higher-Order $\Bang$-Calculus}
		\begin{myboxC}
			\textbf{First-Order Rules} 
			\begin{gather*}
			\infer[\isample]{ \Lambda \vdash \sample d : \X   \dmd { \ft{\phi}{\{\X\}}  } }
			{ \X \not\in \Nm{\Lambda}  
				\quad{\scriptsize (\phi \mbox{ as in \Cref{def:CPTpax}}) }   } 
			\quad
			\infer[\icond]{ \Lambda, y_1:\Y_1, \dots, y_n :\Y_n \vdash \CPT{y_1, \dots, y_n} :  \X \dmd { \ft{\phi}{\{\Y_1, \dots, \Y_n, \X \} } } }
			{ \X \not \in \{\Y_1, \dots, \Y_n\} \tand \X \not \in \Nm{\Lambda} \quad  
				{\scriptsize (\phi \mbox{ as in \Cref{def:CPTpax}}) } }
			\\[2pt]
			\infer[\ivar]{ \Lambda, x  : \iP \vdash x  : \iP \dmd{ \ft{1}{ \emptyset} } }{  }
			\quad
			\infer[\ilet]{\Lambda, \Gamma_1 \uplus \Gamma_2 \vdash \letin{x}{u}{t} : \iA \dmd{ \ft{\sum_\bZ \psi_1 \odot \psi_2}{(\bY_1\cup\bY_2) - \bZ} } }
			{ \Lambda, \Gamma_1 \vdash u : \iP \dmd{ \ft{\psi_1}{\bY_1} } 
				&  \Lambda,\Gamma_2, x :  \iP \vdash t : \iA  \dmd{ \ft{\psi_2}{\bY_2} } 
				&  \scriptstyle{\bZ = (\bY_1\cup \bY_2) - \Nm{\Lambda, \Gamma_1, \Gamma_2, \iA} }
			}
			\\[2pt]	
			\infer[\ipair]{ \Lambda \vdash \pair{v}{w} : \iL_1 \otimes \iL_2  \dmd{ \ft{1}{ \emptyset} } }
			{\Lambda \vdash{} v : \iL_1 \dmd{ \ft{1}{\emptyset} } & \Lambda \vdash w : \iL_2 \dmd{ \ft{1}{\emptyset} } }
			\qquad
			\infer[\iletp]{\Lambda, \Gamma \vdash \letp{\pair{x}{y}}{\val}{t} : 
				\iA  \dmd{ \ft{\psi}{\mathbb{Y}} } }
			{ \Lambda \vdash \val :\iL_1 \otimes \iL_2 \dmd{ \ft{1}{\emptyset} } 
				&  \Lambda, \Gamma, x : \iL_1, y : \iL_2 \vdash t : \iA \dmd{ \ft{\psi}{\mathbb{Y}} } 
			}
			\end{gather*}
		\end{myboxC}	
		\vskip -.4cm		
		\begin{gather*}
		\infer[\iabs]{\Lambda, \Gamma \vdash \lambda x.t : \iP \larrow \iA \dmd{ \ft{\psi}{\bY} } }
		{\Lambda, \Gamma, x : \iP \vdash t : \iA \dmd{ \ft{\psi}{\bY} } }
		\quad
		\infer[\iapp]{\Lambda, \Gamma_1\uplus\Gamma_2 \vdash tv : \iA \dmd{ \ft{\sum_\bZ \psi_1 \odot \psi_2}{(\bY_1 \cup\bY_2) - \bZ} } }
		{\Lambda, \Gamma_1 \vdash t : \iP \larrow \iA \dmd{ \ft{\psi_1}{\bY_1} } 
			&  \Lambda, \Gamma_2 \vdash v : \iP \dmd{ \ft{\psi_2}{\bY_2} } 
			&  \scriptstyle{ \bZ = (\bY_1\cup \bY_2) - \Nm{\Lambda, \Gamma_1, \Gamma_2, \iA} }
		}
		\\[2pt]
		\infer[\ibang]{\Lambda, \biguplus^n_{i=1} \Gamma_i \vdash \oc t : [\iA_1, \dots, \iA_n] \dmd{ \ft{\bigodot_i \psi_i}
		{\bigcup_i \bY_i} } }{ \big( \ \Lambda, \Gamma_i \vdash t: \iA_i \dmd{\ft{\psi_i}{\bY_i}} \ \big)^n_{i=1} & n \geq 0 }
		\qquad
		\infer[\ider]{\Lambda, \Gamma \vdash \der v :  \iA \dmd{ \ft{\psi}{\bY} } }
		{\Lambda, \Gamma \vdash v : [\iA]   \dmd{ \ft{\psi}{\bY} } }
		\end{gather*}	
	\end{mybox}								
	\caption{Inductive Interpretation of typed terms (judgments are annotated---in blue---with their interpretation).}
	\label{fig:decorated}\label{fig:inductive_sem}
\end{figure}

\paragraph{Inductive Interpretation of Type Derivations.} \label{sec:interpretation}
Now that we have proved that our semantics is compositional, we are able  to compute it inductively,
 starting from the axioms, and then 
  following the structure of the type derivation.  Probabilistic axioms are assigned  a \cpt as indicated in 
  \Cref{def:CPTpax}. The $\ivar$ axioms  are assigned the  trivial factor $\ftone_\emptyset$ (see \Cref{ex:var_ax}). Then the semantics of each sub-derivation  is inductively obtained following \Cref{lem:comp}.  In  \Cref{fig:decorated} we decorate the intersection type system with factors, according to  this process. A decorated type judgment  is written $\tjudg{\Sigma}{\tm}{\iA}\dmd{\psi}$, where  $\psi$ is the inductively computed  factor.

\begin{lemma}\label{lem:semantics}
	Let $\pi\dem J\dmd{\psi}$ be a well-formed type derivation. Then $\sem{\pi}=\psi$.
\end{lemma}

\begin{proof}
	By induction on the derivation. The property trivially holds for all axioms.
	Assume that  $\pi \dem J\dmd{\psi_i}$ is obtained from derivations $\pi_i$ $(1 \leq i \leq n)$. Since $\pi$ is well-formed, by definition of well-formedness also the derivations $\pi_i$ are well-formed. Then
	by \ih, $\psi_i = \sem{\pi_i}$. Thus, by letting $ \bZ \defeq \bigcup_i \Nm{\psi_i} - \Nm{J} $, we have
	\[  \psi := \sum_\bZ \Big(\BigFProd_i \psi_i \Big) = 
	\sum_\bZ \Big(\BigFProd_i \sem{ \pi_i} \Big)=\sem{\pi}  \]
	where the last equality follows from \Cref{lem:comp}.
\end{proof}
Since derivations of ground type are always well-formed (\Cref{lem:connexity}), we have proved that:
\begin{theorem}\label{thm:core}
	Let $\pi\dem \tjudg{\Lambda}{\tm}{\iL}\dmd{\psi}$ be a derivation of ground type. Then $ \sem{\pi}=\psi $.
	%
\end{theorem}
This theorem states that 
the   semantics $\sem \pi$ (as in \Cref{def:sem}) of a type derivation $\pi$ can be   \emph{inductively computed}, as described in  \Cref{fig:decorated}. 

%

\paragraph{Invariance of the Semantics.}
The semantics   is invariant under reduction and expansion. This is due to the fact that probabilistic axioms are stable w.r.t. reduction and expansion. 
\SLV{}{Indeed, non-idempotent intersection derivations somehow internalize the process of rewriting.}
\begin{theorem}[Invariance]\label{thm:invariance}Let $\tm$ be a $\lambda_\bang$-term and $t\to u$. Then:
	$\Lambda  \vdash t: \iL\dmd{\psi}  \Leftrightarrow 
	\Lambda \vdash u: \iL\dmd{\psi}$.
\end{theorem}
 

\paragraph{Semantics Completion.}
	We conclude  with a  remark. 
The reader may expect that the interpretation of a type derivation $\pi\dem J$ were a factor over $\Nm J$. 
	For example, one could expect the interpretation of   an identity axiom to be 
		$ y : \nY \vdash y : \nY  \dmd{ \ft  {\ftone}{\set \Y} } $   \quad instead of \quad
		$ y : \nY \vdash y : \nY  \dmd{ \ft  {\ftone}{\emptyset}}.$
	The fact is that our semantics  focuses only on  the \emph{probabilistic content} of the derivation    $\pi \dem J$. Please notice  that  the non-probabilistic information is already fully contained in the type judgment $J$, because intersection types carry such  information. 
	Indeed, an interpretation of $\pi \dem J$ as a factor over $\Nm J$ is easily obtained by a form of \emph{completion}.  
	We give  more details in the \Appendix.  
We mention also  that the   completed interpretation is a needed step  to bridge  the gap between our semantics and \emph{weighted relational models/probabilistic coherence spaces} such as \cite{LairdMMP13,EhrhardT19, EhrPagTas14}.


\section{Bayesian Networks Go With The Flow}\label{sec:flow}
	
In this section  we prove  two results which we have already anticipated:
\begin{itemize}
	\item we show that every (closed) term $t$ of ground type corresponds to a \BN ${\bn_t}$, and that the two have the 
	\emph{same semantics}.\footnote{Precisely, the  joint distribution  underlying $\tm$ and $\bn_t$ is exactly the same. We then have to  take into account that, in general, the term $t$ encodes also a query, defining a \emph{marginal} of the full joint distribution (\Cref{prop:BN}).}
	\item we prove  the central result making the semantics of terms compositional, namely  that every type derivation $\pi \dem \Lambda \vdash t:\iL$ is well-formed  (\Cref{lem:connexity}). 
\end{itemize}
The key ingredient underlying both results is the same:
we associate to each type derivation $\pi$ a directed graph---$\flow \pi$---which  essentially describes the \emph{flow of the computation} in $\pi$. The crucial  fact is proving that $\flow \pi$ is \emph{acyclic}. Our argument exploits sophisticated techniques borrowed from the theory of Linear Logic, and in particular from  Girard's Geometry of Interaction \cite{GIRARD1989221, LagoFVY17, PPDP2020,POPL2021,LICS2021}. In order to convey more clearly the basic ideas, here we give the definitions only for the \linear fragment.  Recall that this fragment suffices to type every BN normal form.
The development for the full higher-order calculus is in   the \Appendix.

\paragraph{The Flow Graph of a Type Derivation is a DAG}\label{sec:dag}
Intuitively, the graph we are going to build tracks the occurrences of atomic types (\ie the  random variables) throughout the type derivation.   
We indicate  a specific occurrence of an atom  inside a  ground type $\iL$ by means of a (type) context, \ie a type with a hole, as follows:
\[\begin{array}{lrcl}
	\textsc{Ground Type Ctxs}&	 \cK, \cL & \grameq&  \ctxhole \gpipe  \cK \otimes \iL \gpipe \iK \otimes \cL 
\end{array}\]
	So $\cL$ denotes an occurrence of atom (here $\X$)
	inside the type $\iL \defeq \cLp{\X}$. Notice for  example that  the type  $(\X\otimes \X)\otimes \Y$ contains  three occurrences of atoms. 
Given a type derivation $\pi$, 
we assume given a distinct label to each occurrence of atom appearing in  each judgment of $\pi$. We call such a label  a \emph{position}.  We can now build a graph that has all the positions of  $\pi$ as vertices, and that tracks the flow of each name  $\X$.
\begin{definition}[Flow Graph]
	Let $\pi \dem \Lambda \vdash t:\iL$.
	The flow graph $\flow{\pi}$ of $\pi$ is the \emph{directed graph} which  has as  \emph{vertices} all the positions occurring in $\pi$, and \emph{edges} as indicated in \Cref{fig:normal_flow}.
\end{definition}

\begin{example}\label{ex:flow}
	We show the type derivation for the term (\ref{fig:BNterm}) encoding our initial example, annotated with the flow graph. The reader can already notice that the flow exactly matches the DAG corresponding to the  Bayesian network  in \Cref{fig:BNrain}.
	\begin{small}
		\begin{prooftree}
		\def\extraVskip{3pt}
		\AxiomC{$\vdash \coin{0.6} : \vnode{14}{\nD}$}
			\AxiomC{$d : \vnode{17}{\nD} \vdash \CPTN{\nS}{d} : \vnode0{\nS}$}
				\AxiomC{$d : \vnode1{\nD} \vdash \CPTN{\nR}{d} : \vnode2{\nR}$}
					\AxiomC{$ s : \vnode4{\nS}, r : \vnode3{\nR} \vdash \CPTN{\nW}{s,r} : \vnode5{\nW}$}
						\AxiomC{$w : \vnode6{\nW} \vdash w : \vnode7{\nW}$}
					\BinaryInfC{$ s : \vnode9{\nS}, r : \vnode8{\nR} \vdash \letin{w}{\CPTN{\nW}{s,r}}{w} : \vnode{10}{\nW} $}
				\insertBetweenHyps{\hskip .2cm}
				\BinaryInfC{$d : \vnode{16}{\nD}, s : \vnode{11}{\nS} \vdash \letin{r}{\CPTN{\nR}{d}}{\letin{w}{\CPTN{\nW}{s,r}}{w}} : \vnode{12}{\nW} $}
				\insertBetweenHyps{\hskip .2cm}
			\BinaryInfC{$ d : \vnode{15}{\nD} \vdash \letin{s}{\CPTN{\nS}{d}}{\letin{r}{\CPTN{\nR}{d}}{\letin{\nW}{\CPTN{\nW}{s,r}}{w}}} : \vnode{13}{\nW} $}
		\insertBetweenHyps{\hskip .2cm}
		\BinaryInfC{$\vdash \letin{d}{\coin{0.6}}{ \letin{s}{\CPTN{\nS}{d}}{\letin{r}{\CPTN{\nR}{d}}{\letin{w}{\CPTN{\nW}{s,r}}{w}}} }  : \vnode{18}{\nW}$}
		$
		\dirflowedges{node9/node4, node8/node3, node7/node10, node10/node12, node12/node13, node15/node16, node15/node17, node16/node1, node11/node9, node13/node18}
		\specdirflowedge{node0}{node11}{to path={ .. controls +(.6,-.5) and +(-.6,-0.5) .. (\tikztotarget) } }
		\specdirflowedge{node5}{node6}{to path={ .. controls +(.5,-.5) and +(-.5,-0.5) .. (\tikztotarget) } }
		\specdirflowedge{node2}{node8}{to path={ .. controls +(.5,-.5) and +(-.5,-0.5) .. (\tikztotarget) } }
		\specdirflowedge{node14}{node15}{to path={ .. controls +(.5,-.5) and +(-.5,-0.5) .. (\tikztotarget) } }
		\bentdirflowedges{node1/node2/55, node6/node7/55, node3/node5/45, node4/node5/45, node17/node0/55}
		$
		\end{prooftree}
	\end{small}
\end{example}
We will 
exploit the fact that the flow graph of a type derivation of ground type is acyclic.
\begin{proposition}[The Flow is Acyclic] \label{prop:DAG}
	Let $\pi \dem \Lambda \vdash t:\iL$. Then $\flow{\pi}$ is a DAG.
\end{proposition}
\begin{proof}[Proof Sketch]
	If $\tm$ is in normal form, it is immediate to verify (by induction on the derivation $\pi$) that $\flow{\pi}$ is \emph{acyclic}. If $\tm$ is not in normal form, by  \Cref{thm:compileI} we know that there exists a term $\tmtwo$ in  normal form such that $\tm\to^*\tmtwo$. It is then enough to prove that cycles are preserved along a reduction sequence, which we do by strengthening the subject reduction statement.
\end{proof}

\begin{figure}[t]
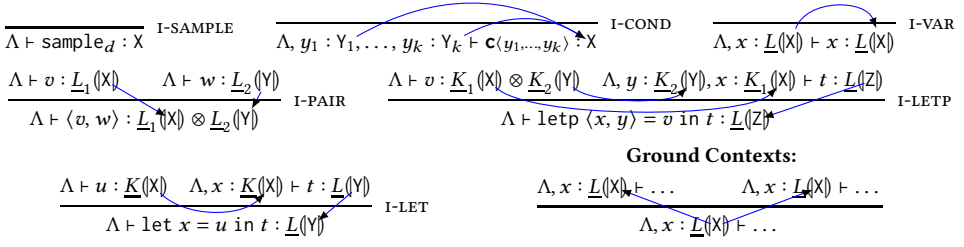

	{\footnotesize 	 
		\centering
		\[
		\infer[\icoin]{\Lambda \vdash \sample{d} : {\X} }{} 
		\qquad 
		\infer[\icond]{\Lambda, y_1 : \vnode0{\Y_1}, \dots, y_k : \vnode1{\Y_k}  \vdash  \CPT{y_1, \dots, y_k}: \vnode2{\X} }{}
		\bentdirflowedges{node0/node2/30, node1/node2/30}
		\qquad
		\infer[\ivar]{ \Lambda, x : \cLp{\vnode0{\X}} \vdash x : \cLp{\vnode1{\X}} }{}  
		\bentdirflowedges{node0/node1/70}
		\]
%
%
		\[
		\infer[\ipair]{ \Lambda \vdash \pair {v}{w} : \cLnp{1}{\vnode0{\X}} \otimes \cLnp{2}{\vnode1{\Y}} }
		{ \Lambda \vdash v : \cLnp{1}{\vnode2{\X}} & & \Lambda \vdash w : \cLnp{2}{\vnode3{\Y}} } 
		\dirflowedges{node2/node0, node3/node1}
		\qquad
		\infer[\iletp]{ \Lambda \vdash \letp{\pair{x}{y}}{v}{t} : \cLp{\vnode5{\Z}} }
		{ \Lambda \vdash v : \cKnp{1}{\vnode0{\X}} \otimes \cKnp{2}{\vnode2{\Y}} & \Lambda, y : \cKnp{2}{\vnode1{\Y}}, x : \cKnp{1}{\vnode3{\X}} \vdash t : \cLp{\vnode4{\Z}} }
		\specdirflowedge{node0}{node3}{to path={ .. controls +(.5,-.5) and +(-.5,-0.5) .. (\tikztotarget) } }
		\specdirflowedge{node2}{node1}{to path={ .. controls +(.3,-.3) and +(-.3,-0.3) .. (\tikztotarget) } }
		\dirflowedges{node4/node5}
		\]
		
		\begin{minipage}[t]{0.4\linewidth}
		\[		
		\infer[\ilet]{\Lambda \vdash \letin{x}{\tmu}{\tm : \cLp{\vnode0{\Y}}} }
		{ \Lambda \vdash \tmu : \cKp{\vnode1{\X}} & \Lambda, x : \cKp{\vnode2{\X}} \vdash \tm : \cLp{\vnode3{\Y}} }
		\specdirflowedge{node1}{node2}{to path={ .. controls +(.4,-.4) and +(-.4,-0.4) .. (\tikztotarget) } }		
		\dirflowedges{node3/node0}	
		\]
		\end{minipage}
		\qquad		
		\begin{minipage}[t]{0.4\linewidth}
		\centering
		\textbf{Ground Contexts:}
		\[
		\infer{\Lambda, x : \cLp{\vnode0{\X}} \vdash \dots}
		{\Lambda, x : \cLp{\vnode1{\X}} \vdash \dots & \quad & \Lambda, x : \cLp{ \vnode2{\X}} \vdash \dots}
		\dirflowedges{node0/node1, node0/node2}
		\]	
		\end{minipage}
	}
	\caption{Flow graph on the low-level fragment. The flow graph for the full calculus is defined in the \Appendix.} \label{fig:normal_flow}
\end{figure}

\subsection{Terms and  Bayesian Networks: Soundness and Completeness}\label{sec:terms_as_BN}
We can encode any \BN $\bn$ with a low-level term $\tm_\bn$, in the standard way; it is immediate to check that $\bn$ and  $\tm_\bn$ have the same semantics. 
Conversely, the flow graph gives us a way to extract a  Bayesian network $\bn_t$ from every  term $t$ of ground type.  Indeed, each probabilistic axiom corresponds to a random variable, and the dependencies   between \paxs in the flow graph correspond to the edges in the Bayesian network. We have already illustrated this process in \Cref{ex:flow}. More formally, we extract a \BN as follows.
{For clarity, we  focus on \emph{closed} terms. If the term is not closed,  what we would  extract is a \emph{conditional} \BN.} 

\begin{definition}[BN Extraction]\label{def:extraction}

Let $t$ be a  closed $\Bang $-term of ground type. The derivation $\pi~\dem \vdash t:\iL$ defines 
 a Bayesian network $\bn_t = (\mathcal G,\aphi)$ over $\Nm{\pi}$ where:
	\begin{itemize}
		\item $\aphi$ maps each name $\X\in\Nm{\pi}$ to the factor which is associated to the axiom introducing $\X$.
		\item $\mathcal G$ is obtained from $\flow{\pi}$ by collapsing all the vertices  labeled by the same name. 
	\end{itemize}
\end{definition}
This extraction process is correct, in the following sense 
\SLV{(the proof is straightforward).}{(the proof is straightforward by unfolding  the definition of $\sem \pi $ and $\sem {\bn_t}$).}
\begin{thm}[From Ground Type Terms to BNs]\label{prop:BN}\label{thm:BN}
		Let $\pi\,\dem \vdash t:\iL$ be a derivation of ground type, and $\bn_t$ the \BN associated to it, as in  \Cref{def:extraction}. Then  \[\sem \pi = \sum_{\Nm{\pi}-\Nm{\iL}}   \sem {\bn_t}.\]
		
	\condinc{}{Let  $ \bn $ be a BN, and $t_{\bn}$ a low-level term term encoding it. 
	Then $\sem {t_{\bn}} = \sem \bn$ (recall that  $t_{\bn}$ has a unique type derivation)}
\end{thm}

\subsection{From DAGs to Compositionality}\label{sec:proof_of_connexity}
We use the technology of the flow graph also to prove the fundamental result which we used to determine  compositionality, namely that all derivations of ground type are well-formed (\Cref{lem:connexity}). 
The crucial property  is the following, which is  remarkably reminiscent of the  characterizing
property of jointrees, the data structure underlying 
the \emph{message passing} algorithm for exact inference on BNs. 
\SLV{The proof relies on the fact that the flow graph is a DAG.}{The proof (obtained from the fact that the flow graph is a DAG) is in the \Appendix.}
\begin{lemma}[Named Paths] \label{lem:jointree} 
	Let $\pi \dem \Lambda \vdash t: \iL$, and let $\nX$ be the main name of a \pax $\alpha^\X$. Then  in 
	$\flow{\pi}$,  each position with  name $\nX$ is  connected to the occurrence
	of  $\, \X$ in $\alpha^\X$ by a path where all  positions have  name $\nX$.
\end{lemma}
%
Then, we are finally able to give the proof of \Cref{lem:connexity}.
\begin{proof}[Proof of \Cref{lem:connexity}]
	We prove that each rule in $\pi\dem \Lambda \vdash t:\iL$ has compatible premises. 
	Let $\pi_1 \dem J_1, \dots, \pi_n\dem J_n$ be the premises of a rule in $\pi$. Assume there exists a $\pi_j$ such that 
	$\nX \in \Nm{\pi_j}$ and $\nX \not \in \Nm{J_j}$. It is easy to see that $\nX$ is necessarily  the main name of  
	a \pax $\alpha^\nX$ in $\pi_j$. 
By	\Cref{lem:jointree}, each position with  name $\nX$  in $\flow \pi$ is  connected to the occurrence
	of $\X$ in $\alpha^\X$ by a path where all positions have name $\nX$. Since $\nX \not \in \Nm{J_j}$,
	it is  impossible that $\nX \in \Nm {\pi_i}$ for $i \not = j$.
\end{proof}


\section{Dealing with Evidence, aka  Computing the Posterior}\label{sec:evidence}
\newcommand{\Q}{\mathtt Q}
\newcommand{\q}{\mathtt q}
\newcommand{\E}{\mathtt E}
\newcommand{\e}{\mathtt e}


Quoting  \cite{GordonHNR14}, inference is the task of ``computing an explicit representation of the probability distribution implicitly specified by a probabilistic program''. So far,  we have focused  on the computation of  marginals, without explicitly dealing with evidence.
Admittedly, 
the most interesting inference problem is to compute  a \emph{posterior} distribution, such as $ \Pr(\Rain|\Wet=\true) $. 
We know  that
 this is the same as $ \Pr (\Rain, \Wet=\true) / \Pr(\Wet=\true) $ (as recalled in \Cref{sec:informal}). In standard practice, to compute a  posterior means to compute the  unnormalized  posterior 
  \begin{equation}\label{eq:joint}
\Pr(\Rain, \Wet=\true).
 \end{equation} We  can then
 normalize \eqref{eq:joint}  by dividing for 
 $\Pr(\Wet=\true) $, 
which   is also  available---\emph{for free}---from \eqref{eq:joint},  just by adding up all the probabilities  appearing in  it. 
Indeed, $ \Pr(\Wet=\true) =\sum_{\Rain} \Pr(\Rain, \Wet=\true)$.

Technically, in the setting of BNs (and in our setting), 
there is no essential difference between computing a marginal like  $\Pr (\Rain, \Wet=\true)$ or a marginal like $\Pr(\Rain, \Wet)$. 
The key to deal with evidence (such as $\Wet=\true$) is to update each \cpt in the model, by zeroing out those rows   that are inconsistent with the observed data. 
The resulting  factors are  no
longer normalized, but are   still  valid factors, and their product  gives the unnormalized posterior distribution.
We 
illustrate this with an example, referring to \cite{DarwicheBook} (Ch. 6.7) or  \cite{KollerBook} (Ch. 9.3.2)  for the details.


\begin{example}[Basic]\label{ex:evidence} Consider a simple BN of graph $\Rain \rightarrow \Wet$. The factors associated to $\Rain$ and $\Wet$ are respectively $\phi_1$ and $\phi_2$, as given in \Cref{fig:paxs}. We know (\Cref{sec:BN}) that $\Pr(\Rain,\Wet) = \phi_1 \Fprod \phi_2$.
Assume we have  \emph{evidence}   $\e\defeq (\Wet=\true)$. The (unnormalized) posterior distribution 
 $\Pr(\Rain, \Wet=\true)$ is again the product of the factors associated to the nodes of the BN (as in \Cref{thm:pearl}) but this time starting from the updated factors $\phi_1^\e$ and $\phi_2^\e$.   In our case, $\phi_1^\e = \phi_1$, while $\phi_2^\e$ is modified as in  \Cref{fig:evidence}. So  $\Pr(\Rain, \Wet=\true) = \phi_2^\e\Fprod\phi_1^\e $ (the result is  in \Cref{fig:evidence}).
{\small
\begin{figure}[h]\label{fig:efactors}\centering
\[	
\phi_2^\e =  
\begin{array}{|cc|c|}
	\hline	
	\R  & \W  & \Pr(\W=\true | \R) \\
	\hline	
	\true  & \true	& 0.7 \\
	\hline
	\false & \true	& 0.01 \\
	\hline
\end{array}
\quad
{ \phi_2^\e\Fprod\phi_1^\e } =
\begin{array}{|cc|c|}
	\hline	
	\R   & \W  &   \Pr(\R,\W=\true) \\
	\hline	
	\true  & \true	& 0.14 \\
	\hline
	\false & \true	& 0.008 \\
	\hline
\end{array}
\quad
\dfrac{\phi_2^\e\Fprod\phi_1^\e}{0.148} =
\begin{array}{|cc|c|}
	\hline	
	\R  & \W  & \Pr(\R|\W=\true) \\
	\hline	
	\true  & \true	& 0.14 /0.148=  0.946 \\
	\hline
	\false & \true	& 0.008 /0.148 = 0.054 \\
	\hline
\end{array}
\]
\caption{Dealing with evidence} \label{fig:evidence}
\end{figure}
}
Hence, we immediately have: 
 \begin{itemize}
 	\item $\Pr(\Wet=\true) =0.148$, which is obtained   by adding up all  the entries in  $\Pr(\Rain, \Wet=\true)$;
 	\item 
 	$\Pr(\Rain| \Wet=\true)$, which is obtained by normalizing $\Pr(\Rain, \Wet=\true)$, as depicted in \Cref{fig:evidence}. 
 \end{itemize}
So, for example, we have the posterior $\Pr(\Rain=\true|\Wet=\true)=0.14 /0.148=  0.946$, a huge increase from our prior belief  $\Pr(\Rain=\true)\defeq \phi_1(\true)=0.2$.

\end{example}

\subsection{The  Semantics of Terms, with Evidence}
The main ingredients being the same, our   approach is easily adapted to deal with evidence,  syntactically encoded by  the construct $\obsb x$. 
To express evidence, we enrich the definition of  intersection types  with   two more sets of atomic types, namely  $\{\Xt \mid  \X\in \Names\}$ and $\{\Xf \mid  \X\in \Names\}$.  
%
\[\begin{array}{lrcl}
	\textsc{Atomic Types}& \XX	 & \grameq&  \X \gpipe    \Xt   \gpipe  \Xf  \\[4pt]
	\textsc{Ground Types}&	 \iK, \iL & \grameq&  \XX  \gpipe \iK\otimes \iL  
\end{array} 
\qquad\vline\quad
\begin{array}{lrcl}
	\textsc{Positive Types}&	 \iP,\iQ & \grameq & L \gpipe\, \mset{A_1,\ldots,A_n}  \\[4pt]
\textsc{Types}&\iA,\iB &\grameq&  \sP \gpipe \iP \larrow \iA  
\end{array} \]
Technically, we require the atomic types appearing in a type derivation to be pairwise consistent, meaning that if two atomic types $\XX_1,\XX_2$ have the same name $\X$, then they  are both either  $\X$ or  $\Xt$ or $\Xf$. It is easy to check that for this property to hold for a type derivation  $\pi\dem J$, it just suffices to require that it holds for the  atomic types in the conclusion  $J$.

\condinc{}{NOTE For CF: add this result in the Appendix, formally}

Given a derivation $\pi$ and a name $\X$, we define $\Val{\X}=\{\true\}$  if $\X$ occurs in $\pi$ as $\Xt$, $\Val{\X}=\{\false\}$  if $\X$ occurs in $\pi$ as $\Xf$, and $\Val{\X}=\{\true,\false\}$ otherwise. Essentially, we treat an observed \rv as unary, akin to factor reduction in \cite{KollerBook}.

\paragraph{Types and Semantics.}
We can now complete  \Cref{fig:iTypes}  (and  \Cref{fig:inductive_sem}) with the   typing rule for $\obsb{x}$: 
\[	\infer[\iobs]{ \Lambda, x:\Xb  \vdash \obsb{x} : \Xb  \dmd{ \ft{1}{ \emptyset} }}{} \]
where we  annotated the rule with  the  semantics  indicated by \Cref{def:sem} (a remark similar to \ref{ex:var_ax} applies).
The typing rules for \isample and \icond in \Cref{fig:iTypes} (and \Cref{fig:inductive_sem}) are  updated as follows:
{\small \begin{gather*}
	\infer[\isample]{ \Lambda \vdash \sample d : \XX   \dmd { \ft{\phi}{\{\X\}}  } }
{ \X \not\in \Nm{\Lambda} } 
\qquad
\infer[\icond]{ \Lambda, y_1:\YY_1, \dots, y_n :\YY_n \vdash \CPT{y_1, \dots, y_n} :  \XX \dmd { \ft{\phi}{\{\Y_1, \dots, \Y_n, \X \} } } }
{ \X \not \in \{\Y_1, \dots, \Y_n\} \tand \X \not \in \Nm{\Lambda} }
\end{gather*}}
 where we annotate the rules with  the factor $\phi$ indicated  by \Cref{def:CPTpax}, taking into account the definition of $\Val \X$ given above (see \Cref{ex:evidence1}).  
Everything else stays the same: all definitions and results  in Sections  \ref{sec:term-semantics} and \ref{sec:flow} remain valid.
  An  example will clarify what happens, and how the  evidence is reflected into the semantics, via  the type derivation.
\begin{example}[Basic]\label{ex:evidence1} Consider a term modeling the  BN of \Cref{ex:evidence}, \ie  $\Rain \rightarrow \Wet$:  
	\[\letin \rain {\coin {0.2}} {\letin \wet { \CPT{\rain}} {\pair \rain \wet}}:\R\otimes \W.\] 
	where $ \CPT{\rain} \defeq   \texttt{case}\; \rain
	\;\texttt{of}\;\{\tmt\caseSym\coin{0.7}; 
	\tmf\caseSym	\coin {0.01}\}
	$.
	To express the evidence $\W=\true$, we use the  construct $\obs \wet \true$, yielding the (sugared) term		
	\[	\letin \rain {\coin {0.2}} {\letin \wet { \CPT{\rain}} {\pair {\rain} {\obs \wet \true }}}: \R \otimes \W^\true \]
	Let us see how  the evidence is reflected into the semantics, thanks to the type derivation $\pi$. 
	{\small 	\[ \infer{\pi \dem \vdash \letin r {\mathtt{bern}_{0.2}}  
			{\letin{w}{\CPT  r}{\pair r  {\obs w \true}}}  :{ \R}\otimes { \W^{\true}} \dmd{\phi_1\FProd \psi}}
		{ {\vdash \mathtt{bern}_{0.2}:{\R} \dft{\phi_1}{\set{\R}}}   
			& 
			{ \infer[{\ilet}]{{r:{\R} \vdash {\letin{w}{\CPT  r}{\pair r  {\obs w \true}}}} :\R \otimes \W^{\true} \dmd{\psi}}
				{{\infer[{\icond}]{ r: \R \vdash \CPT  r: {\W^{\true}} \dft {\psi}{\set{\R,\W}}}{}} & \infer{r:{ \R},   w: \W^{\true}  \vdash \pair r  {\obs w \true} : { \R}\otimes { \W^{\true}}\dmd{\ftone_\emptyset}}
					{ r:{ \R} \vdash  r :{ \R} \dmd{\ftone_\emptyset}  & { \infer[{\iobs}]{ w:{ \W^{\true}} \vdash  {\obs w \true}:{ \W^{\true}}\dmd{\ftone_\emptyset}}{}}   } }    }}
		\]}
	The \iobs axiom forces the  type  $\W^\true$. In turn, $\W^\true$ is  propagated to the $\icond$ axiom, via the $\ilet$ rule.
	What is  the semantics of the \icond axiom?  Spelling out \Cref{def:CPTpax}, and recalling that here  $\Val{\W}=\{\true\}$, we have that the factor $\psi$ associates a scalar to each of the \emph{two tuples} in $\Val{\R}\times\Val{\W}$: 
\begin{center}
		$\psi{(\true\true)}\defeq\coin{0.7}(\true)$ and $\psi{(\false\true)}\defeq\coin{0.01}(\true)$.
\end{center}
That is, $\psi$ is exactly the factor $ \phi_2^\e$ in \Cref{fig:evidence}. Please observe that the difference in the interpretation is entirely due to the type $\W^{\true}$. With the same definitions as in \Cref{sec:term-semantics} , we have  $\sem \pi 
=\phi_1\Fprod \psi$, which  is exactly 
$ \phi_1^\e\Fprod \phi_2^\e ~= ~\Pr(\Rain,\Wet=\true)$.
\end{example}
A very similar reasoning would apply to express $\Pr(\Rain|\Wet=\true)$ in a model with several variables, like the one in \Cref{fig:BNrain}; marginalization is taken care by the semantics exactly as before.
\condinc{}{ \todoc{Give this  in the \Appendix, if time allows}}
 It is more interesting to revisit \Cref{ex_coins}, to  illustrate how we  deal with  
 evidence together with \emph{copies}.



\begin{example}[Coin Tosses]\label{ex_coins2} Consider again the term $\tmu$ of \Cref{ex_coins}, modeling two tosses of the same (biased) coin. 
Assume we want to \emph{infer (learn) the bias of the coin} from the result of the tosses.  
For example, if we observe that both  the  coin tosses yield $\true$, this would increase our confidence that the coin is biased toward $\true$. Indeed the semantics of the updated term  changes in this sense, as we show here.
 The following  term\footnote{As usual, we  use some syntactic sugar.} $\tmu'$ expresses the evidence, given   the model encoded by $\tmu$. Notice the \emph{modularity} in the encoding.
	\[\tmu' ~~\defeq ~~\letpin{x}{y_1,y_2}{\tmu}{\pair{x}{\obst{y_1},\obst{y_2}}} \]
	
%
	Once again, each  observation  $ \obst{y_i} $ forces  the type $\Y_{i}^{\true}$, so recording the evidence $\true$ for each $\Y_i$;  
	the type derivation takes care of propagating the evidence, as shown below:
	{\small \begin{prooftree}
		\AxiomC{\RED{$\pi' \dem \tmu: \X\otimes \nY_1^\true \otimes \nY_2^\true$}}
		\AxiomC{$x : \nX \vdash x : \nX$}
		\AxiomC{$ y_1 : \nY_1^\true \vdash \obst{y_1} : \nY_1^\true$}
		\AxiomC{$  y_2 : \nY_2^\true \vdash \obst{y_2} : \nY_2^\true$}
		\BinaryInfC{$ y_1 : \nY_1^\true, y_2 : \nY_2^\true \vdash \pair{\obst{y_1}}{\obst{y_2}} : \nY_1^\true \otimes \nY_2^\true$}
		\BinaryInfC{$ x : \nX, y_1 : \nY_1^\true, y_2 : \nY_2^\true \vdash \pair{x}{\obst{y_1},\obst{y_2}} : \nX \otimes \nY_1^\true \otimes \nY_2^\true$}
		\RightLabel{\iletp}
		\BinaryInfC{$\rho \dem ~\vdash \letpin{x}{y_1,y_2}{\tmu}{\pair{x}{\obst{y_1},\obst{y_2}}} : \nX \otimes \nY_1^\true \otimes \nY_2^\true$}
	\end{prooftree}}
The derivation $\pi' \dem \tmu$ (highlighted in red) is the same as the derivation $\pi\dem \tmu$ of \Cref{ex_coins}, but with each type $\nY_i$  replaced by  $\nY_i^\true$ (as forced by the \iletp  typing rule).
The semantics  changes as expected, reflecting that the evidence   increases our confidence that the coin is biased toward $\true$. 
The reader can find  the computation of the semantics developed in full in the \Appendix.
\end{example}


\section{A Cost-Aware Semantics}\label{sec:cost}
In this section we examine the cost of computing the semantics of a term $t$ of ground type, that is the cost of inferring the marginal distribution defined by $t$. We assume  $\pi$ to be an arbitrary  derivation of a judgment $J$ of shape $\Lambda \vdash t:\iL$.
Looking at  \Cref{def:core}  one easily realizes the following:
\begin{proposition}[Cost Upper Bound]\label{prop_uppercost}
The cost of computing $\sem \pi$ according to \Cref{def:core} is $\BigO{m_\pi \cdot 2^{n}}$, 
where $n= |\Nm{\Cpts\pi}|$ is the number of names which appear in $\Cpts\pi$, 
and $m_\pi \leq n$ is the number of \paxs in $\pi$. 
\end{proposition}
This is the \emph{upper bound} to the cost of computing the semantics of $\pi$. But there is more to the story.
Non-idempotent type systems are known to provide quantitative information about the typed term, 
such as bounds on the execution time. 
There is a rich  literature on systems which capture (possibly tight) complexity bounds in different styles of computation \cite{deCarvalho18, DBLP:journals/jfp/AccattoliGK20,ICFP2022}, including \emph{the expected runtime} of probabilistic computations \cite{LagoFR21}.
However, in  the setting of Bayesian modeling and exact inference, the runtime, and even the expected runtime, has little  relevance, because the \emph{dominating  cost} (in time and space) is due to the computation of the probability distribution. 
Such a cost is what interests us---our type system is able to provide accurate bounds.

\paragraph{On the Cost of Inductively Computing  $\sem{\pi}$.} 
Let us point out the differences, \emph{computationally}, between the definition of $\sem{\pi \dem J}$ according to \Cref{def:core}---which only considers the \paxs and the names occurring in the conclusion $J$---and  the inductive 
definition illustrated in  \Cref{fig:decorated}. 
The latter, which computes the semantics following the structure of the type derivation,  allows for a \emph{more efficient} computation, in general.
This is because summing out step-by-step yields intermediate factors of smaller size.
This indeed is exactly the way algorithms for exact inference work.
  
\begin{example}
	Consider a \BN whose DAG is  a chain $\nX_1 \to \nX_2 \to \dots \to \nX_n $. The following term defines the marginal over the single variable $\X_n$.
\[
\pi \dem \letin {x_1}{\sample d}
{\letin {x_2} {\CPTN{2}{x_1}}
		{\dots \ \letin{x_n}{\CPTN{n}{x_{n-1}}}{x_n}}} : \nX_n
	\]
\begin{itemize}
	\item 	Computing $\sem{\pi}$ according to \Cref{def:core} has cost $\BigO{n\cdot 2^n}$. Indeed, we   first compute the product of $n$ factors---that is  \emph{the full joint distribution} over $\nX_1, \dots, \nX_n$---and then  sum out all the variables but $\X_n$.
	\item 	Regardless of the number of variables in the chain, the cost of computing the  semantics  \emph{following the structure of the derivation} is $\BigO{n\cdot 2^3}$, as one can easily verify by trying to write down the derivation: we do not actually need to explicitly build a distribution over $n$ random variables.
\end{itemize}
Clearly, there is no guarantee that the cost of  inductively computing  the semantics of $\pi$  is always \emph{strictly less} than  computing  $\sem \pi$ directly.
\end{example}

%
%

If we examine the cost of computing the semantics of $\pi$ inductively, 
we obtain a better upper bound    than in  \Cref{prop_uppercost}---please observe  that the latter essentially corresponds to computing the full joint distribution underlying the model. 
\begin{proposition}[Inductive Cost]\label{prop_inductive_cost}
Let  $ \pi$ be a type derivation, $m_\pi$  the number of \paxs in $\pi$, 
and $n=|\Nm{\Cpts\pi}|$  the number of names which appear in $\Cpts\pi$ (as in \Cref{prop_uppercost}).
The cost of inductively computing the semantics of $\pi$ following the structure of the derivation is 
\[\BigO{m_\pi \cdot 2^W}\]
where 
$W\leq n $ is  the maximal cardinality of any set of names $\bY$ appearing in the derivation, 
when decorated as in \Cref{fig:decorated}.
\end{proposition}

\begin{remark}
The reader familiar with inference algorithms will realize that inductively computing the semantics following the structure of the derivation implements inference by a form of the Variable Elimination algorithm, and indeed 
such an observation can be made formal. The needed technical details, however, are beyond the scope of this paper.  
\end{remark}
\begin{remark}[Observed Variables]\label{rem:cost}
	Both bounds---in \Cref{prop_uppercost}  and \Cref{prop_inductive_cost}--- can be made tighter by restricting $n$ to the set $\widehat\bY$ of non observed names. The set $\bE$ of observed variables does not actually contribute to the number of entries in a factor, since each one  has a single possible value. So $\Val {\widehat\bY} \times \Val \bE$ is isomorphic to $\Val {\widehat\bY}$.
\end{remark}

\paragraph{Exact Bounds.}
Given that the  type system is resource-sensitive by design, it is not difficult to refine 
type derivations by  decorating  $\pi$ with the number of elementary steps needed to compute its semantics. We stress that this information can be retrieved \emph{without} performing calculations on factors---it is enough to  keep track of the relevant names occurring in each judgment. 
It is straightforward to fine-tune the complexity measure and provide bounds with different granularity.

In \Cref{fig:cost-system} we  give a minimalist  example, limited to the \emph{decidable} first-order system, where the weight on the turnstile counts the number of multiplications performed during the computation (the cost  of sums being negligible w.r.t. the cost of products). Here, given a judgment $J \, \diamond \, \bX$, we write $\widehat{\bX}$ for the restriction of  $\bX$ to the names that are not observed in $J$: this is to take into account that observed names do not actually contribute to the factor size, hence to the overall cost, as remarked above. 

\paragraph{Same Model, Different Cost.}
Since the cost of computing the semantics (\ie the cost of exact inference) depends on the structure of the term itself, different terms describing the very same marginal distribution may have different costs. Indeed, a term encodes not only a \BN and a marginal, but also a way to compute it, in agreement with the  ideas which underlay  exact inference \emph{on terms}, as proposed by \citet{KollerP97}. Our type system detects such a difference.
The following example  illustrates this point.
\begin{example}[Same BN, Different Cost]\label{ex:cost}
Consider the following two terms, both in normal form, and both corresponding to the same \BN 
whose graph is the  chain $\nX \rightarrow \nY \rightarrow \nZ$.
\begin{enumerate}
	\item  $t_1 \defeq \letin{x}{\sample{d}}{\letin{y}{\CPTN{\Y}{x}}{\letin{z}{\CPTN{\Z}{y}}{z}}} : \nZ $
	\item $t_2 \defeq  \letin{z}{(\letin{y}{(\letin{x}{\sample{d}}{\CPTN{\Y}{x}})}{\CPTN{\Z}{y}})}{z} : \nZ $
\end{enumerate}
Both  terms define the same marginal distribution over  $\nZ$. 
However, inductively computing such a distribution has a 
different cost\footnote{The reader can find the explicit computations   in the \Appendix.}:  $12$ multiplications for $t_1$, and  $8$ for $t_2$. 
\end{example}	
\begin{figure}[t]
	\small	
	\begin{gather*}	
	\infer[\icoin]{ \Lambda \ovdash{0} \sample d : \XX     { \dftC{\phi}{\{\X\}}  } }
	{ \X\not\in \Nm{\Lambda} } 
	\quad
	\infer[\icond]
	{ \Lambda, y_1:\YY_1,\dots,y_n :\YY_n \ovdash{0} \CPT{y_1, \dots y_n} : \XX \dftC{\phi}{\{\nY_1,\dots,\nY_n,\nX \} } } 
	{ \X \not \in\{\Y_1, \dots, \Y_n\} \tand \X \not \in \Nm{\Lambda} }
	\\[4pt]
	\infer[\ivar]{ \Lambda, x : \iP \ovdash{0} x : \iP { \dftC{1}{\emptyset} } }{ }
	\qquad 
	\infer[\iobs]{ \Lambda, x : \X^\tb \ovdash{0} \obs{x}{\tb} : \X^\tb { \dftC{1}{\emptyset} } }{ }
	\\[4pt]	
	\vcenter{
	\infer[\ilet]{\Lambda \ovdash{n_1+n_2+m} \letin{x}{u}{t} : \iA 
		{ \dftC{\sum_\bZ \psi_1 \odot \psi_2}{(\bY_1\cup\bY_2) - \bZ} } }{ \Lambda \ovdash{n_1} u  : \iP { \dftC{\psi_1}{\bY_1} } 
		&  \Lambda, x :  \iP \ovdash{n_2} t : \iA { \dftC{\psi_2}{\bY_2} } 
		&   \scriptstyle{\bZ = (\bY_1\cup \bY_2) - \Nm{\Lambda,  \iA} }
	}}
	\quad
	\violet{
	\small
	m =	
	\begin{cases}
		2^{|\widehat{\bY_1} \cup \widehat{\bY_2}|}  & \text{if } \bY_1,\bY_2 \not = \emptyset  \\
		0  &  \text{otherwise}. 
	\end{cases}
	}
	\\[4pt]
	\infer[\ipair]{ \Lambda \ovdash{0} \pair{v}{w} : \iL_1 \otimes \iL_2  { \dftC{1}{ \emptyset}} }
	{ \Lambda \ovdash{0}{} v : \iL_1  { \dftC{1}{ \emptyset}} & \Lambda \ovdash{0} w :  \iL_2  { \dftC{1}{ \emptyset} } } 
	\qquad
	\infer[\iletp]{\Lambda \ovdash{n} \letp{\pair{x}{y}}{\val}{t} : \iA  { \dftC{\psi}{\mathbb{Y}} } }
	{ \Lambda \ovdash{0} \val :\iL_1 \otimes \iL_2 
		\dftC{1}{ \emptyset} 
		&  \Lambda, x : \iL_1, y : \iL_2 \ovdash{n} t : \iA { \dftC{\psi}{\mathbb{Y}} } 
	}
	\end{gather*}			
	\caption{First-order type system annotated with the \emph{cost} of computing the factor.}
	\label{fig:cost-system}
\end{figure}

\paragraph{Cost and Reduction.} 
The upper bound in \Cref{prop_uppercost} is invariant by reduction, 
because the number of \paxs  in a type derivation is invariant. 
On the other hand,  the cost of \emph{inductively} computing the semantics of a type derivation $\pi$  (\Cref{prop_inductive_cost})
is \emph{not} stable by reduction, because  the structure of the derivation changes,   and so 
the size $W$ of the largest factor to be inductively computed may grow or shrink. 
This is indeed the rationale behind the program transformations which correspond to the Variable Elimination algorithm performed on terms \cite{LAFI2023}.
Starting from a normal form, the efficiency may be improved via expansion (the reverse of reduction).

\paragraph{Cost of  Factors Product vs.   Matrices Product.}\label{sec:on_product} 
We  stress how  much the product of factors  \emph{differs} from the product of matrices. In this difference  lies the efficiency of a factors-based semantics w.r.t. to a categorical 
\cite{JacobsZ}  or relational \cite{EhrPagTas14,EhrhardT19} one, where a central role is played by a  product $\otimes$ which behaves as the tensor product of matrices.
	\begin{example}\label{ex:tensor}
	Let $t_1,t_2$ be \texttt{case} expressions, respectively encoding two CPTs $\phi^{\X_1} = \Pr(\X_1|\Y_1, \Y_2,\Y_3)$ 
	and $\phi^{\X_2} = \Pr(\X_2|\Y_1, \Y_2,\Y_3)$, where two distinct   \rvs ($\X_1$ and $\X_2$, respectively) are conditioned to the \emph{same} set of \rvs  $\Y_1, \Y_2,\Y_3$. 
	We can see each  $\phi^{\X_i}$ interpreting $t_i$ as a stochastic matrix (of size $2^4$). One   easily realizes that computing the tensor product  $\phi^{\X_1} \otimes \phi^{\X_2}$ of the two matrices, requires to \emph{compute and store}  $2^4\cdot 2^4=2^8$ entries. In contrast,  
	the factor product $\phi^{\X_1} \FProd \phi^{\X_2}$  computes $2^5$ entries. 
	Indeed, in  a categorical or relational model, to compute 
	the semantics of the  term $\pair {t_1}{t_2}$ will  (in general) pass via  $\phi^{\X_1} \otimes \phi^{\X_2}$. 
	On this basis, it is 
		easy to build a term which encodes a BN  over the 5 variables $\X_1, \X_2,\Y_1, \Y_2,\Y_3$, and  whose   inductive interpretation  (in a categorical or relational model) requires to compute and store $2^8$ values---one such a term is given  in the \Appendix. 
		This  is somehow weird---from the point of view of BNs---given that  the \emph{full} joint distribution over 5 variables  has size $2^5$. 
\end{example}



\section{Lifting the Type System to Call-by-Push-Value PCF}\label{sec:PCF}
To streamline the presentation, 
we  carried out  our analysis in the setting of the untyped $\lambda$-calculus, on top of which we have defined an intersection type system. As  mentioned in \Cref{sec:itypes}, all our methods can be lifted to a more
user friendly PCF-like sintax, similar  to that in \cite{EhrhardT19}. Terms of a call-by-push-value PCF are those of the $\lambda$-calculus  in \Cref{sec:calculus} plus some operations needed to handle natural numbers (including a constant $\underline{\mathtt{n}}$ for each  number $n$) and a fixed point combinator:
\[\begin{array}{rrcl}
	\textsc{Terms} &	\tm,\tmu &\grameq & \ldots\gpipe \mathtt{succ}\ \val \gpipe \mathtt{pred} \ \val \gpipe \mathtt{ifZero}\ \val\ \tm\ \tmtwo\gpipe \fix{\var}{\tm}\\[2pt]
	\textsc{Values} &	\val,\valtwo & \grameq &  \ldots \gpipe \underline{\mathtt{n}}
\end{array} \]
The reduction rules are as expected. As standard, PCF  types (represented in \red{red}) are  the simple types of \Cref{sec:calculus} extended with the ground type of natural numbers:
\[\begin{array}{lrcl}
	\textsc{Ground Types}& \red\sL, \red\sK & \grameq & \ldots\gpipe \red{\mathtt{N}}
\end{array}\]
The fact that PCF is a \emph{typed} calculus does \emph{not} mean that we can  get rid of intersection types (here represented in \blue{blue}) to give semantics to the terms. Indeed, we still need to keep track of the names which are generated during the computation.
As usual with PCF, ground \emph{intersection types} need to be extended  with a type constant $\overline{\mathtt{n}}$ for each natural number $n$:
\[\begin{array}{lrcl}
	\textsc{Ground Intersection Types}& \blue\sL, \blue\sK & \grameq & \ldots\gpipe \blue{\overline{\mathtt{n}}}
\end{array}\]
At this point we need to type PCF terms. A PCF term will come with \emph{two} types: a standard type, and an intersection type, as in~\cite{DBLP:journals/corr/abs-1104-0193,EhrPagTas14,Ehrhard16}. For example, the term $\sample{d}$ is typed as $\vdash\sample{d}:\red{\Bool}:\blue\X$. The idea is that in a judgment $\vdash\tm:\red A:\blue C$, the intersection type $\blue C$ refines the standard PCF type $\red A$. In the example, $\blue\X$ refines $\red{\mathtt{B}}$, being $\blue\X$ the name of a boolean variable. In Figure~\ref{fig:PCF}, we report a selection of the typing rules for this language. One can notice that the rules $\psample$, $\pcond$, and $\plet$ are obtained by overlapping the ones for simple types with the ones for intersection types. The typing rules which are specific to  PCF constructs are standard,  and have been adapted from~\cite{Ehrhard16}.
\begin{figure}[t]
	\[
	\infer[\psample]{ \Lambda \vdash {\sample{d}} : \red\Bool:\blue\X}{ \blue\X \not \in \Nm{\Lambda} } 
	\qquad
	\infer[\pcond]{ \Lambda, y_1 : \red\Bool:\blue{\Y_1}, \dots, y_n : \red\Bool:\blue{\Y_n} \vdash \CPT{y_1,\dots, y_n} : \red\Bool:\blue\X }
	{ \blue\X \not \in \{\blue{\Y_1},\dots,\blue{\Y_n}\} \tand \blue\X \not \in \Nm{\Lambda} }
	\]\\[-8pt]
	\[
	\infer[\plet]{ \Lambda, \Gamma_1\uplus\Gamma_2  \vdash \letin{x}{\tmu}{\tm} : \red\iA:\blue C }
	{ \Lambda, \Gamma_1 \vdash \tmu : \red\iP:\blue Q  &  \Lambda, \Gamma_2, x : \red\iP:\blue Q \vdash \tm : \red\iA:\blue C }
	\]\\[-8pt]
	\[
	\infer{\tjudg{\Lambda}{\underline{\mathtt{n}}}{\red{\mathtt{N}}}:\blue{\overline{\mathtt{n}}}}{}
	\qquad\quad \infer{\tjudg{\Lambda,\Gamma_1\uplus\Gamma_2}{\mathtt{ifZero}\ \val\ \tm\ \tmtwo}{\red A}:\blue C}{\tjudg{\Lambda, \Gamma_1}{\val}{\red{\mathtt{N}}:\blue{\overline{\mathtt{0}}}}  & \tjudg{\Lambda,\Gamma_2}{\tm}{\red A}:\blue C} \qquad\quad
	\infer{\tjudg{\Lambda,\Gamma_1 \uplus \Gamma_2}{\mathtt{ifZero}\ \val\ \tm\ \tmtwo}{\red A}:\blue C}{\tjudg{\Lambda,\Gamma_1}{\val}{\red{\mathtt{N}}:\blue{\overline{\mathtt{n+1}}}}  & \tjudg{\Lambda,\Gamma_2}{\tmtwo}{\red A}:\blue C}
	\]\\[-8pt]
	\[
	\infer{\Lambda,\tjudg{\biguplus_{i=0}^n\Gamma_i}{\fix{\var}{\tm}}{\red A}:\blue C}{\tjudg{\Lambda,\Gamma_0,x:\,\red{!A}:\blue{[C_1,\ldots,C_n]}}{\tm}{\red A}:\blue C & \left( \tjudg{\Lambda,\Gamma_i}{\fix{\var}{\tm}}{\red A}:\blue {C_i} \right)_{i=1}^n }
	\]
	\vspace{-16pt}
	\caption{The type system for call-by-push-value PCF (selected rules).}
	\vspace{-8pt}
	\label{fig:PCF}
\end{figure}

\section{Related Work}

Theoretical research on   functional PPLs, pioneered by \cite{Saheb-Djahromi78} is a very active area. Early investigation  \cite{KollerMP97, FriedmanKP98, PlessL01, Park03, RamseyP02, ParkPT05} has evolved in a large body of work, to support the growing development of 
software.  Landmark foundational work aiming at providing sound and compositional methods for  \emph{higher-order} Bayesian inference  includes \cite{BorgstromLGS15} and \cite{ScibiorKVSYCOMH18}. The former is  based on operational techniques, while the latter on  a modular denotational semantics. Our work  is inspired by  both. The importance of   capturing the data flow of \emph{higher-order} probabilistic programs, 
is   stressed in   \cite{CastellanP19,Paquet21}, which rely on event structures. Relevantly, 
\BNs are  at the core of  \citet{Paquet21} game semantics proposal,  motivated by  the fact  that---in practice---most modern inference engines do not manipulate directly the program syntax but rather some graph  representation of it. We share this view, and the goal of  providing reasoning tools for such implementations. We also mention \cite{GorinovaGSV22}, which introduces an information flow type system,  in a \emph{first-order}, imperative setting. 
As  already recalled in the introduction, the theory of Bayesian networks has  been investigated extensively from a categorical viewpoint by Jacobs and Zanasi~\cite{JacobsZ16,JacobsZ,JacobsKZ19,bookJacobs}.

The foundational studies based on \BNs which we have cited above focus on    expressiveness and compositionality. However, they do not take into account space and time consumption of  probabilistic reasoning, which is the very motivation for the introduction \cite{Pearl86} and development of BNs.
As a matter of fact, the cost of  \emph{actually computing the semantics} explodes when taking a categorical \cite{JacobsZ} or relational \cite{EhrPagTas14,EhrhardT19} approach, because of the   product \emph{$\otimes$}, which behaves as the tensor product of matrices (see \Cref{ex:tensor}); 
inductively computing the semantics of $n$ binary \rvs easily leads to intermediate 
computations whose size is much larger than $2^n$ (the size of the full joint distribution).
This fact has already been pointed out in a recent paper  by~\citet{EhrhardFP23}, which advocates the need for a new approach to quantitative semantics, more attentive to the resource consumption.
That paper  imports factors techniques into the setting of \emph{multiplicative linear logic}, essentially  a \emph{linear} $\lambda$-calculus with tuples (roughly, our first-order fragment)---the authors leave as an open challenge the treatment  of linear logic exponentials  (roughly, the calculus of \Cref{fig:simple_types}). Our framework, indeed, is  able to deal with 
a fully-fledged  $\lambda$-calculus, thanks to the intersection type system and the flow-graph techniques built on top of it. 
Actually, our factor-based semantics can be seen as an optimized version of semantics based on the weighted relational model, such as  \cite{LairdMMP13} and  Probabilistic Coherence Space \cite{DanosH02, DBLP:journals/iandc/DanosE11, EhrPagTas14, EhrhardT19}. Finally, we mention that 
\citet{DBLP:journals/pacmpl/ChiangMS23} also use techniques inspired from linear logic to provide a denotational semantics and exact inference procedures to recursive probabilistic programs.

%
%

The idea that   programming languages provide a way to overcome the limitations of
standard Bayesian networks  
has been actively propounded by \citet{KollerP97,PfefferK00}, which introduced a modeling framework   based on  Object-Oriented programming. 
This  approach has eventually led to the object-oriented language Figaro \cite{PfefferBook}.
		
\section{Conclusions}\label{sec:conclusion}

We have presented a \emph{higher-order} probabilistic programming language which  allows  for  the specification
of recursive probability models and hierarchical structures.  Higher-order programs are \emph{compiled into standard \BNs} operationally, via rewriting, and denotationally, via an intersection type system. The novelty of our contribution
is that (1) the compositional  semantics  is based on \emph{factors}, the very mathematical notion which is used to give semantics to \BNs, and which is the basis of \emph{exact} inference algorithms, and (2) our semantics and type system are \emph{resource-sensitive.} The notion of  \emph{resource} appears here in two forms: (1) we precisely track generation and sharing of \emph{random variables}, and (2) we account for the actual \emph{cost} of inference---the cost of computing the semantics of a typed term. We obtain these quantitative results relying  on advanced semantic techniques rooted into linear logic, intersection types, rewriting theory,  and Girard's geometry of interaction, which are here combined in a novel way. 



The fact that our semantics    is  factor-based  implies that  
standard algorithms for exact inference (which are acting on factors) can  be  applied. 
A further natural direction is to investigate inference  techniques   based on term transformations, in the spirit of~\cite{KollerMP97, LAFI2023}.


\begin{acks}
	The authors are in debt with  Thomas Ehrhard and Michele Pagani for many insightful discussions. We  thank  Ugo Dal Lago, Beniamino Accattoli,  Delia Kesner for their useful remarks.  We are also grateful to the anonymous referees, whose valuable comments have 
improved the presentation of this work. 
	This research was supported by the ANR project PPS: ANR-19-CE48-0014. The third author is also supported by the European Union’s Horizon 2020 research and
	innovation programme under the Marie Sklodowska-Curie
	grant agreement No 101034255.
	\end{acks}

\section*{Data Availability Statement}
Missing proofs and more examples are available in the \SLV{Technical Report \cite{long}}{Appendix}. Please notice that 
while in the paper we postpone  to \Cref{sec:evidence} the treatment of types and semantics for the \texttt{observe} construct, 
 the proofs in the \Appendix directly integrate it.

\bibliography{biblioBN}


\newpage

\appendix


\renewcommand{\pink}[1]{#1}   

\newcommand{\iTypes}{\iBang}
\newcommand{\iTypesObs}{\texttt{iTypes+}\xspace}

\section*{APPENDIX}
We include here  proofs  that have been omitted in the body of the article, and some more examples.

\subsection*{Reduction, in the  notation of Explicit Substitutions}
 Since let-terms can be quite heavy to manage, in the Appendix we  usually  employ \textbf{the notation of  explicit substitutions}, which allows for more concise proofs: 

	\[  \begin{array}{cll}
\tm \esub x {\tmu}	 & \defeq &  \letin x \tmu \tm,   \\
 \tm\esub{\pair\var\vartwo}{\val}     &   \defeq  & 	\letp{\pair\var\vartwo}{\val}{\tm}.
\end{array}\]

Substitution lists $\sctx$, evaluation contexts $\ss$, and root reduction rules can therefore be written as follows. As before:
\begin{itemize}
	\item $\shole  {\tm}\sctx$ stands for the term obtained from $\sctx$ by replacing the  hole $\holebag{\cdot}$ with $\tm$,
	\item  $\ss\hole{\tm}$ stands for the term obtained from $\ss$ by replacing the  hole $\ctxhole$ with $\tm$,
\end{itemize}
possibly capturing the free variables of $\tm$.

\[\begin{array}{r@{\hspace{.5cm}}rlll}
	\textsc{Substitution Lists} & \sctx & \grameq &  \shole{\cdot} \mid   \esub{\var}{\tmtwo}\sctx \gpipe \esub{\pair\var\vartwo}{\val}\sctx
\end{array}
\]
\[\begin{array}{lrcl}
	\textsc{Evaluation Contexts} & \ss & \grameq & \ctxhole \gpipe \ss\val \gpipe 
	{\tm}\esub {x}{\ss} \gpipe {\ss}\esub{x}{\tmu}
\end{array}\]
\[\begin{array}{rll@{\hspace{1cm}}rll}
	\multicolumn{6}{c}{\textsc{Root Rules}}\\
	\scp {\la\var\tm}\ \val & \mapsto_{\dB}   &  \scp {\tm\esub\var\val} &
	\tm\esub\var{ \scp {\val}}& \mapsto_{\dS}& \scp{\tm\isub\var\val} \\[2pt]
	\der  !\tm  &\mapsto_{\dbang}& \tm &	{\tm} \esub  {\pair x y}{\pair {\val}{\valtwo}}  &\mapsto_{\dpair}&  
	 {  \tm \esub x {\val}  \esub y {\valtwo}}
\end{array}
\]

Please notice that here we use a  formulation for the $ \mapsto_{\dB}$ and  $\mapsto_{\dpair}$ rules which is slightly different but \emph{equivalent} to that in \Cref{sec:operational}. The advantage (at the level of proofs) is that   actual  substitution is  now  carried by the $\mapsto_{\dS}$ rule only. 

\section{Complements to  Sect.~\ref{sec:calculus} (the Calculus)}\label{app:progress}
We write $\pi\dem_{\low}$ when $\pi$ contains first-order rules, only.
\begin{lemma*}[\ref{lem:progressSimple}, Progress]  
	Let $\tm$ be a $\Bang$-term in \emph{normal} form such that $\pi \dem \LL \vdash \tm : \sL$,
	where $\sL$ and all types in the context $\LL$ are ground. 
	Then $\pi$ only uses first-order rules , \ie $\pi \dem_\low  \LL \vdash \tm : \sL$.
\end{lemma*}

\begin{proof}
We will prove a more general result, namely that if $t$ is a $\Bang$-term in \emph{normal} form and $\pi \dem  \LL\vdash t :\sC$, where all types in the  context $\LL$ are ground, then the following properties hold:
\begin{enumerate} 
	\item if $\sC= \oc \sA$ then $\tm = \scp{!u} $;
	\item if $\sC= \sP \larrow \sA$ then $t = \scp{\lam x. \tmu}$; 
	\item if $\sC = \sL_1 \otimes \sL_2$ then $t = \scp{v}$;
	\item if $\sC = \sL$ then $t$ is a \linear term, \ie $\pi \dem_\low \LL \vdash t: \sL$.
\end{enumerate}
The proof proceeds by induction on the type derivation $\pi$, considering its last rule. 
\begin{itemize}	
	\item Case $\rvar$ is immediate. Since we assume that the typing context only contains ground types, necessarily: 
	\[
	\pi \dem \LL', x : \sL \vdash x: \sL
	\]
		
	\item Case $\rpair$. Immediate by \ih: 
	\[
	\infer[\rpair]{ \pi \dem \LL \vdash \pair {\val_1}{\val_2} : \sL_1 \otimes \sL_2}
	{ \LL \vdash \val_1 :  \sL_1 & \LL \vdash \val_2 : \sL_2} 
	\]	
		
	\item Case $\rsample$ and $\rcond$  are immediate.	
	
	\item 	 {Case $\robs$ is  immediate.	}
		
	\item Case $\rbang$ is immediate:
	\[
	\infer[\rbang]{ \pi \dem \LL \vdash \oc \tmu :  \oc \sA}{ \LL \vdash \tmu: \sA  }
	\]	
		
		
	\item Case $\rabs$. Immediate: 
	\[
	\infer[\rabs]{ \pi \dem \LL \vdash \lambda x.u :  \sP \larrow \sA }{\LL, x : \sP \vdash u :  \sA}
	\]
		
	\item Case $\rlet$. The derivation $\pi$ has shape:
	\[ 
	\infer[\rlet]{\pi \dem \LL \vdash u \esub{x}{s}  : \sC }{ \LL \vdash s : \sP & \LL, x : \sP \vdash u : \sC }
	\]
	Observe that $\sP$ is necessarily a ground type. Indeed, $\sP = \oc\sA$ is not possible,
	because by \ih we would have $s = \scp{\oc r}$, making $u \esub{x}{s} $ a redex. 
	By \ih, all the claims hold for both $u$ and $s$, hence they hold for $ u \esub{x}{s}$.
			
	\item Case $\rletp$. Necessarily, $t \defeq u\esub{\pair x y }{v} $ in normal form
	implies $v = z$ (a single variable), because $u \esub{\pair x y }{\pair {v_1} {v_2}} $ would be a redex. So we have	
	\[
	\infer[\rletp]{ \pi \dem \LL \vdash u \esub{\pair x y }{z} : \sC }
	{ \infer[\rvar]{\LL \vdash z : \sL_1 \otimes \sL_2}{}  &
				\LL, x : \sL_1, y : \sL_2 \vdash u : \sC }
	\]
	By \ih, the claim holds for $u$, and therefore for $ u \esub {\pair x y}{z} $.  		
	
\end{itemize}		
The following cases never apply  because of the assumption that $t$ is normal:	

\begin{itemize}

	\item Case $\rapp$.  We have 
	\[
	\infer[\rapp]{\pi \dem \LL \vdash u v : \sC }{\LL \vdash u : \sP \larrow \sC &  \LL \vdash v : \sP}
	\]
	where by \ih $u=\scp{\lam x.s}$, which is not possible, because $u v$ would be a redex.
		
	\item Case $\rder$. This case does not apply, because by \ih we would have $u = \oc s$, making $\der u $ a redex.
	\[
	\infer[\rder]{\pi \dem \LL \vdash \der u : \sA}{\LL \vdash u : \oc \sA}
	\]	
\end{itemize}
\end{proof}


%

By Progress and the fact that typable terms are strongly normalizing, we have
\begin{corollary}
	If $\pi \rhd \LL \vdash t : \sL$, then $\tm  \to^* \tmu$, where $\tmu$ is a \linear term in normal form (with  type derivation 
	$\pi' \dem_\low \LL \vdash u : \sL$).
\end{corollary}

\section{Complements to  Sect.~\ref{sec:itypesSEC} (the Type System)}

\condinc{}{
\subsection{Typed names, more examples}
We give also a second example to illustrate as tracking  random variables become quickly involved even at the first order.
	\begin{example}\label{ex:names2} Consider the term
		\[	\tm \defeq \letin a {\sample d} {\letin h a {\letin b {\CPTN2 h}  {\letin c {\CPTN3 a} {\pattern {a,c, h}}}}}: \Bool \times \Bool \times \Bool\]
		%
		Once again, 	simple type are not informative enough.  We know that the term defines three \rvs $\X_1,\X_2,\X_3$, but the dependencies among them are not immediate.  Again, we track the random variables by naming the types.
				
					
		We can easily obtain a named derivation 
		\begin{center}
			$t:  \Bool_{\X_1} \otimes \Bool_{\X_3} \otimes \Bool_{\X_1}$, 
		\end{center}
		Now we    see that the output is a marginal distribution over the two  variables $\X_1,\X_2$, only. From the probabilistic axioms, we can also easily 
		read the dependencies: $\X_1\rightarrow \X_2, \X_1\rightarrow \X_3$. 
		
	\end{example}

}

\subsection{The Full Type System}
For convenience, we give explicitly the full type system, here denoted \iTypesObs,  which takes into account also  the construct $\obsb x$,
as we have described in \Cref{sec:evidence}.

 The grammar of types is as follows, where  $\X\in \Names$:
\[\begin{array}{lrcl}
	\textsc{Atomic Types}& \XX	 & \grameq&  \X \gpipe   { \Xt } \gpipe { \Xf}  \\[4pt]
	\textsc{Ground Types}&	 \iK, \iL & \grameq& { \XX}  \gpipe \iK\otimes \iL  \\[4pt]
	\textsc{Positive Types}&	 \iP,\iQ & \grameq & L \gpipe\, \mset{A_1,\ldots,A_n}  \\[4pt]
	\textsc{Types}&\iA,\iB &\grameq&  \sP \gpipe \iP \larrow \iA  
\end{array} \]

The typing rules for the full system \iTypesObs are in \Cref{fig:iTypes_obs}. 
\begin{figure}[H]
	\small
	\begin{mybox}		
		\textbf{Higher-order calculus}
		\begin{myboxC}
			\textbf{First-order rules}
			\begin{gather*}
				\infer[\isample]{ \Lambda \vdash {\sample{d}} : \XX}{ \X \not \in \Nm{\Lambda} } 
				\qquad
				\infer[\icond]{ \Lambda, y_1 : \YY_1, \dots, y_n : \YY_n \vdash \CPT{y_1,\dots, y_n} : \XX }
				{ \X \not \in \{\Y_1,\dots,\Y_n\} \tand \X \not \in \Nm{\Lambda} } 
				\\[4pt]	
					\infer[\iobs]{ \Lambda, x:\Xb  \vdash \obsb{x} : \Xb}{\quad \quad \bv \in \{\true,\false\}} 
				\qquad
				\infer[\ivar]{ \Lambda, x : \iP \vdash x : \iP}{}   	
				\qquad
				\infer[\ilet]{ \Lambda, \Gamma_1\uplus\Gamma_2  \vdash \letin{x}{\tmu}{\tm} : \iA }
				{ \Lambda, \Gamma_1 \vdash \tmu : \iP  &  \Lambda, \Gamma_2, x : \iP \vdash \tm : \iA }
				\\[4pt]
				\infer[\ipair]{ \Lambda \vdash \pair{v}{w} :  \iL_1 \otimes \iL_2 }
				{ \Lambda \vdash v : \iL_1  &  \Lambda  \vdash w :\iL_2 } 
				\qquad
				\infer[\iletp]{ \Lambda,   \Gamma  \vdash \letp{\pair x y }{\val}{t} :  \iA }
				{ \Lambda \vdash \val : \iL_1 \otimes \iL_2 
					& \Lambda, x : \iL_1, y : \iL_2, \Gamma \vdash t : \iA 
				}
			\end{gather*}
		\end{myboxC}
		\vskip -.4cm
		\begin{gather*}
			\infer[\iabs]{ \Lambda, \Gamma \vdash  \lam x.t : \iP \larrow \iA }{ \Lambda, \Gamma, x : \iP \vdash t :  \iA}
			\qquad
			\infer[\iapp]{ \Lambda, \Gamma_1\uplus\Gamma_2 \vdash tv : \iA }
			{ \Lambda, \Gamma_1 \vdash t : \iP \larrow \iA  &  \Lambda, \Gamma_2 \vdash v : \iP}
			\\[4pt]
			\infer[\ibang]{  \Lambda,   \biguplus_i\Gamma_i \vdash \oc \tm : \mset{\iA_1,\dots,\iA_n} }{ \big( \Lambda, \Gamma_i \vdash \tm: \iA_i \big)_{i = 1}^n  }
			\qquad
			\infer[\ider]{ \Lambda,   \Gamma \vdash \der \val :  \iA}
			{ \Lambda, \Gamma \vdash v : \mset \iA}
		\end{gather*}
	\end{mybox}
	\caption{The full intersection type system:  \iTypesObs }
	\label{fig:iTypes_obs}
\end{figure}

\subsection{Properties of the Type System}
\paragraph{Subject Reduction, Subject Expansion, Termination and Prograss.} The extensive proofs of Subject Reduction and Subject Expansion  are  in \Cref{sec:subject-reduction} and \Cref{sec:subject-expansion}, respectively.  The statements are strengthened to obtain several other properties we need.
The proof of Progress is very similar to the proof in \Cref{app:progress}. 
Please notice that the results mentioned above hold both for the system in \Cref{fig:iTypes} and  for the full system \iTypesObs  in \Cref{fig:iTypes_obs}. The proofs are carried out in the latter. 

\paragraph{Unique type derivation.}
A separate treatment in  needed for the property of unique type derivation in \Cref{prop:uniqueness}: this property---in the way it is stated in \Cref{sec:itypesSEC}---only holds for the system in \Cref{fig:iTypes}, that is a system 
with \emph{no}  observed types, but \emph{not}  for the  full system (\Cref{fig:iTypes_obs}). 
In \Cref{sec:unique_general} we generalize  the   property of "unique derivation", providing a  form which holds   in general.

\subsection{Unique Derivation (\Cref{prop:uniqueness}), holding for the system in \Cref{fig:iTypes}}	

The interest  of \Cref{prop:uniqueness} (and its immediate corollary \Cref{cor:unique})  is that given a  term  $\tm$ in BN normal form,  then its \emph{semantics is uniquely   determined}, because  its  type derivation $\pi$ 
 is  uniquely determined. The same holds for any term $\tm$ which reduces to BN normal form.
In fact, if $\tm$ is closed, we 
 can  simply (and \emph{uniquely})  write $\sem \tm$ for $\sem \pi$.

To prove \Cref{prop:uniqueness} we  need a preliminary lemma.
\begin{lemma}\label{lem:uninf}
	Let  $\tm$ be a $\Low$-term, and $\Lambda$  a ground context. If there exists a type derivation 
	$\pi \dem_\low \tjudg{\Lambda}{\tm}{\iL}$, then it is unique.
\end{lemma}
\begin{proof}
	The proof proceeds by induction on $t$. 
	Remark that $\pi$ does not contain \iobs rules, nor observed types.
	\begin{itemize}
		\item Case $t = x$.
		Necessarily $\Lambda = \Lambda', x : \iL$, 
		therefore the derivation $\pi \dem \Lambda', x : \iL \vdash x : \iL$ is uniquely determined.
		\item Cases $t = \sample{d}$ and $t = \CPT{x_1,\dots,x_n}$ are immediate, 
		as the derivation $\pi$ is uniquely determined by the choice of $\Lambda$.
		\item Case $t = \pair{v_1}{v_2}$. 
		By \ih there are is at most one $\iL_i$ such that $\pi_i \dem \Lambda \vdash v_i : \iL_i$ ($i \in \{1,2\}$); 
		moreover, if they exists, $\pi_1$ and $\pi_2$ are both unique. 
		Then there is only one way to obtain the derivation $\pi \dem \Lambda \vdash \pair{v_1}{v_2} : \iL$, 
		where $\iL = \iL_1 \otimes \iL_2$, by rule \ipair.
		\item Case $t = \letin{x}{u}{s}.$ By \ih there is at most one $\iK$ such that $\pi_1 \dem \Lambda \vdash u : \iK$.
		Similarly, by \ih there is at most one $\iL$ such that $\pi_2 \dem \Lambda, x : \iK \vdash s : \iL$, 
		Moreover, if they exist, both $\pi_1$ and $\pi_2$ are unique. 
		Then there is only one way to obtain the derivation $\pi \dem \Lambda \vdash \letin{x}{u}{s} : \iL$, by rule \ilet.
		\item Case $t = \letp{\pair{y_1}{y_2}}{v}{u}$. Similar to the previous case. 
	\end{itemize}
\end{proof}


		\begin{prop*}[\ref{prop:uniqueness}, Unique Derivation]Let  $\tm$ be a $\Bang$-term, and $\Lambda$  a ground context. Then there exists at most one type derivation $\pi$ such that
		$\pi \dem \tjudg{\Lambda}{\tm}{\iL}$.
	\end{prop*}

\begin{proof}
Let us suppose $\tm$ is typable as $\pi \dem \tjudg{\Lambda}{\tm}{\iL}$. We prove that $\pi$ is unique. By \Cref{thm:S-soundness}, we have that $\tm$ is strongly normalizing. Then we argue by induction on the number of steps $n$ needed to reach its normal form.
\begin{itemize}
	\item If $n=0$, then $\tm$ is in normal form. In particular, by \Cref{thm:compileI}, we have that $\tm$ is $\low$-term and therefore by \Cref{lem:uninf} satisfies the claim.
	\item If $n>0$, then we have $\tm\to\tmtwo\to^{n-1}\tmthree$, where $\tmthree$ is in normal form. By \Cref{prop:subsconv} $\tm$ is typable as $\pi \dem \tjudg{\Lambda}{\tm}{\iL}$ if and only if $\tmtwo$ is typable as $\pi' \dem \tjudg{\Lambda}{\tmtwo}{\iL}$. Then, we apply the \ih to $\tmtwo$, obtaining that if there exists $\pi' \dem \tjudg{\Lambda}{\tmtwo}{\iL}$, then $\pi'$ is unique. We conclude by \Cref{thm:subject-expansion}.
\end{itemize}
\end{proof}

%
%

\newcommand{\general}{general\xspace}
\newcommand{\canonical}{canonical\xspace}
\newcommand{\ObsT}[1]{\mathtt{Obs}(#1)}

\newcommand{\gen}[1]{#1^*}
\newcommand{\moreg}{\succeq}
\newcommand{\lessg}{\preceq}

\newcommand{\skeleton}[1]{\widehat{#1}}

\subsection{Unique General Derivation, holding for the full  type system (with \texttt{observe})}\label{sec:unique_general}
When we consider  the full system  \iTypesObs (in \Cref{fig:iTypes_obs}), which includes \emph{observed types}, then Subject Reduction, Subject Expansion,  Progress,  and \Cref{thm:compileI} still hold, in the same formulation of \Cref{sec:TS_properties}. 
The \emph{unique derivation}  property in \Cref{prop:uniqueness} instead does not hold as stated. For example, the term $\sample d$ can be typed in three valid ways:
\[\vdash \sample d: \X   \qquad \vdash \sample d: \Xt  \qquad \vdash \sample d: \Xf\]

We notice however that  all the three derivations above have the same shape  $\vdash \sample d: \XX,$ for $\XX\in \{\X,\Xt,\Xf\}$, and that the type $\X$ can be seen as   ``\emph{more general}'' than $\Xt$ and $\Xf$. 
This remark turns out to be a  property of all derivations of ground type, and allows us to generalize \Cref{prop:uniqueness} in a way that holds for the full type system.

\subsubsection{Unique General Derivation}
 Recall that an atomic type $\XX$ may assume  two forms: either as observed type ($\Xt,\Xf$), or as    unobserved type $\X$.  Since the type system is \emph{syntax driven},  the  term itself  can stipulate that  some  types are observed, via the construct $\obsb x$, which is  reflected in a $\iobs$-rule. 
We prove that if a closed term $\tm$ is typable, it admits \emph{exactly one}  derivation $\pi$ which is \general, in the sense that  (intuitively) it contains \emph{no more information than  that  provided by the term} $\tm$ (\Cref{prop:general}). Moreover, all the other valid  type derivations for the same term $\tm$ are obtained as refinement of $\pi$ (\Cref{prop:most_general}).\\

 Let us first formalize the intuitive notion that
a type derivation is \general if it declares as observed  all and only the types that the term prescribes  as observed. 
\begin{definition}[General Derivation] A type derivation $\pi$ in  system $\iTypes^+$ is \emph{\general} if  each atomic type which appears in $\pi$ and which does no  
 type the subject of any $\iobs$ rule, is  \emph{unobserved}	 (\ie, has form $\X$ and not $\Xb$).
\end{definition}

The proposition below (that we prove in Sect.~\ref{sec:general_proof}) allows us to recover uniqueness. Please notice that any derivation in the type system of \Cref{fig:iTypes} is \general. This way, \Cref{prop:uniqueness} is a special case of 
\Cref{prop:general}.
  \begin{proposition}[Unique \general derivation]\label{prop:general} Let  $\tm$ be a closed  $\Bang$-term. In   system  \iTypesObs, 
	if there exists a  derivation 
	$\pi \dem \tjudg{}{\tm}{\iL}$, then there exists a \emph{unique} \emph{\general}  derivation
	$ \pi'\dem \tjudg{}{\tm}{{\iL'}}$. 
\end{proposition}

If $\tm$ is a closed term, and	$ \pi\dem \tjudg{}{\tm}{{\iL}}$ is its \general derivation, then any other  type derivation for $\tm$ 
is obtained by refinement, in the following sense.
We define an order relation  on atomic types: 
\begin{center}
	$\X \moreg  \Xb$ ($\X$ is \emph{more general} than $\Xb)$, for each $\X\in \Names$.
\end{center}
The order extends to all types, and to type derivations, in the natural  way, \emph{point-wise}. 

\begin{proposition}[Most \general derivation]\label{prop:most_general} Let  $\tm$ be a closed  $\Bang$-term and $ \pi\dem \tjudg{}{\tm}{{\iL}}$ its \general derivation. Then  $ \pi \moreg \pi'$ for any derivation 
	$\pi' \dem \tjudg{}{\tm}{\iL'}$.
\end{proposition}

 \subsubsection{Proof of  \Cref{prop:general}}\label{sec:general_proof}
 \paragraph{Existence.} We first establish the existence of a \general derivation, for every term which has a type derivation.
The following property is  immediate to check by inspecting the typing rules.

 \begin{lemma}[Generalization]\label{lem:generalization}Let $\pi \dem \tjudg{\Sigma}{\tm}{\iA}$ be  a type derivation, and let   $\ObsT{\pi}$ be the  set of the atomic types which occur as the subject of an $\iobs$ rule. 
 		 Replacing in $\pi$ all occurrences of  atom  $\Xb$  with $\X$, for each atom which appears in $\pi$ but does not belong to 
 		 $\ObsT{\pi}$ yields a valid type derivation, written $\gen \pi \dem \tjudg{\gen \Sigma}{ \tm}{\gen \iA}$.  Moreover, the derivation $\gen{\pi}$ is \emph{general}; we  call  $\gen \pi$  the \emph{generalization} of $\pi$.
 \end{lemma}

 \paragraph{Uniqueness.}  To prove that a closed term of ground type has  a \emph{unique} \general derivation, we rely on $\Cref{prop:uniqueness}$. We first need the  notion of  skeleton, and some technical lemmas.
 
	We write $\skeleton \iA$ for the type obtained form  $\iA$ by replacing  all occurrences of the atom $\Xb$ with $\X$, for each $\X\in \Names$. We write $ \skeleton \tm $ for the term obtained from $\tm$ by replacing all occurrence of subterm $\obsb x$ with $x$,
for each $x\in \mathcal{V}$.  
\begin{lemma}[Skeleton]
	Let $\pi \dem \tjudg{\Sigma}{\tm}{\iA}$ be  a type derivation.
	Replacing in $\pi$ each judgment $ \tjudg{\Sigma_u}{\tmu}{\iA_u}$ with the judgment $\tjudg{\skeleton {\Sigma_u}}{\skeleton \tmu}{\skeleton {\iA_u}}$ yields a type derivation  $\skeleton \pi \dem  \tjudg{\skeleton{\Sigma}}{\skeleton \tm}{\skeleton \iA}$, which we call the \emph{skeleton} of $\pi$.
\end{lemma}

By construction, all derivations with the same skeleton are \emph{identical}, modulo quotienting   $\X$ and $\Xb$.  In fact, we can be more precise. 
\begin{lemma}[Main]Let $\pi_1 \dem \tjudg{\Sigma_1}{\tm}{\iL_1}$ and $\pi_2 \dem \tjudg{\Sigma_2}{\tm}{\iL_2}$
be type derivations with the same skeleton $\skeleton {\pi_1} = \skeleton{\pi_2}$. The two derivations may only differ 
on the atoms which do not type an $\iobs$-rule. 
\end{lemma}
\begin{proof}
	By induction on the term $\tm$.
\end{proof}

\begin{corollary}Two type derivations with the same skeleton have the same generalization:
	\[\skeleton {\pi_1} = \skeleton {\pi_2} \Rightarrow \gen {\pi_1} = \gen {\pi_2}. \]
\end{corollary}

By using \Cref{prop:uniqueness}, we have 	
\begin{proposition}[Uniqueness] Let  $\tm$ be a closed $\Bang$-term. If  $\tm$ has two \emph{\general} derivations of ground type 
	$ {\pi_1}\dem \tm:L_1$ and $  {\pi_2} \dem \tm: L_2$ then   $ {\pi_1} =  {\pi_2}$.
\end{proposition}
\begin{proof}
By observing that the skeleton of a type derivation only uses the rules in \Cref{fig:iTypes}, and therefore satisfies  \Cref{prop:uniqueness}, we have that   all the \emph{ground type} derivations of the same term  have the same skeleton. Hence the claim (notice that if $\pi_i$ is general, then $\gen {\pi_i }=\pi_i$).
\end{proof}


\section{Complements to Sect.~\ref{sec:term-semantics} (the Semantics of Typed Terms)}

\subsection{Interpretation of the Probabilistic Axioms,  Formally } 
In the paper, we are a slightly  informal when describing how to associate a factor to a $\icond$ axiom. Let us 
 describe how to associate a \cpt to each \pax, formally, and what are the subtleties. 

\begin{itemize}
	\item To a \icoin axiom of shape:
	\[ \infer{ \Lambda\vdash {\sample{\dist}} : \X  \dmd { \ft{\phi}{\{\X\}}  }}{} \]
	we associate the factor $\phi$ over the singleton $\set{\X}$ such that, for all $\x \in \Val{\X}$:
	\[\phi (\x) =\dist(\x) \]
	\item To a \icond axiom of shape:
	\[
	\infer{ \Lambda, y_1 : \Y_1, \dots, y_n : \Y_n 
		\vdash \texttt{case}\;\langle y_1,\dots, y_n \rangle\;\texttt{of}\; \{\bvs \caseSym \sample{\dist_{\bvs}}\}_{\bvs\in\{\tmt,\tmf\}^n}
		: \X  \dmd{ \ft{\phi}{\bY \cup \set{\X}} } }
	{ \X \not \in \{\Y_1,\dots,\Y_n\} \tand \X \not \in \Nm{\Lambda} }
	\]
	we associate a factor $\phi$ over the set of names $\bY \cup \set{\X}$, where $\bY= \set{\Y_1,\dots,\Y_n}$. 
	Remark that $\Y_1, \dots, \Y_n$ need not to be pairwise distinct, hence $\bY$ is a set
	whose cardinality $|\bY| = k$ may be less than $n$. The factor $\phi$ is defined as follows, for all $\x \in \Val{\X}$ and $\ys\in  \Val{\bY}$ (please observe that $\ys$ is a tuple of $k$ elements):
	\[ \phi(\ys \x)= \dist_{\bvs}(\x)\]  
where	$\bvs = \langle  \Proj \ys {\Y_1}, \dots,  \Proj \ys {\Y_n}\rangle$, recalling that  (by \Cref{notation_BN})  $\, \Proj  \ys{\Y_i} \in  \Val {\Y_i}$  is the 
value that   $\Y_i$ assumes in $\ys \in \Val{\bY}$. 

 The subtlety in the  definition is that if the same names $\Y$ occurs  several times in the context (say, $\Y_1=\Y_2=\Y_4$) then the \emph{same value} will be repeated  several times in the tuple $\bvs$, at the corresponding positions.
We give two examples to clarify.

	
\end{itemize}

\begin{example}
Taking the \BN in \Cref{fig:BNrain} as reference, we provide an example of a \texttt{case} expression and the associated \cpt. In the term, $x : \nS$ and $y :\nR$ are two variables of suitable type.
{\small \[
	\begin{array}{ccc}
		\centering
		\arraycolsep=1.5pt
		\begin{array}{rl}
		\texttt{case} &\pair{x}{y}  \texttt{ of} \\
		 \{ &\pair{\true}{\true} \Rightarrow \coin{0.99} \\
			&\pair{\false}{\true} \Rightarrow \coin{0.7} \\
			&\pair{\true}{\false} \Rightarrow \coin{0.9} \\
			&\pair{\false}{\false}  \Rightarrow \coin{0.01} \ \} : \nW \\
		\end{array}
		&  \qquad \qquad  &
        \begin{array}{|ccc|c|}
            \hline	
            \nS  & \nR  & \nW  & \Pr(\nW | \nS,\nR) \\
            \hline
            \true	& \true  & \true	& 0.99 \\
            \true	& \true  & \false	& 0.01 \\
            \hline
            \false  & \true  & \true	& 0.7 \\
            \false	& \true  & \false	& 0.3 \\
            \hline
            \true	& \false & \true	& 0.9 \\
            \true	& \false & \false	& 0.1 \\
            \hline
            \false  & \false & \true	& 0.01 \\
            \false  & \false & \false	& 0.99 \\
            \hline
          \end{array}
	\end{array}
\]}
\end{example}
	
	\begin{example}Let us also give an example where the same name appear several times. 
If   we have 
$ y_1 : \Y,y_2:\Y \vdash \texttt{case}\;\langle y_1,y_2 \rangle\;\texttt{of}\; \{\bvs \caseSym \sample{\dist_{\bvs}}\}_{\bvs\in\{\tmt,\tmf\}^2}:\X $
the factor over the names $\{\Y,\X\}$ associates to $\pair \true \true \in \Val \Y \times \Val \X$ the value $\dist_{\true,\true}(\true)$,
to $\pair \false \true \in \Val \Y \times \Val \X$ the value $\dist_{\false,\false}(\true)$, and so on.
	\end{example}

\subsection{Invariance of the Semantics}
The factor semantics we have defined is invariant under reduction and expansion, as expected. 
This is due to the fact that probabilistic axioms are stable w.r.t. reduction and expansion.
	
\begin{theorem*}[\ref{thm:invariance}, Invariance]
	Let $\tm$ be a $\Bang$-term and $t \to u$. Then
	$\Lambda \vdash t : \iL \dmd{\psi}$ if and only if $\Lambda \vdash u: \iL \dmd{\psi}$.
\end{theorem*}

\begin{proof} 
	Invariance of the semantics is a consequence of \Cref{thm:core}, 
	together with the Subject Reduction and Subject Expansion (\Cref{sec:subject-reduction}, \Cref{sec:subject-expansion}): 
	in the proof of such properties it suffices to observe that $\Cpts{\pi} = \Cpts{\pi'}$. 
	Then $\sem \pi = \sem {\pi'}$ by \Cref{thm:core}.
\end{proof}

\begin{example}\label{ex:weakening}
To better understand the invariance of $\Cpts{\pi}$ under reduction, recall  that \emph{only values can be duplicated and deleted}. One can see the consequences of this fact by considering the two following terms, and their type derivation. Observe that the first term is in normal form, precisely because {only values can be deleted}. By contrast, the subterm $!\sample d$, which is going to be deleted, can only be assigned an empty type.
\[
\begin{array}{ccc}
	\bottomAlignProof
	\AxiomC{$\vdash  \sample{d} : \X \dft {\phi} {\set \X}$}
		\AxiomC{$x:\X, y : \Y \vdash y : \Y $ }
		\UnaryInfC{$x : \X\vdash \lambda y.y : \Y \arrow \Y  \dft \ftone \emptyset$}
	\BinaryInfC{$\vdash \letin{x}{ \sample{d}}{\lambda y.y : \Y \arrow \Y} \dft \ftone \emptyset  $}
	\DisplayProof 
	& \not\rightarrow & \quad \\[12pt]
	\bottomAlignProof
	\AxiomC{$\vdash \oc \sample{d} : [] \dft \ftone \emptyset$}
		\AxiomC{$y : \nY \vdash y : \nY $}
		\UnaryInfC{$x : [] \vdash \lambda y.y : \nY \arrow \nY  \dft \ftone \emptyset $}
	\BinaryInfC{$\vdash \letin{x}{\oc \sample{d}}{\lambda y.y : \nY \arrow \nY} \dft \ftone \emptyset  $}
	\DisplayProof 
	&  \rightarrow  & 
	\bottomAlignProof
	\AxiomC{$y : \nY \vdash y : \nY$}
	\UnaryInfC{$\vdash \lambda y.y : \nY \arrow \nY   \dft \ftone \emptyset$}
	\DisplayProof
\end{array}
\]

\end{example}

\condinc{}{
\subsection{Closed Terms}
\begin{remark}[Semantics of Closed Terms]
	By \Cref{prop:uniqueness}, 	each closed term  $\tjudg{}{t}\iL$ of ground type has a unique type derivation  $ \pi\dem \vdash t:\iL$. Therefore 
	we could simply  write $\sem t$ for $\sem \pi$.
	
	Notice moreover that  necessarily $\Nm{\sem{\tm}} = \Nm{\iL}$ and this way $\sem t$ is a factor over the set of names which appear in $\iL$.
\end{remark}
}
\subsection{Completion of the Semantics}\label{sec:completion}
 As briefly discussed in \Cref{sec:compositionality}, the reader may expect that the interpretation of a type derivation $\pi\dem J$ were a factor over $\Nm J$. For example, one would expect to interpret an identity axiom  $ \alpha $ as 
	\begin{center}
		$ y : \nY \vdash y : \nY  \dmd{ \ft  {\ftone}{\set \Y} } $   \qquad instead of \qquad
		$ y : \nY \vdash y : \nY  \dmd{ \ft  {\ftone}{\emptyset} } $
	\end{center}
	%
	The fact is that our semantics  focuses on  the \emph{probabilistic content} of a derivation    $\pi \dem J$, only. Please notice  that  the non-probabilistic information is already fully contained in the type judgment $J$ (because intersection types carry such  information). 
	Indeed, an interpretation of $\pi \dem J$ as a factor over $\Nm J$ is easily obtained by a form of \emph{completion},  which is simply a multiplication for the trivial factor: 
	\[\den \pi \defeq \sem \pi \FProd \ft \ftone {\Nm{J}}\]  In the case of the identity axiom, this gives  $ \den {\alpha}= \ft  {\ftone} {\emptyset} \FProd \ft  {\ftone} {\Y} =\ft  {\ftone} {\Y} $.
	This completion  commutes with composition. We give some more details below.

\condinc{}{As briefly discussed in \Cref{sec:compositionality}, one may want the interpretation of a type derivation $\pi \dem J$ 
to be a factor over all the names $\Nm{J}$ appearing in the conclusion. We can easily \emph{complete} $\sem \pi$
and obtain a factor over $\Nm J$, exploiting the fact that the types appearing in the conclusion contain all the information we need. }


\begin{definition}[Completion]
Given a type derivation $\pi \dem J$, the completion of $\sem \pi$ is the
factor over the set of names $\Nm J$ defined as:
\[\den \pi \defeq \sem \pi \Fprod \ft{\ftone}{\Nm J}\]
\end{definition}

\begin{example}[Completed factor]
The example below shows, respectively, the factor semantics and its completion for \ivar and \isample axioms.
\[
\begin{array}{ccc}
	\infer{\pi_1 \dem y:\Y \vdash y:\Y}{}  & \qquad\qquad & \infer{\pi_2 \dem y:\Y \vdash \coin{p} :\X}{}  \\[8pt]	
	{	\small
	\sem{\pi_1} = \ftone_\emptyset 
	\qquad \qquad
	\begin{array}{c|c}
		\Y & {\den{\pi_1}}_\y  \\
		\hline
		\false & 1 \\
		\true &  1\\
	\end{array}
	}
	& \qquad\qquad &
	{	\small
	\begin{array}{c|c}
			\X  &  {\sem{\pi_2}}_{\x }\\ \hline
			\false  &p \\
			\true  &(1-p)\\
	\end{array}
	\qquad \qquad
	\begin{array}{cc|c}
			\Y & \X &{\den{\pi_2}}_{\y\x } \\
			\hline
			\false &\false  &p \\
			\false & \true  &(1-p)\\
			\true &\false  &p \\
			\true & \true  &(1-p)\\
	\end{array}
	}
\end{array}
\]

\end{example}

\begin{example}[Completed derivation] 
	Below we annotate the same derivation $\pi$ with $\sem{\pi}$ (l.h.s.) 
	and with its completion $\den{\pi}$ (r.h.s.).
	{\small 
	\[	
	\bottomAlignProof
	\AxiomC{$\vdash  \sample{d} : \X \dmd {\ft{\phi}{\set \X} }$}
		\AxiomC{$x:\X, y : \Y \vdash y : \Y  \dmd { \ft{\ftone}{\emptyset} } $ }
		\UnaryInfC{$x : \X\vdash \lambda y.y : \Y \arrow \Y  \dmd { \ft{\ftone}{\emptyset} } $}
	\BinaryInfC{$\vdash \letin{x}{ \sample{d}}{\lambda y.y : \Y \arrow \Y}  \dmd { \ft{\ftone}{\emptyset} }  $}
	\DisplayProof 
	\quad\quad
	\bottomAlignProof
	\AxiomC{$\vdash  \sample{d} : \X \dmd { \ft {\phi} {\set \X} } $}
		\AxiomC{$x:\X, y : \nY \vdash y : \nY  \dmd { \ft{\ftone}{\set{\X, \Y}} } $ }
		\UnaryInfC{$x : \X\vdash \lambda y.y : \nY \arrow \nY  \dmd { \ft{\ftone}{\set {\X,\Y}} } $}
	\BinaryInfC{$\vdash \letin{x}{ \sample{d}}{\lambda y.y : \nY \arrow \nY} \dmd{ \ft{\ftone}{\set \Y} } $}
	\DisplayProof 
	\]}
\end{example}



\subsubsection{Bridging with the Relational Models.}
The   completed interpretation  is a key step in order  to bridge  the gap between our semantics and \emph{weighted relational models/probabilistic coherence spaces} such as \cite{LairdMMP13,EhrhardT19, EhrPagTas14}. 

The other observation we need is  the following, suggested by \Cref{ex:names1}. 
\begin{remark}[On Redundancy]\label{rem:redundancy}
Looking at	\Cref{ex:names1} 
	one notice an interesting fact, which then shines in our semantics: to memorize the probability distribution over the (four) tuples in  $\Bool_\X \otimes \Bool_\X$, we only need  the distribution over the (two) values of $\Bool_\X$, since (as one see in the example) reconstructing the former from the latter is trivial.  For example:
	\begin{itemize}
		\item $\Val {\Bool_\X} \to \Real$:\quad
		$\true \mapsto 0.2, \false\mapsto 0.8$.
		
		\item $\Val {\Bool_\X} \otimes \Val{\Bool_\X} \to \Real$: \quad
		$\pattern{\true, \true} \mapsto 0.2, \pattern{\false, \false}\mapsto 0.8, \pattern{\true, \false} \mapsto 0,
		\pattern{\false, \true} \mapsto 0$.
	\end{itemize}

\end{remark}				

A similar process allows us to retrieve the relational model  from the completed semantics.


\section{Complements to  Sect.~\ref{sec:flow} (the Flow Graph)}

\subsection{The Flow Graph}

To define the flow graph for a \emph{general} type derivation in $\iBang$, we need to specify what are the positions (\ie the vertexes), what are the edges, and how the edges are oriented. 

\paragraph{Orientation and Polarity.}
The edges of the flow graph are oriented according to the input/output  polarity of the positions: the orientation is upwards on inputs, and downwards on outputs. 

In a judgment of ground type 
$\Lambda \vdash t: \iL$, all occurrence of atoms in $ \Lambda $  are seen as input, and all occurrences of atoms in $\iL$
are seen as output. 
However, when considering a generic type derivation in $\iBang$, 
we need to  take into account arrow types: note for example that the occurrence $\Y$ is an input in the judgment 
$~ \vdash t:\Y \larrow \X$. The definition of polarity given below is rather standard in proof-theory.

\begin{definition}[Input and Output Polarity]
To each occurrence of atom $\XX\in \{\X,\X^\true,\X^\false\}$ in a judgment is associated either an input polarity (denoted $\up \XX$) or output popularity (denoted $\down \XX$),  according to the following definition.

\[
Pol(x:\iP_1, \dots, x_n:\iP_n \vdash t:\iA) \defeq x: \Up{\iP_1}, \dots, x_n: \Up{\iP_n} \vdash t:\Down{\iA}
\]

\begin{minipage}{.49\linewidth}
	\begin{align*}
	\Up \XX &\defeq \up \XX \\
	\Up {A\otimes B} &\defeq \Up A \otimes \Up B \\
	\Up {A\larrow B} &\defeq \Down A \larrow \Up B  \\
	\Up {\mset{A_1, \dots, A_n}} &\defeq \mset{\Up {A_1}, \dots, \Up {A_n}}  \\
	\end{align*}
\end{minipage}
\begin{minipage}{.49\linewidth}
	\begin{align*}
		\Down \XX &\defeq \down \XX \\
	 	\Down {A\otimes B} &\defeq \Down A \otimes \Down B  \\
		\Down {A\larrow B } &\defeq \Up A \larrow \Down B	\\
		\Down {\mset{A_1, \dots, A_n}} &\defeq \mset{\Down {A_1}, \dots, \Down {A_n}}\\
	\end{align*}
\end{minipage}

%

\end{definition}

\paragraph{Polarized Positions.}
Similarly to what we have already done for ground types, we indicate a specific occurrence of an atom inside a type $\iA$ by means of a (type) context, \ie types with a hole, as follows:
\[\begin{array}{lrcl}
	\textsc{Ground Type Ctxs} & \cK, \cL & \grameq&  \ctxhole \gpipe  \cK \otimes \iL \gpipe \iK \otimes \cL \\[4pt]
		\textsc{Positive Type Ctxs}&	 \cP,\cQ & \grameq & \cL \gpipe\, \mset{\iA_1,\ldots,\cA,\ldots \iA_n}  \\[4pt]
		\textsc{Type Ctxs}&\cA,\cB &\grameq &  \cP \gpipe \cP \larrow \iA \gpipe \iP \larrow \cA  
\end{array}\]
The definition can be extended to ground and exponential contexts in a straightforward way; we will use $\cLamnp{}{\XX}$ and $\cGamnp{}{\XX}$ to refer to a specific atom occurrence in $\Lambda$ and $\Gamma$, respectively.
We assume that each atom occurrence appearing in a sequent of a derivation $\pi$ is given a distinct label. 
We call such a label a \emph{position}; each position has the polarity of the corresponding atom.

\paragraph{Flow Graph.}
The flow graph associated to a type derivation $\pi$ is the \emph{directed graph} that has for \emph{vertexes} the positions of $\pi$, and \emph{edges} as indicated in \Cref{fig:flowB}. The \emph{orientation} of the edges is given by the polarity of the positions: each  edge enters and exits an\emph{ input   position $\XX^\uparrow$  going upwards} (\ie, going from the conclusion of a rule to its premises). Similarly, edges  enter and exit  an \emph{output position $\XX^\downarrow$  going downwards} (\ie, going from the premises to the conclusion).

\begin{figure}
	\[
	\infer[\isample]{ \Lambda \vdash \sample{d} : \XX }{ \X \not \in \Nm{\Lambda}}
	\qquad 
	\infer[\icond]{ \Lambda, y_1: \vnode0{\YY_1}, \dots, y_k : \vnode1{\YY_k} \vdash \CPT{y_1, \dots, y_k} :  \vnode2{\XX} }
		{ \X \not \in \{\Y_1,\dots,\Y_n\} \tand \X \not \in \Nm{\Lambda} }
	\bentdirflowedges{node0/node2/40, node1/node2/40} 
	\qquad 
	\]
	
	\[\infer[\iobs]{ \Lambda, x : {\vnode0{\nX}}^\bv \vdash \obsb{x} : {\vnode1{\nX}^\bv} }{}  
	\bentdirflowedges{node0/node1/40}\]
	
	\[	
	\infer[\ivar]{ \Lambda, x : \cPp{\vnode0{\XX}} \vdash x :  \cPp{\vnode1{\XX}} }{}  
	\bentflowedges{node0/node1/60} 
	\qquad 
	\infer[\ipair]{ \Lambda \vdash \pair{v}{w} : \cLnp{1}{\vnode2{\XX}} \otimes \cLnp{2}{\vnode3{\YY}} }
	{ \Lambda \vdash v : \cLnp{1}{\vnode0{\XX}} & & \Lambda \vdash w: \cLnp{2}{\vnode1{\YY}} }
	\dirflowedges{node0/node2, node1/node3} 
	\]
	
	\[
	\quad
	\infer[\iletp]{\Lambda, \Gamma \vdash \letp{\pair{x}{y}}{v}{t} : \cAp{\vnode5{\ZZ}} }
	{ \Lambda \vdash v : \cLnp{1}{\vnode0{\XX}} \otimes  \cLnp{2}{\vnode1{\YY}}  &
		 \Lambda, \Gamma, y :  \cLnp{2}{\vnode3{\YY}}, x : \cLnp{1}{\vnode2{\XX}} \vdash t : \cAp{\vnode6{\ZZ}} }
	\specdirflowedge{node0}{node2}{to path={ .. controls +(.5,-.5) and +(-.5,-0.5) .. (\tikztotarget) } }
	\specdirflowedge{node1}{node3}{to path={ .. controls +(.3,-.3) and +(-.3,-0.3) .. (\tikztotarget) } }
	\flowedge{node5}{node6}
	\]	
	
	\[
	\infer[\ilet]{\Lambda, \Gamma_1 \uplus \Gamma_2 \vdash \letin{x}{u}{t} : \cAp{\vnode2{\YY}} }
	{ \Lambda, \Gamma_1 \vdash u : \cPp{\vnode0{\XX}}  &  \Lambda, \Gamma_2, x : \cPp{\vnode1{\XX}} \vdash t : \cAp{\vnode3{\YY}} }
	\flowedge{node2}{node3}
	\specflowedge{node0}{node1}{to path={ .. controls +(.5,-.5) and +(-.5,-0.5) .. (\tikztotarget) } }
	\]
	
	\[
	\infer[\iabs]{\Lambda, \Gamma \vdash \lambda x.t : \cPp{\vnode0{\XX}} \larrow \cAp{\vnode1{\YY}} }
	{\Lambda, \Gamma, x : \cPp{\vnode2{\XX}} \vdash t : \cAp{\vnode3{\YY}} }
	\flowedge{node0}{node2}
	\flowedge{node1}{node3}
	\qquad 
	\infer[\iapp]{\Lambda, \Gamma_1 \uplus \Gamma_2  \vdash t u : \cAp{\vnode0{\YY}} }
	{\Lambda, \Gamma_1 \vdash t : \cPp{\vnode1{\XX}} \larrow  \cAp{\vnode2{\YY}}  &  \Lambda, \Gamma_2 \vdash u : \cPp{\vnode3{\XX}} }
	\specflowedge{node1}{node3}{to path={ .. controls +(.5,-.5) and +(-.5,-0.5) .. (\tikztotarget) } }
	\flowedge{node0}{node2}
	\]
	
	\[
	\infer[\ibang]{\Lambda, \biguplus_{i=1}^n \Gamma_i \vdash \oc t : [ \cAnp{1}{\vnode0{\XX_1}}, \dots, \cAnp{n}{\vnode1{\XX_n}} ] }
	{\Lambda, \Gamma_1 \vdash t :  \cAnp{1}{\vnode2{\XX_1}}  & \dots &  \Lambda, \Gamma_n \vdash t : \cAnp{n}{\vnode3{\XX_n}} }
	\flowedges{node0/node2, node1/node3}
	\qquad
	\infer[\ider]{\Lambda, \Gamma \vdash \der t : \cAp{\vnode0{\XX}} }{\Lambda, \Gamma \vdash t : [ \cAp{\vnode1{\XX}} ] }
	\flowedge{node0}{node1}
	\]
	 
	\medskip
	\textbf{Contexts:}
	\medskip
	\[
	\infer{\cLamnp{}{\vnode0{\XX}}, \cGamnp{1}{\vnode3{\YY_1}} \uplus \dots \uplus \cGamnp{n}{\vnode4{\YY_n}} \vdash t : \iA }
	{\cLamnp{}{\vnode1{\XX}}, \cGamnp{1}{\vnode5{\YY_1}} \vdash t_1 : \iA_1  & \dots &
		\cLamnp{}{\vnode2{\XX}}, \cGamnp{n}{\vnode6{\YY_n}} \vdash t_n : \iA_n}
	\dirflowedges{node0/node1, node0/node2}
	\flowedges{node3/node5, node4/node6}
	\]
	\caption{Flow graph of a $\iBang$ derivation.} \label{fig:flowB}
\end{figure}

\begin{prop*}[\ref{prop:DAG}, The Flow is Acyclic]
	Let $\pi \dem \Lambda \vdash t:\iL$. Then $\flow{\pi}$ is a DAG.
\end{prop*}

\begin{proof}
	If $\tm$ is in  normal form, it is immediate to verify, 
	by induction on the derivation $\pi$, that $\flow{\pi}$ is \emph{acyclic}.
	If $\tm$ is not in normal form, by  \Cref{thm:compileI} we know that there
	exists a term $\tmtwo$ in  normal form such that $\tm\to^*\tmtwo$. 
	Therefore it is enough to prove that cycles are preserved along a 
	reduction sequence; this can be done by strengthening the 
	subject reduction statement (see \Cref{sec:subject-reduction}).
\end{proof}

\subsection{From DAGs to Compositionality}\label{sec:jointree}

We say that a position $p$ has underlying \pink{name  $\X$} (resp. underlying atom $\XX$) if $p$ denotes an occurrence of the atom  $\XX\in \{\X,\Xt,\Xf\}$. The main position of a probabilistic axiom is the one corresponding to the type of the subject.

\begin{lemma*}[\ref{lem:jointree}, Named paths] 
	Let $\pi \dem \Lambda \vdash t: \iL$, and \pink{let $\X$  be the main name   of a \pax $\alpha^\X$}. Then  in 
	$\flow{\pi}$,  each position with underlying  name  $\X$ is  connected to the  main position in $\alpha^\X$ by a path in which all  positions have underlying name $\X$.
\end{lemma*}

\begin{proof}Let us start with two easy observations. Notice that both  also holds by replacing "underlying name" with "underlying atom"
	\begin{enumerate}
		\item 
	 By inspecting  \Cref{fig:flowB} and taking into account the orientation of the edges , we realize  that in
	$\flow{\pi}$ each position has \emph{at most one parent with the same underlying name (resp., the same atom)}. Observe that 	the  only positions which do not have any parent with  the same name (resp., the same atom),  are either  the main position of a \pax, or those in the context of the conclusion $ \Lambda \vdash t:\iL$ (which have no parent at all).
	
\item From (1), 
	by recalling that a DAG where each node has
	at most one parent is a tree, we have the following key observation:  if we partition  $\flow{\pi}$ into maximal 
	connected subgraphs whose vertexes are \emph{positions with the same name $\X$},
	then each such subgraph  is a \emph{directed tree}, whose root corresponds  either 
	to the main type of a \pax, or to a position  in the context $\Lambda$ of the conclusion. 
	
	The same holds true if we partition  $\flow{\pi}$ into maximal 
	connected subgraphs whose vertexes are \emph{positions with the same atom  $\XX$},
		\end{enumerate}

	Now let $\nX$ be the main name of a \pax. One can easily verify that
	$\nX \not \in \Nm{\Lambda}$ (because ground contexts are additive,
	and because of the conditions on the axioms names). Since we require
	the main names of $\pi$  to be pairwise distinct, we conclude that
	in $\flow{\pi}$ there is \emph{exactly one} connected component of 
	positions with  the same name $\nX$.
\end{proof}
An immediate consequence is that if a name $\nX$ is summed out, then necessarily it is a main name.

We isolate two useful properties which we have shown as part of the proof of \Cref{lem:jointree}, because we will use in the following.
\begin{fact}[Main Names]	Let $\pi \dem \Lambda \vdash t: \iL$ be a type derivation, and 
	let $\nX$ be the main name of a \pax. One can easily verify that
	$\nX \not \in \Nm{\Lambda}$.
\end{fact}

\begin{fact}[Atoms] 
	Let $\pi \dem \Lambda \vdash t: \iL$ be a type derivation. Then  in 
	$\flow{\pi}$,  each position with underlying  atom  $\XX$ is  connected either to the  main type  $\XX$ of a probabilistic axiom or to an atom $\XX$  occurring in $\Lambda$   (in the conclusion of $\pi$), by a path in which all  positions have underlying atom $\XX$.
\end{fact}


\section{Complements to Sect.~\ref{sec:evidence} (Dealing with Evidence)}\label{app:evidence}

	Let us complete  the coin tosses example (\Cref{ex_coins2}), using   actual probabilities.  Assume that the tossed coin is  biased either  towards $\true$
	(bias $0.7$), or  towards false (bias $0.4$). Since we lack   knowledge, our prior belief  is that either coin could be tossed with equal probability. We want to infer the bias of the coin from the result of the tosses.  
\begin{example*}[\ref{ex_coins2},  cont.]
In    the term $\tmu$ of \Cref{ex_coins} modeling two tosses of the same (biased) coin 
	 \[	\tmu ~\defeq ~ \letin{x}{\sample{d}}{\letin{y}{\oc(\CPTN{}{x})}{\pair{x}{\der y, \der y}}}\]
 we  assume $\CPT{x}$ stands for the conditional expression $\case{\pattern{x}}{r_1\caseSym\coin{0.7}; r_2\caseSym\coin{0.4}}$ and $\sample d = \coin{0.5}$.
	
		Assume that   we observe that both  the  coin tosses yield $\true$, which is encoded as follows  (\Cref{ex_coins})
			\[\tmu' \defeq \letpin{x}{y_1,y_2}{\tmu}{\pair{x}{\obst{y_1},\obst{y_2}}} \]
Let us compare the 	 type derivation  $\pi$  for the term  $\tmu$  with the type derivation $\rho$ for the term $\tmu'$ 
and see \emph{how the semantics changes}.  In both cases, the semantics is  the product of the factors  associated to the  probabilistic 
axioms.
\begin{itemize}
	\item $\sem {\pi \dem \tmu: \X\otimes \Y_1 \otimes \Y_2}$. The    probabilistic 
	axioms in $\pi$ (see \Cref{ex_coins}) are the  following 
	\[ \vdash \sample{d} : \X \dmd \psi  \qquad\qquad  x : \X  \vdash \CPTN{}{x} : \Y_1 \dmd {\psi_1}\qquad\qquad  x : \nX  \vdash \CPTN{}{x} : \Y_2  \dmd{ \psi_2}\]
	where we have annotated the respective interpretations, which  are
		\[\psi\defeq
		\begin{array}{|c|c|}
			\hline	
			\X  & \Pr (\X) \\
			\hline	
			\true	& 0.5 \\
			\false	& 0.5 \\
			\hline
		\end{array}
\qquad\qquad
		\psi_i \defeq 
		\begin{array}{|cc|c|}
			\hline	
			\X  & \Y_i  & \Pr(\Y_i | \X) \\
			\hline	
			\true  & \true	& 0.7 \\
			\true  & \false	& 0.3 \\
			\hline
			\false & \true	& 0.4 \\
			\false & \false	& 0.6 \\
			\hline
		\end{array}\]

		\item $\sem {\rho \dem \tmu: \X\otimes \Yt_1 \otimes \Yt_2}$. The    probabilistic 
		axioms in $\rho$ (see \Cref{ex_coins2})are the  following three
		\[ \vdash \sample{d} : \X \dmd \psi  \qquad\qquad  x : \X  \vdash \CPTN{}{x} : \Yt_1 \dmd {\psi_1^\e}\qquad\qquad  x : \nX  \vdash \CPTN{}{x} : \Yt_2  \dmd{ \psi_2^\e}\]
		where 
		\[	
		\psi_i^\e =  
		\begin{array}{|cc|c|}
			\hline	
			\X & \Y_i  & \Pr(\Y_i=\true | \X) \\
			\hline	
			\true  & \true	& 0.7 \\
			\hline
			\false & \true	& 0.4 \\
			\hline
		\end{array}
	\]
	Proceeding like in \Cref{ex:evidence1}, $\sem \rho = \psi_1^\e \Fprod\psi_2^\e \Fprod\psi $ is exactly $\Pr(\X,\Y_1=\true, \Y_2=\true)$. From it, we have $\Pr(\Y_1=\true, \Y_2=\true) = 0.325$ and  $\Pr(\X|\Y_1=\true, \Y_2=\true)$.
	\[
		{ \psi_1^\e \Fprod\psi_2^\e \Fprod\psi } =
		\begin{array}{|ccc|c|}
			\hline	
			\X   & \Y_1  &  \Y_2  &  \Pr(\X,\Y_1=\true, \Y_2=\true) \\
			\hline	
			\true  & \true	&  \true	& 0.245 \\
			\hline
			\false & \true	&  \true	& 0.08 \\
			\hline
		\end{array}
		\qquad
		\begin{array}{|ccc|c|}
			\hline	
			\X  & \Y_1 &\Y_2  & \Pr(\X|\Y_1=\true, \Y_2=\true) \\
			\hline	
			\true  & \true & \true	& 0.753  \\
			\hline
			\false & \true	& \true & 0.246 \\
			\hline
		\end{array}
		\]
\end{itemize}			
Notice how 	the evidence has  increased  our confidence that the coins is biased towards $\true$ from $0.5$ to $0.753$.
\end{example*}

\condinc{}{
\RED{Variante 2}
\begin{example}[Coin tosses, cont.]
	Let us complete  the coin tosses example, using   actual probabilities. 	 For simplicity, we re-use the same 
	distributions as in \Cref{fig:paxs}.
	So   for  the term $\tmu$  modeling two tosses of the same (biased) coin 
	\[	\tmu ~\defeq ~ \letin{x}{\sample{d}}{\letin{y}{\oc(\CPTN{}{x})}{\pair{x}{\der y, \der y}}}\]
	we  assume $\CPT{x}$ stands for the conditional expression $\case{\pattern{x}}{r_1\caseSym\coin{0.7}; r_2\caseSym\coin{0.01}}$ and $\sample d = \coin{0.2}$. In words:
	\begin{itemize}
		\item the tossed coin either has bias $0.7$, or  is very heavily biased 
		towards $\false$   (bias $0.01$);
		\item our prior belief (perhaps due to previous bad experiences) is that the  coin which will be  tossed twice is the one  biased towards $\false$.  
	\end{itemize} 
	
	We want to infer the bias of the coin from the result of the tosses.  
	Assume that   we observe that both  the  coin tosses yield $\true$, which is encoded as follows
	\[\tmu' \defeq \letpin{x}{y_1,y_2}{\tmu}{\pair{x}{\obst{y_1},\obst{y_2}}} \]
	Let us compare the 	 type derivation  $\pi$  for $\tmu$ (given in \Cref{ex_coins}) with the type derivation $\rho$ for the term $\tmu'$ (given in \Cref{ex_coins2})
	and see how the semantics changes.  In both cases, the semantics is  the product of the factors  associated to the  probabilistic 
	axioms.
	
	\begin{itemize}
		\item $\sem \pi$. The    probabilistic 
		axioms in $\pi$ are the  following three
		\[ \vdash \sample{d} : \X \dmd \psi  \qquad\qquad  x : \X  \vdash \CPTN{}{x} : \Y_1 \dmd {\psi_1}\qquad\qquad  x : \nX  \vdash \CPTN{}{x} : \Y_2  \dmd{ \psi_2}\]
		where the respective interpretations are
		\[\psi\defeq
		\begin{array}{|c|c|}
			\hline	
			\X  & \Pr (\X) \\
			\hline	
			\true	& 0.2 \\
			\false	& 0.8 \\
			\hline
		\end{array}
		\qquad\qquad
		\psi_i \defeq 
		\begin{array}{|cc|c|}
			\hline	
			\X  & \Y_i  & \Pr(\Y_i | \X) \\
			\hline	
			\true  & \true	& 0.7 \\
			\true  & \false	& 0.3 \\
			\hline
			\false & \true	& 0.01 \\
			\false & \false	& 0.99 \\
			\hline
		\end{array}\]

		\item $\sem \rho$. The    probabilistic 
		axioms in $\rho$ are the  following three
		\[ \vdash \sample{d} : \X \dmd \psi  \qquad\qquad  x : \X  \vdash \CPTN{}{x} : \Yt_1 \dmd {\psi_1^\e}\qquad\qquad  x : \nX  \vdash \CPTN{}{x} : \Yt_2  \dmd{ \psi_2^\e}\]
		where 
		\[	
		\psi_i^\e =  
		\begin{array}{|cc|c|}
			\hline	
			\X & \Y_i  & \Pr(\Y_i=\true | \X) \\
			\hline	
			\true  & \true	& 0.7 \\
			\hline
			\false & \true	& 0.01 \\
			\hline
		\end{array}
		\]
		Proceeding like in \Cref{ex:evidence1} we obtain $\Pr(\X,\Y_1=\true, \Y_2=\true)$. From it we have $\Pr(\Y_1=\true, \Y_2=\true) = 0.98008$ and  $\Pr(\X|\Y_1=\true, \Y_2=\true)$.
		\[
		{ \psi_1^\e \Fprod\psi_2^\e \Fprod\psi } =
		\begin{array}{|ccc|c|}
			\hline	
			\X   & \Y_1  &  \Y_2  &  \Pr(\X,\Y_1=\true, \Y_2=\true) \\
			\hline	
			\true  & \true	&  \true	& 0.98 \\
			\hline
			\false & \true	&  \true	& 0.00008 \\
			\hline
		\end{array}
		\qquad
		\begin{array}{|ccc|c|}
			\hline	
			\X  & \Y_1 &\Y_2  & \Pr(\X|\Y_1=\true, \Y_2=\true) \\
			\hline	
			\true  & \true & \true	& 0.999...  \\
			\hline
			\false & \true	& \true & 0.0... \\
			\hline
		\end{array}
		\]
	\end{itemize}			
	The semantics $\sem \rho$ is exactly $\Pr(\X, \Y_1=\true, \Y_2=\true)$. Notice how 			
	the evidence flips our prior belief! 
\end{example}
}



\newcommand{\fct}[1]{\mathit{fcts}(#1)}
\section{Complements to Sect.~\ref{sec:cost} (Cost)}
In \Cref{fig:cost_derivations} we give the cost-annotated type derivation for the terms in \Cref{ex:cost}, according to the 
system in 
\Cref{fig:cost-system}.

\begin{figure}	[H]
	\small
	\fbox{
	\begin{minipage}{\textwidth}
	\begin{prooftree}
	\AxiomC{$ \ovdash{0} \sample{d} : \nX \dmdC{ \{\nX\} } $}
		\AxiomC{$x : \nX \ovdash{0} \CPTN{\Y}{x} : \nY \dmdC{ \{\nX,\nY\} }$}
			\AxiomC{$y : \nY \ovdash{0} \CPTN{\Z}{y} : \nZ \dmdC{ \{\nY,\nZ\} }$}
				\AxiomC{$z : \nZ \ovdash{0} z : \nZ \dmdC{ \emptyset }$}
			\BinaryInfC{$ y: \nY \ovdash{0} \letin{z}{\CPTN{\Z}{y}}{z} : \nZ \dmdC{ \{\nY,\nZ\} }$}
		\BinaryInfC{$x : \nX \ovdash{ 2^{|\{\nX,\nY,\nZ\}|}} \letin{y}{\CPTN{\Y}{x}}{\letin{z}{\CPTN{\Z}{y}}{z}} : \nZ \dmdC{ \{\nX,\nZ\} }$}		\BinaryInfC{$ \ovdash{2^3+2^{|\{\nX,\nY\}|}} \letin{x}{\sample{d}}{\letin{y}{\CPTN{\Y}{x}}{\letin{z}{\CPTN{\Z}{y}}{z}}} : \nZ \dmdC{ \{\nZ\} } $}
	\end{prooftree}	
	\dotfill
	\begin{prooftree}
	\AxiomC{$ \ovdash{0} \sample{d} : \nX \dmdC{ \{\nX\} } $}
		\AxiomC{$x : \nX \ovdash{0} \CPTN{\Y}{x} : \nY \dmdC{ \{\nX,\nY\} }$}
	\BinaryInfC{$\ovdash{2^{|\{\nX,\nY\}|}} \letin{x}{\sample{d}}{\CPTN{\Y}{x}} : \nY \dmdC{ \{\nY\} }$}
		\AxiomC{$ y : \nY \ovdash{0} \CPTN{\Z}{y} : \nZ \dmdC{ \{\nY,\nZ\} } $}
	\BinaryInfC{$\ovdash{2^2+2^{|\{\nY,\nZ\}|}} \letin{y}{(\letin{x}{\sample{d}}{\CPTN{\Y}{x}})}{\CPTN{\Z}{y}} : \nZ \dmdC{ \{\nZ\} }$}
		\AxiomC{$z : \nZ \ovdash{0} z : \nZ \dmdC{ \emptyset }$}
	\BinaryInfC{$ \ovdash{2^2+2^2} \letin{z}{(\letin{y}{(\letin{x}{\sample{d}}{\CPTN{\Y}{x}})}{\CPTN{\Z}{y}})}{z} : \nZ \dmdC{ \{\nZ\} } $}
	\end{prooftree}
	\end{minipage}
	}
	\caption{Type derivations for the terms in \Cref{ex:cost}, cost-annotated.}
		\label{fig:cost_derivations}
\end{figure}

\subsection{Cost of  Inductively Computing the  Interpretation of a Derivation}\label{sec:inductive_cost}
We now analyze  the cost of inductively computing the semantics of $\pi$ following the structure of the derivation. 
We prove that  the inductive  computation admits  a better upper  bound than the one    in \Cref{prop_uppercost}.

\begin{prop*}[\ref{prop_inductive_cost} Inductive cost]
Let  $ \pi$ be a type derivation,  $m_\pi$  the number of \paxs in $\pi$, $n_\pi=|\Nm{\Cpts\pi}|$  the number of names which appear in $\Cpts\pi$ (as in \Cref{prop_uppercost}).
The cost of inductively computing the semantics of $\pi$ following the structure of the derivation is 
	\[\BigO{m_\pi \cdot 2^W}\]
	where 
	$W\leq n_\pi $ is  the maximal cardinality $|\bY|$ of any set of names $\bY$ appearing in the derivation when decorated as in \Cref{fig:decorated}.
\end{prop*}
\begin{remark}Notice that given a type derivation $\pi$, we can easily  establish $W=|\bY|$ (the larger size  of any factor $\ft{\psi}{\bY}$ which will be inductively computed)  without actually performing any factor computation. This is easily obtained by restricting 
the decoration in \Cref{fig:decorated}  to the set of names, only. 
\end{remark}

We prove \Cref{prop_inductive_cost} by counting separately the number of multiplications and additions needed to compute the semantics of $\pi$. In the rest of the section, we set $\pi \dem J_{\pi}$ to  be an arbitrary but fixed type derivation. 


\paragraph{Counting Multiplications.}
We first restrict our attention only to \emph{multiplications}, ignoring the cost of summing out. 
To each  judgment $J$ occurring in  $\pi$,  we assign a number  $\fct{J}$.  Intuitively,  $\fct{J}$ is  the maximal number of non-empty  factors we need to multiply to compute the semantics  associated to $J$, given the semantics of its premises. 
Notice that  the product for an empty factor ( $\phi \FProd \ftone$)  demands no actual computation.

\begin{itemize}
	\item Case  of $J$ axiom.
\begin{itemize}
	\item  $ \fct{J}=0$ if  $J$ is conclusion of $\ivar$
		\item $ \fct{J}=1$ if $J$ is conclusion of  $\isample \tor \icond$.
\end{itemize}
\item Case of $J$ conclusion of a rule with $h\geq 1$ premises: 
$  \infer[R] { J}{{ J_1} & \dots & { J_h}} $
\begin{itemize}
\item $ \fct{J} $ is the number of premises $J_i$ such that $\fct{J_i}\not= 0$.
\end{itemize}
\end{itemize}

 \begin{lemma}[ Nodes vs leaves, weighted]\label{lem:tree_nodes} Let  $ \pi\dem J_\pi$ be a type derivation,  and 
	$m_\pi$  the number of \paxs  in $\pi$. Let   \[r_{\pi} \defeq  \sum\limits_J
	(\fct{J}-1)\]  where $J$  ranges over all judgments  in $\pi$ such that $\fct{J}\not=0$. 
	Then
	\[
	\begin{array}{lll}
		r_\pi=0=m_\pi,   & \mbox{if }~\fct{J_{\pi}}=0\\
		r_\pi=m_\pi-1, & \mbox{otherwise}.
	\end{array}
	\]
\end{lemma}
\begin{proof}By  induction on the structure of  $\pi$. If $\pi$ consists of a single axiom, the property holds trivially:  for $\ivar$  $m_\pi=0=r_{\pi}$, for a \pax $m_\pi=1$ and $r_\pi=0$. 
	
	Assume $\pi\dem J_{\pi}$	 has last rule $R$:
	\[ \infer[R] {\pi \dem J_{\pi} }{\pi_1 \dem J_1 & \dots & \pi_h  \dem J_h} \]
	
	\begin{itemize}
		\item Case $\fct{J_{\pi}}=0$ is immediate, observing that   $\fct{J_{\pi}}=0$ iff $\fct{J_i}=0$ for each $i$, and so  $\fct{J_{\pi}}=0$  iff  $m_\pi=0$ (\ie, the derivation of $J_{\pi}$ contains no \pax).
		\item  Case $\fct{J_{\pi}}=k>0$. By assumption, there
		 are $k$ premises such that $\fct{J_i}\not= 0$.  By \ih, for each such premise  $r_{\pi_i}=m_{\pi_i}-1$.
		Hence $r_{\pi}= (\sum_{1}^k r_{\pi_i}) + (k-1)  = (\sum_{1}^k m_{\pi_i}) -k + (k-1) = m_{\pi}-1 $.
	\end{itemize}

\end{proof}

\begin{lemma}[Multiplications] With the same notations as above,  the total number of multiplications to compute the semantics of a derivation $\pi$ is \[\pink{\BigO{m_{\pi} \cdot 2^W}. }\]  
\end{lemma}
\begin{proof}Immediate consequence of \Cref{lem:tree_nodes}, \pink{recalling that to multiply  $k$ factors  requires  
	${(k-1)\cdot 2^\textsc{w}}$ multiplications}, where $\textsc{w}$ is the number of names in the resulting factor. 
	
 Assume we are inductively computing the semantics of $\pi$, following  \Cref{fig:decorated}. Let   $W\leq n_\pi$ be  the largest number of names involved in the computation of any factor. To compute the semantics of a sub-derivation $\pi'\dem J$  given the interpretation of its premises,
we  need to compute  the product of $\fct{J}$ factors, which has cost $0$ if $\fct{J}=0$, and otherwise has cost 
\pink{at most ${(\fct{J}-1) \cdot 2^W}$. }
So (using \Cref{lem:tree_nodes}) the total number of multiplications to compute the semantics of $\pi$ is bounded by 
%
%
\pink{\[{r_{\pi} \cdot 2^W}= {(m_\pi-1) \cdot 2^W} \]}
\end{proof}

\paragraph{Counting Additions.}
The cost of \pink{summing out any number of (binary) \rvs from a factor of size $W$} is $\BigO{2^W}$\footnote{Summing out observed variable does not imply any actual computation, because---as we have already remarked in \Cref{sec:cost}---the set $\bE$ of observed variable does not actually contribute to the number of entries in a factor:  $\Val {\widehat\bY} \times \Val \bE$ is isomorphic to $\Val {\widehat\bY}$.}. 
Notice that, when computing the semantics of $\pi$ following \Cref{fig:decorated}, we perform a sum out when (in a rule $\ilet$ or $\iapp$) some names appear in the premisses, but not in the conclusion. By the observations in \Cref{sec:jointree}, a name can be summed out at most once in the derivation $\pi$, and moreover, only main names are summed out.
\condinc{}{Includere lemma opportuno}
Therefore the number of possible sum outs is bounded by the number $m$ of \paxs in $\pi$. We immediately have the following. 
\begin{lemma}[Additions] With the same notations as above,  the total number of additions  to compute the semantics of a derivation $\pi$ is \[\BigO{m_\pi \cdot 2^W}. \]
\end{lemma}

\begin{remark}[Reduction and Cost of the Semantics]  The upper bound in \Cref{prop_uppercost} is invariant by reduction (because the number of \paxs axiom in a type derivation is invariant). Instead,  
	the cost of \emph{inductively} computing the semantics of a type derivation $\pi$ is \emph{not} stable by reduction, because as the derivation changes,   the size $W$ of the largest factor to be inductively computed may  grow or decrease. Notice that both cases are possible, so given a term $t$, the cost of inductively computing the semantics  of its normal form may be larger than the cost for $t$.
\end{remark}

\subsection{On the Product of Factors  vs. Product of  Matrices.}\label{app:on_product} 
We stress how  much the product of factors is \emph{different} from the product of matrices. In this difference  lies the efficiency of a factors-based semantics w.r.t. to a categorical 
\cite{JacobsZ}  or relational \cite{EhrPagTas14,EhrhardT19} one, where the  tensor  product $\otimes$ (which behaves as the tensor product of matrices) plays instead a central role.
\begin{example}[Factors are compact]\label{ex:product}
	Let $t_1,t_2$ be \texttt{case} expressions, respectively encoding two CPTs $\phi^{\X_1} = \Pr(\X_1|\Y_1, \Y_2,\Y_3)$ 
	and $\phi^{\X_2} = \Pr(\X_2|\Y_1, \Y_2,\Y_3)$  conditioned to the \emph{same} variables. 
	We can see each  $\phi^{\X_i}$ interpreting $t_i$ as a stochastic matrix (of size $2^4$). One   easily realizes that computing the tensor product of matrices $\phi^{\X_1} \otimes \phi^{\X_2}$ requires to \emph{compute and store}  $2^4\cdot 2^4=2^8$ entries, while 
	the factor product $\phi^{\X_1} \FProd \phi^{\X_2}$  computes $2^5$ entries. 
	In a categorical or relational model, to compute 
	the semantics of the   term $\pair {t_1}{t_2}$ will  (in general) require to compute $\phi^{\X_1} \otimes \phi^{\X_2}$. 
	{On this basis, it is 
		easy to build a term $t$ such as the following one
		{\small 	\begin{align*}
				\leti y_1 = & ~\sample{d_1} \inl~ \leti y_1'= y_1 \inl\\
				\leti y_2 =& ~\sample{d_2} \inl~ \leti y_2'= y_2 \inl\\
				\leti y_3 =& ~\sample{d_3} \inl~ \leti y_3'= y_3 \inl\\
				\leti x_1=& ~\CPTN{\X_1}{y_1,y_2,y_3}\inl\\
				\leti x_2= &~\CPTN{\X_2}{y_1',y_2',y_3'} \inl~\pattern {x_1,x_2}.
		\end{align*}}
		which encodes a BN   over the 5 variables $\X_1, \X_2,\Y_1, \Y_2,\Y_3$, where  computing the inductive interpretation 
		of $t$ in a categorical or relational model requires to compute and store $2^8$ values. 
		This  is somehow weird, since  the full joint distribution over 5 variables  has size $2^5$. 
		In contrast, as we have stressed in this section, 
		the cost of computing our  semantics is never larger  than the cost  of computing the joint distribution.
	}
\end{example}


\section{Subject Reduction} \label{sec:subject-reduction}
We strengthen Subject Reduction statement, in order to have the properties we need to prove also the strong
normalization of typable terms, the invariance of the factor-based semantics, and the acyclicity of the flow graph of a ground derivation. In the present section and in the following one we use the symbol $\mult$ to refer to a (possibly empty) multiset 
without explicitly naming its elements.

\begin{definition}
The measure of a derivation, denoted $\meas{\pi}$, is obtained by counting the number of rules \ilet, \ider, \iapp, \iletp occurring in $\pi$, giving them weight $1,1,2,3$ respectively. All the other rules have $0$ weight. 
\end{definition}

\begin{lemma} \label{lemma:basic_val}
If $\pi \dem \Lambda, \Gamma \vdash v : \iL$ then $\Gamma$ is an empty context and $\meas{\pi} = 0$. Moreover $\pi$ contains no \paxs, and $\flow{\pi}$ is acyclic. 
\end{lemma}

\begin{proof}
Since $v:\iL$,  then either $v=x$  or $v= \pair {v_1}{v_2}$. We reason by induction on $\pi$.
\begin{itemize}
	\item Case $v=x$. Then $\pi \dem \Lambda \vdash x : \iL$ where $\Lambda = \Lambda', x : \iL$, and the result immediately follows.
	\item Case $v=\pair{v_1}{v_2}$. By inductive hypothesis the property holds both for $\pi_1 \dem \Lambda \vdash v_1: \iL_1$ and $\pi_2 \dem \Lambda \vdash v_2: \iL_2$; then it immediately follows that $\pi \dem \Lambda \vdash v : \iL_1 \otimes \iL_2$ satisfies the property, as the flow goes towards the axioms in $\Lambda$, and towards the conclusion in $\iL_1 \otimes \iL_2$.
\end{itemize}
\end{proof}

\begin{lemma}[split] \label{lemma:split} 
Let $\pi \dem \Lambda, \Gamma \vdash v : \mult_1 \uplus \dots \uplus \mult_n$. Then there exist $\pi_i \dem \Lambda, \Gamma_i \vdash v : \mult_i $ ($1 \leq i \leq n$) such that $\Gamma = \biguplus_{i=1}^n \Gamma_i$ and $\meas{\pi} = \sum_{i=1}^n \meas{\pi_i}$.  Moreover, if  $\pi \dem \Lambda, \Gamma \vdash v :[]$, then $\Gamma$ is empty, $\meas{\pi} =0$ and $\flow{\pi}$ is acyclic.
\end{lemma}

\begin{proof}
By observing that the subject can be assigned a multiset type either by a rule \ivar or by a rule \ibang; in both cases we  can choose a suitable partition of the multiset. Note that $v = \oc u:[]$ can only be obtained by rule \ibang with $n = 0$ premises, hence $\Gamma$ is empty and $\flow{\pi}$ is trivially acyclic.
\end{proof}

\begin{lemma}[substitution lemma] \label{lemma:substitution}
Let  $\pi^v \dem \Lambda, \Gamma \vdash v : \iP$ and $\pi \dem \Lambda, \Delta, x : \iP \vdash t : \iA $. 
Then there exists $\pi' \dem \Lambda, \Gamma \uplus \Delta \vdash t \subs{x}{v} : \iA $ such that $\meas{\pi'} = \meas{\pi^v} + \meas{\pi}$.
Moreover, consider the rewriting:
\[
\infer[\ilet]{\pi^{\tlet}\dem \Lambda, \Gamma \uplus \Delta \vdash t \esub{x}{v} : \iA }
{\pi^v \dem \Lambda, \Gamma \vdash v : \iP & \pi \dem \Lambda, \Delta, x : \iP \vdash t : \iA }
\quad \rightsquigarrow \quad
\pi' \dem \Lambda, \Gamma \uplus \Delta \vdash t \subs{x}{v} : \iA
\]
The following properties hold:
\begin{enumerate}
\item $\pi^{\tlet}$ and $\pi'$ contain the same \paxs, modulo renaming of free variables;
\item if $\flow{\pi^{\tlet}}$ contains a directed cycle, then $\flow{\pi'}$ contains a directed cycle;
\item if there is a directed path between two positions $i$ and $j$ in the conclusion of $\pi^{\tlet}$, then there is a path between $i$ and $j$ in the conclusion of $\pi'$. 
\end{enumerate}
\end{lemma}

\begin{proof} By induction on the derivation $\pi$, considering its last rule. 
\begin{itemize}
\item Case \ivar. Two subcases:
\begin{enumerate}
\item If $t = x$, then  $\iA = \iP$, that is $\pi \dem \Lambda, x : \iP \vdash x : \iP $. Observe that $\Delta$ is empty. We have:
\[
\infer[\ilet]{\pi^{\tlet}\dem \Lambda, \Gamma \vdash x \esub{x}{v} : \vnode3{\iP} }
  {\pi^v \dem \Lambda, \Gamma \vdash v : \vnode0{\iP}  & 
     \infer[\ivar]{\pi \dem \Lambda, x : \vnode1{\iP } \vdash x : \vnode2{\iP} }{}}
\specflowedge{node0}{node1}{to path={ .. controls +(.5,-.5) and +(-.5,-0.5) .. (\tikztotarget) }}
\flowedge{node2}{node3}
\bentflowedges{node1/node2/55}
\quad \rightsquigarrow \quad
\pi' := \pi^v \dem \Lambda, \Gamma \vdash v:\iP
\]
Clearly $\meas{\pi'} = \meas{\pi^v} = \meas{\pi^v} + \meas{\pi}$ and the \paxs contained in $\pi^{\tlet}$ and $\pi'$ are the same. Lastly, even if we delete $\pi$, the flow remains essentially unaltered, as $\pi$ contains no cyclic path.
\item If $t=y\not=x$ then $\pi \dem \Lambda', x : \iP, y : \iA  \vdash y : \iA$, and $\Lambda = \Lambda', y : \iA$ if $\iA$ is ground, $\Lambda = \Lambda'$ otherwise. It is immediate to check that  $\meas{\pi^v} = 0$, $\Gamma$ is empty, $\pi^v$ contains no \paxs and $\flow{\pi^v}$ is acyclic: 
indeed, either  $\iP = \iL$ is ground, and we use \Cref{lemma:basic_val}, or  $\iP = []$, and we resort to \Cref{lemma:split}.
Therefore we have:
\[
\infer[\ilet]{\pi^{\tlet}\dem \Lambda', y : \iA \vdash y\esub{x}{v} : \iA}
  {\pi^v \dem \Lambda \vdash v : \vnode0{\iP}  & 
    \infer[\ivar]{\pi \dem \Lambda', x : \vnode1{\iP}, y : \iA  \vdash y : \iA}{}}
\specflowedge{node0}{node1}{to path={ .. controls +(.5,-.5) and +(-.5,-0.5) .. (\tikztotarget) }}
\quad \rightsquigarrow \quad 
\infer[\ivar]{\pi' \dem \Lambda', y : \iA \vdash y : \iA}{}
\] 
Clearly $\meas{\pi'} = 0 = \meas{\pi^v} + \meas{\pi}$. Observe that erasing $\pi^v$ has no effect on the flow, because $\pi^v$ does not contain cyclic paths; similarly, there are no \paxs in $\pi^{\tlet}$ nor in $\pi'$, as $\pi^v$ contains no \paxs.
\end{enumerate} 
\item Case \icoin is  similar to the second subcase of \ivar. 
\item Case \icond. Let $A = \XX$; the only interesting subcase that differs from the second subcase of \ivar is: 
\[
\infer[\ilet]{\pi^{\tlet}\dem \Lambda', y: \vnode3{\YY} \vdash \CPT{\dots, x, \dots}\esub{x}{y} : \XX}
{ \infer[\ivar]{\pi^v \dem \Lambda', y: \vnode0{\YY} \vdash y: \vnode1{\YY} }{}  &  
  \infer[\icond]{\pi \dem \Lambda', y: \vnode5{\YY}, x : \vnode2{\YY} \vdash \CPT{\dots, x, \dots} : \XX}{} }
\specdirflowedge{node1}{node2}{to path={ .. controls +(.5,-.5) and +(-.5,-0.5) .. (\tikztotarget) }}
\bentdirflowedges{node0/node1/55}
\dirflowedges{node3/node0, node3/node5}
\]
\[\rightsquigarrow\]
\[
\infer[\icond]{\pi' \dem \Lambda', y : \YY   \vdash \CPT{\dots, y , \dots} : \XX}{}
\]
Note that both $\Gamma$ and $\Delta$ are empty contexts, and that the free variable $x$ has been renamed. Again we have $\meas{\pi'} = 0 = \meas{\pi^v} + \meas{\pi}$; moreover the requirements on cyclic paths and \paxs are trivially satisfied, as $\pi^v$ contains neither of them.
\item Case \iobs. Similarly to what happens with case \ivar, we distinguish two subcases. The first one is $\pi \dem \Lambda, x : \nX^\tb \vdash \obsb{x} : \nX^\tb$, that is $A = P = \nX^\tb$; observe that necessarily $ \pi^v \dem \Lambda \vdash y : \nX^\tb$, so that $\pi' \dem \Lambda, y : \nX^\tb \vdash \obsb{y} : \nX^\tb$ is immediately otained by variable renaming. The other subcase is $\pi \dem \Lambda', x : P, y : \nX^\tb \vdash \obsb{y} : \nX^\tb$, that is $\Lambda = \Lambda', y : \nX^\tb$; the proof proceeds as in the second subcase of \ivar, yielding $\pi' \dem \Lambda', y : \nX^\tb \vdash \obsb{y} : \nX^\tb$. Clearly in both cases one has $\meas{\pi'} = 0 = \meas{\pi} + \meas{\pi^v}$, no probabilistic axiom is involved, and the flow, after rewriting $\pi^\tlet$ into $\pi'$, is essentially unchanged because $\pi^v$ contains no cyclic paths.
\item Case \ipair. Posing $A = \iL_1 \otimes \iL_2$, we have:
\[
\infer[\ilet]{\pi^{\tlet}\dem \Lambda \vdash \pair{w_1}{w_2}\esub{x}{v}: \iL_1 \otimes \iL_2}
{\pi^v \dem \Lambda \vdash v : \vnode0{\iP}   &
  \infer[\ipair]{\pi \dem \Lambda, x : \vnode1{\iP} \vdash \pair{w_1}{w_2} : \iL_1 \otimes \iL_2 }
   {\pi_1 \dem \Lambda, x : \vnode2{\iP} \vdash w_1 : \iL_1 & \pi_2 \dem \Lambda, x : \vnode3{\iP} \vdash w_2 : \iL_2 }
\specdirflowedge{node0}{node1}{to path={ .. controls +(.5,-.4) and +(-.5,-0.4) .. (\tikztotarget) }}
\dirflowedges{node1/node2, node1/node3}
}
\]
\[\rightsquigarrow\]
\[
\infer[\ilet]{\pi^{\tlet}_1 \dem \Lambda \vdash w_1 \esub{x}{v} : \iL_1}
	{\pi^v \dem \Lambda \vdash v : \vnode0{\iP} & \Lambda, x : \vnode1{\iP} \vdash w_1 : \iL_1}
\specdirflowedge{node0}{node1}{to path={ .. controls +(.4,-.4) and +(-.4,-0.4) .. (\tikztotarget) }}
\qquad
\infer[\ilet]{\pi^{\tlet}_2 \dem \Lambda \vdash w_2 \esub{x}{v} : \iL_2}
	{\pi^v \dem \Lambda \vdash v : \vnode0{\iP} & \Lambda, x : \vnode1{\iP} \vdash w_2 : \iL_2}
\specdirflowedge{node0}{node1}{to path={ .. controls +(.4,-.4) and +(-.4,-0.4) .. (\tikztotarget) }}
\]
\[\rightsquigarrow\]
\[
\infer[\ipair]{\pi' \dem \Lambda \vdash \pair{w_1}{w_2}\subs{x}{v}: \iL_1 \otimes \iL_2}
{\infer=[i.h.]{ \pi_1' \dem \Lambda \vdash w_1 \subs{x}{v} : \iL_1}{}
&
\infer=[i.h.]{\pi_2' \dem \Lambda  \vdash w_2 \subs{x}{v} : \iL_2}{}
}
\]
By \Cref{lemma:basic_val}, $\iP$ is either ground or the empty multiset $[]$; we also deduce that $\meas{\pi^v} = \meas{\pi} = \meas{\pi'} = 0$, that both $\Gamma$ and $\Delta$ are empty, and that $\pi^{\tlet}$ contains no \pax. 
Since all the types involved are ground (with the possible exception of $\iP = []$), both $\flow{\pi^{\tlet}}$ and $\flow{\pi'}$ are clearly acyclic, and paths between positions in the conclusion are trivially preserved.
\item Case \iabs.
By letting $A = \iA_1 \arrow \iA_2$, we have:
\[ 
\infer[\ilet]{\pi^{\tlet} \dem \Lambda, \Gamma \uplus \Delta \vdash \lam y. u\esub{x}{v} : \iA_1 \larrow \iA_2}
{\pi^v \dem \Lambda, \Gamma \vdash v : \vnode1{\iP}  & 
	\infer[\iabs]{\pi \dem \Lambda, \Delta, x : \vnode2{\iP} \vdash \lambda y.u : \iA_1 \larrow \iA_2  }
	{\pi_1 \dem \Lambda, \Delta, x : \vnode0{\iP}, y : \iA_1 \vdash u : \iA_2}}
\flowedge{node0}{node2}
\specflowedge{node1}{node2}{to path={ .. controls +(.5,-.4) and +(-.5,-0.4) .. (\tikztotarget) }}
\]
\[\rightsquigarrow\]
\[
\infer[\ilet]{\pi_1^{\tlet} \dem \Lambda, \Gamma \uplus \Delta, y : \iA_1 \vdash u\esub{x}{v} : \iA_2}{ \Lambda, \Gamma \vdash v : \vnode0{\iP} & \Lambda, \Delta, x : \vnode1{\iP}, y : \iA_1 \vdash u : \iA_2  }
\specflowedge{node0}{node1}{to path={ .. controls +(.5,-.4) and +(-.5,-0.4) .. (\tikztotarget) }}
\]
\[\rightsquigarrow\]
\[
\infer[\iabs]{\pi' \dem \Lambda, \Gamma \uplus \Delta \vdash \lam y. u\isub{x}{v} : \iA_1 \larrow \iA_2}
{\infer=[\ih]{\pi_1' \dem \Lambda, \Gamma \uplus \Delta, y : \iA_1 \vdash u\isub{x}{v} : \iA_2}{} }
\]
Alpha conversion ensures $y \not \in \dom(\Gamma)$, and the \ih guarantees there exists $\pi_1' \dem \Lambda, \Gamma \uplus \Delta, y : \iA_1 \vdash u\isub{x}{v} : \iA_2$ containing the same \paxs as $\pi^{\tlet}_1$ and such that $\meas{\pi_1'} = \meas{\pi^v} + \meas{\pi_1}$; observe that $\meas{\pi_1} = \meas{\pi}$ and $\meas{\pi_1'} = \meas{\pi'}$. The \ih also ensures that cycles and paths between positions of the conclusion are preserved when one rewrites $\pi^{\tlet}_1$ into $\pi'_1$.
\item Case \iapp. 
We have:
\begin{small}
\[
\infer[\ilet]{\pi^{\tlet} \dem \Lambda, \Gamma \uplus \Delta \vdash {u_1 u_2}\esub{x}{v}: \iA}
{\pi^v \dem \Lambda, \Gamma \vdash v : \vnode2{\iP}  &
\infer[\iapp]{\pi \dem \Lambda, \Delta, x : \vnode3{\iP} \vdash u_1 u_2 : \iA }
{\pi_1 \dem \Lambda, \Delta_1, x : \iP_1 \vdash u_1 : \vnode0{\iQ} \arrow \iA & 
 \pi_2 \dem \Lambda, \Delta_2, x : \iP_2 \vdash u_2 : \vnode1{\iQ} }
\specflowtwoedge{node0}{node1}{to path={ .. controls +(.6,-.5) and +(-.6,-0.5) .. (\tikztotarget) } }
\specflowedge{node2}{node3}{to path={ .. controls +(.6,-.5) and +(-.6,-0.5) .. (\tikztotarget) }}
}
\]
\[\rightsquigarrow\]
\[
\infer[\ilet]{ \pi_1^{\tlet}\dem \Lambda, \Gamma_1 \uplus \Delta_1 \vdash u_1 \esub{x}{v} : \iQ \arrow \iA}
{\pi^v_1\dem \Lambda, \Gamma_1 \vdash v : \vnode0{\iP_1} & \Lambda, \Delta_1, x : \vnode1{\iP_1} \vdash u_1 : \iQ \arrow \iA }
\quad\quad
\infer[\ilet]{\pi_2^{\tlet} \dem \Lambda, \Gamma_2 \uplus \Delta_2  \vdash u_2 \esub{x}{v} : \iQ}
{\pi^v_2\dem\Lambda, \Gamma_2 \vdash v : \vnode2{\iP_2} & \Lambda, \Delta_2, x : \vnode3{\iP_2} \vdash u_2 : \iQ}
\specflowedge{node0}{node1}{to path={ .. controls +(.5,-.5) and +(-.5,-0.5) .. (\tikztotarget) } }
\specflowedge{node2}{node3}{to path={ .. controls +(.5,-.5) and +(-.5,-0.5) .. (\tikztotarget) }}
\]
\[\rightsquigarrow\]
\[
\infer[\iapp]{\pi' \dem \Lambda, \Gamma \uplus \Delta \vdash {u_1 u_2}\subs{x}{v}: \iA}
{\infer=[\ih]{ \pi_1'\dem \Lambda, \Gamma_1 \uplus \Delta_1 \vdash u_1 \isub{x}{v} : \vnode0{\iQ} \arrow \iA}{}
&
\infer=[\ih]{\pi_2' \dem \Lambda, \Gamma_2 \uplus \Delta_2  \vdash u_2 \isub{x}{v} : \vnode1{\iQ} }{}
\specflowtwoedge{node0}{node1}{to path={ .. controls +(.6,-.5) and +(-.6,-0.5) .. (\tikztotarget) } }
}
\]
\end{small}
We distinguish to subcases:
\begin{itemize}
\item case  $\iP = \iP_1 = \iP_2$ is ground, i.e. $\pi^v = \pi^v_1 = \pi^v_2$. The result follows by  \Cref{lemma:basic_val}, which guarantees $\meas{\pi^v} = 0$ and $\Gamma$ empty, and by \ih, ensuring there exist $\pi_1' \dem \Lambda, \Delta_1  \vdash u_1 \subs{x}{v}: \iQ \arrow \iA$ and $\pi_2' \dem \Lambda, \Delta_2  \vdash u_2 \subs{x}{v} : \iQ $ such that $\meas{\pi_i'} = \meas{\pi^v} + \meas{\pi_i} = \meas{\pi_i}$, for $i \in \{1,2\}$. Indeed, $\meas{\pi'} = \meas{\pi_1'} + \meas{\pi_2'} + 2 = \meas{\pi^v} + \meas{\pi}$.
\item case $\iP_1 = \mult_1$, $\iP_2 = \mult_1$ and $\iP = \mult_1 \uplus \mult_2$. By \Cref{lemma:split}, there exist $\pi^v_i \dem \Lambda, \Gamma_i \vdash v : \mult_i$ ($i \in \{1,2\}$) such that $\Gamma = \Gamma_1 \uplus \Gamma_2$ and $\meas{\pi^v} = \meas{\pi^v_1} + \meas{\pi^v_2}$. Moreover, by \ih there are $\pi_1' \dem \Lambda, \Gamma_1 \uplus \Delta_1 \vdash u_1 \subs{x}{v} : \iQ \arrow \iA$ and $\pi_2' \dem \Lambda, \Gamma_2 \uplus \Delta_2 \vdash u_2 \subs{x}{v} : \iQ$ such that $\meas{\pi_i'} = \meas{\pi^v_i} + \meas{\pi_i}$ for $i \in \{1,2\}$. Therefore $\meas{\pi'} = \meas{\pi_1'} + \meas{\pi'_2} + 2 = \meas{\pi^v_1} + \meas{\pi^v_2} + \meas{\pi_1} + \meas{\pi_2} + 2 = \meas{\pi^v} + \meas{\pi}$.
\end{itemize}
In both cases, the \ih ensures that $\pi^{\tlet}_i$ and $\pi'_i$ ($i \in \{1,2\}$) contain the same \paxs.
Lastly we examine the flow. Observe that:  
\begin{itemize}
\item any path or cycle in $\pi^v$ appears in $\pi^v_1$ or in $\pi^v_2$. 
\item any (possibly cyclic) path in $\pi^{\tlet}$ and be partitioned into  paths in $\pi^v$, paths in $\pi_1$, paths in $\pi_2$, the edges between $\pi_1$ and $\pi_2$ via $\iQ$, 
the paths  between $\pi^v$ and $\pi_i$ via $\iP_i$ ($i\in \{1,2\}$), plus the trivial paths  between the border of $\pi^{\tlet}$ 
and the borders of $\pi^v$, $\pi_1$ and $\pi_2$. 
\end{itemize}
Consider a (possibly cyclic) path $\gamma$ in $\pi^{\tlet}$:
\begin{itemize}
\item If $\gamma$ does not use $\iQ$, then necessarily it is composed of subpaths $\{\gamma_j\}_{j\in I_1}$  in $\pi^v_1$ and $\pi_1$, possibly via $ \iP_1 $,  or paths $\{\gamma_j\}_{j \in I_2}$ in  
$\pi^v_2$ and $\pi_2$, possibly via $ \iP_2 $. In the first case, for each $\gamma_j$ $(j \in I_1)$, a path with the same behaviour  
appears in $\pi^{\tlet}_1$ and, by \ih,  in $\pi_1'$: consequently it appears also in $\pi'$. Similarly for the second case. 
\item  If $\gamma$ does use $\iQ$, we consider the subpaths using $\pi^v_1$ and $\pi_1$, possibly via $\iP_1$, and those using $\pi^v_2$ and $\pi_2$, possibly via $ \iP_2 $. By \ih, we find all of them in $\pi^{\tlet}_1$ and $\pi^{\tlet}_1$, and therefore in $\pi_1'$ and $\pi_2'$. Hence in $\pi'$, using the edges via $\iQ$, we find again a cycle if $\gamma$ was a cycle, or a path between the positions $i$ and $j$  in the conclusion of $\pi'$, if $\gamma$ was a path between the positions $i$ and $j$ in the conclusion of $\pi^{\tlet}$.
\end{itemize}
\item Case \ilet. We have: 
\begin{small}
\[
\infer[\ilet]{\pi^{\tlet} \dem \Lambda, \Gamma \uplus \Delta \vdash u\esub{y}{s}\esub{x}{v}: \iA}
{\pi^v \dem \Lambda, \Gamma \vdash v : \vnode0{\iP}   &
 \infer[\ilet]{\pi \dem \Lambda, \Delta, x : \vnode1{\iP} \vdash u \esub{y}{s} : \iA }
 {\pi_1 \dem \Lambda, \Delta_1, x : \iP_1 \vdash s : \vnode2{\iQ} & \pi_2 \dem \Lambda, \Delta_2, x : \iP_2, y : \vnode3{\iQ} \vdash u : \iA }
}
\specflowedge{node0}{node1}{to path={ .. controls +(.5,-.4) and +(-.5,-0.4) .. (\tikztotarget) }}
\specflowtwoedge{node2}{node3}{to path={ .. controls +(.5,-.5) and +(-.5,-0.5) .. (\tikztotarget) }}
\]
\[\rightsquigarrow\]
\[
 \infer[\ilet]{ \pi^{\tlet}_1 \dem \Lambda, \Gamma_1 \uplus \Delta_1 \vdash s\esub{x}{v} : \vnode2{\iQ} }
  {\pi^v_1 \dem \Lambda, \Gamma_1 \vdash v : \vnode0{\iP_1} & \Lambda, \Delta_1, x : \vnode1{\iP_1} \vdash s : \vnode3{\iQ} }
\specflowedge{node0}{node1}{to path={ .. controls +(.5,-.5) and +(-.5,-0.5) .. (\tikztotarget) }}
\qquad
  \infer[\ilet]{ \pi^{\tlet}_2 \dem \Lambda, \Gamma_2 \uplus \Delta_2, y : \vnode4{\iQ} \vdash u\esub{x}{v} : \iA}
  {\pi^v_2 \dem \Lambda, \Gamma_2 \vdash v : \vnode0{\iP_2} & \Lambda, \Delta_2, x : \vnode1{\iP_2},  y : \vnode5{\iQ} \vdash u : \iA}
\specflowedge{node0}{node1}{to path={ .. controls +(.5,-.5) and +(-.5,-0.5) .. (\tikztotarget) }}
\]
\[\rightsquigarrow\]
\[
\infer[\ilet]{\pi' \dem \Lambda, \Gamma \uplus \Delta \vdash u\esub{y}{s}\isub{x}{v}: \iA}
{   \infer=[i.h.]{ \pi_1' \dem \Lambda, \Gamma_1 \uplus \Delta_1 \vdash s\isub{x}{v} : \vnode0{\iQ} }{}
	&
	\infer=[i.h.]{\pi_2' \dem \Lambda, \Gamma_2 \uplus \Delta_2, y : \vnode1{\iQ} \vdash u\isub{x}{v} : \iA}{}
\specflowtwoedge{node0}{node1}{to path={ .. controls +(.5,-.5) and +(-.5,-0.5) .. (\tikztotarget) }}
}
\]
\end{small}
The reasoning is similar to case \iapp, with two subcases:
\begin{itemize}
\item case $\iP = \iP_1 = \iP_2$ is ground. Again \Cref{lemma:basic_val} assures $\meas{\pi^v} = 0$ and $\Gamma$ empty. By \ih there exist $\pi_1' \dem \Lambda, \Delta_1  \vdash s \subs{x}{v}: \iQ$ and $\pi_2' \dem \Lambda, \Delta_2, y : \iQ \vdash u \subs{x}{v} : \iA$ such that $\meas{\pi_i'} = \meas{\pi^v} + \meas{\pi_i} = \meas{\pi_i}$, for $i \in \{1,2\}$; the result follows observing that $\meas{\pi'} = \meas{\pi_1'} + \meas{\pi_2'} + 1 = \meas{\pi^v} + \meas{\pi}$.
\item case $\iP_1 = \mult_1$, $\iP_2 = \mult_1$ and $\iP = \mult_1 \uplus \mult_2$. We use \Cref{lemma:split} and the inductive hypothesis to obtain $\pi_1' \dem \Lambda, \Gamma_1 \uplus \Delta_1  \vdash s \subs{x}{v}: \iQ$ and $\pi_2' \dem \Lambda, \Gamma_2 \uplus \Delta_2, y : \iQ \vdash u \subs{x}{v} : \iA$ such that $\meas{\pi_i'} = \meas{\pi^v_i} + \meas{\pi_i}$ for $i \in \{1,2\}$. Therefore $\meas{\pi'} = \meas{\pi_1'} + \meas{\pi'_2} + 1 = \meas{\pi^v_1} + \meas{\pi^v_2} + \meas{\pi_1} + \meas{\pi_2} + 1 = \meas{\pi^v} + \meas{\pi}$.
\end{itemize}
As in case \iapp, the \ih ensures that $\pi^{\tlet}_i$ and $\pi'_i$ ($i \in \{1,2\}$) contain the same \paxs; the analysis of the flow is also analogous to the aforementioned case.
\item Case \iletp. 
\begin{small}
\[
\infer[\ilet]{\pi^{\tlet} \dem \Lambda, \Gamma \uplus \Delta \vdash u\esub{\pair{y_1}{y_2}}{w} : \iA}
{\pi^v \dem \Lambda, \Gamma \vdash v : \vnode0{\iP}   &
 \infer[\iletp]{\pi \dem \Lambda, \Delta, x : \vnode1{\iP} \vdash u\esub{\pair{y_1}{y_2}}{w} : \iA }
 {\pi_1 \dem \Lambda, x : \iP_1 \vdash w : \vnode2{\iK_1} \otimes \vnode4{\iK_2} & \pi_2 \dem \Lambda, \Delta, x : \iP_2, y_2 : \vnode5{\iK_2}, y_1 : \vnode3{\iK_1},  \vdash u : \iA }
}
\specflowedge{node0}{node1}{to path={ .. controls +(.5,-.4) and +(-.5,-0.4) .. (\tikztotarget) }}
\specdirflowtwoedge{node2}{node3}{to path={ .. controls +(.5,-.5) and +(-.5,-0.5) .. (\tikztotarget) }}
\specdirflowtwoedge{node4}{node5}{to path={ .. controls +(.35,-.35) and +(-.35,-0.35) .. (\tikztotarget) }}
\]
\[\rightsquigarrow\]
\[
 \infer[\ilet]{ \pi^{\tlet}_1 \dem \Lambda \vdash w\esub{x}{v} : \vnode2{\iK_1 \otimes \iK_2} }
  {\pi^v_1 \dem \Lambda \vdash v : \vnode0{\iP_1} & \Lambda, x : \vnode1{\iP_1} \vdash w : \vnode3{\iK_1 \otimes \iK_2} }
\specflowedge{node0}{node1}{to path={ .. controls +(.4,-.5) and +(-.4,-0.5) .. (\tikztotarget) }}
\qquad
  \infer[\ilet]{ \pi^{\tlet}_2 \dem \Lambda, \Gamma \uplus \Delta, y : \vnode4{\iK_1 \otimes \iK_2} \vdash u\esub{x}{v} : \iA}
  {\pi^v_2 \dem \Lambda, \Gamma \vdash v : \vnode0{\iP_2} & \Lambda, \Delta, x : \vnode1{\iP_2},  y_1 : \vnode5{\iK_1}, y_2 : {\iK_2} \vdash u : \iA}
\specflowedge{node0}{node1}{to path={ .. controls +(.5,-.5) and +(-.5,-0.5) .. (\tikztotarget) }}
\]
\[\rightsquigarrow\]
\[
\infer[\iletp]{\pi' \dem \Lambda, \Gamma \uplus \Delta \vdash u\esub{\pair{y_1}{y_2}}{w}\isub{x}{v} : \iA}
{   \infer=[i.h.]{ \pi_1' \dem \Lambda \vdash w\isub{x}{v} : \vnode0{\iK_1} \otimes \vnode2{\iK_2} }{}
	&
	\infer=[i.h.]{\pi_2' \dem \Lambda, \Gamma \uplus \Delta, y_2 : \vnode3{\iK_2}, y_1 : \vnode1{\iK_1} \vdash u\isub{x}{v} : \iA}{}
\specdirflowtwoedge{node0}{node1}{to path={ .. controls +(.5,-.5) and +(-.5,-0.5) .. (\tikztotarget) }}
\specdirflowtwoedge{node2}{node3}{to path={ .. controls +(.35,-.35) and +(-.35,-0.35) .. (\tikztotarget) }}
}
\]
\end{small}
There are two subcases:
\begin{itemize}
\item case $\iP_1 = \iP_2 = \iP$ is ground. \Cref{lemma:basic_val} assures $\meas{\pi^v} = 0$ and $\Gamma$ empty. By \ih there are $\pi_1' \dem \Lambda \vdash w\subs{x}{v}: \iK_1 \otimes \iK_2$ and $\pi_2' \dem \Lambda, \Delta, y_1 : \iK_1, y_2 : \iK_2 \vdash u \subs{x}{v} : \iA$ such that $\meas{\pi_i'} = \meas{\pi^v} + \meas{\pi_i} = \meas{\pi_i}$, for $i \in \{1,2\}$; observe that $\meas{\pi_1} = \meas{\pi_1'} = 0$, therefore $\meas{\pi'} = \meas{\pi_2'} + 3 = \meas{\pi^v} + \meas{\pi}$.
\item case $\iP_1 = []$ and $\iP_2 = \iP$. By \Cref{lemma:split} and the \ih we obtain $\pi_1' \dem \Lambda  \vdash w\subs{x}{v}: \iK_1 \otimes \iK_2$ and $\pi_2' \dem \Lambda, \Gamma \uplus \Delta, y_1 : \iK_1, y_2 : \iK_2 \vdash u\subs{x}{v} : \iA$ such that $\meas{\pi_i'} = \meas{\pi^v_i} + \meas{\pi_i}$ for $i \in \{1,2\}$. Remark that $\meas{\pi_1} = \meas{\pi_1'} = \meas{\pi^v_1} = 0$, hence $\meas{\pi'} = \meas{\pi'_2} + 3 = \meas{\pi^v_2} + \meas{\pi_2} + 3 = \meas{\pi^v} + \meas{\pi}$.
\end{itemize}
The analysis of both the \paxs and of the flow proceeds as usual.
\item Case \ibang. Let $\iA = [\iA_1,\dots,\iA_n]$; then we have:
\[ 
\infer[\ilet]{\pi^{\tlet} \dem \Lambda, \Gamma \uplus \Delta \vdash \oc u\esub{x}{v} : [\iA_1,\dots,\iA_n]}
{\pi^v \dem \Lambda, \Gamma \vdash v : \vnode0{\iP}  & 
	\infer[\ibang]{\pi \dem \Lambda, \Delta, x : \vnode1{\iP} \vdash  \oc u : [\iA_1,\dots,\iA_n] }
    { \big( \ \pi_i \dem \Lambda, \Delta_i, x : \iP_i \vdash u : \iA_i \ \big)_{i=1}^n }  }
\specflowedge{node0}{node1}{to path={ .. controls +(.5,-.5) and +(-.5,-0.5) .. (\tikztotarget) } }
\]
\[ \rightsquigarrow \]
\[
\Big( \ \vcenter{
 \infer[\ilet]{ \pi^{\tlet}_i \dem \Lambda, \Gamma_i \uplus \Delta_i \vdash u\esub{x}{v} : \iA_i }
  { \pi^v_i \dem \Lambda, \Gamma_i \vdash v : \vnode0{\iP_i} & \Lambda, \Delta_i, x : \vnode1{\iP_i} \vdash u : \iA_i  }  } \
\Big)^n_{i=1}
\specflowedge{node0}{node1}{to path={ .. controls +(.5,-.5) and +(-.5,-0.5) .. (\tikztotarget) } }
\]
\[ \rightsquigarrow \]
\[
\infer[\ibang]{\pi' \dem \Lambda, \Gamma \uplus \Delta \vdash \oc u \subs{x}{v} : [\iA_1,\dots,\iA_n]}
{ \infer=[\ih]{\big( \ \pi_i' \dem \Lambda, \Gamma_i \uplus \Delta_i \vdash u\subs{x}{v} : \iA_i \ \big)_{i=1}^n}{}
	  }
\]
As always, consider the two subcases:
\begin{itemize}
\item case $\iP = \iP_i$ is ground, i.e. $\pi^v = \pi^v_i$ for all $i$. By \Cref{lemma:basic_val} we know that $\meas{\pi^v} = 0$ and $\Gamma $ is an empty context. Moreover, by \ih there exist $\pi_i' \dem \Lambda, \Delta_i  \vdash u \subs{x}{v} : \iA_i$ such that $\meas{\pi_i'} = \meas{\pi^v} + \meas{\pi_i} = \meas{\pi_i}$ ($1 \leq i \leq n$).
\item case $\iP_i = \mult_i$ and $\iP = \biguplus_{i=1}^n \mult_i$. Then by Lemma \ref{lemma:split} there exist $\Lambda, \Gamma_i \vdash v : \mult_i$ ($1 \leq i \leq n$) such that $\biguplus_{i=1}^n \Gamma_i = \Gamma$ and $\meas{\pi} = \sum_{i=1}^n \meas{\pi_i}$. By \ih there are  $\pi_i' \dem \Lambda, \Gamma_i \uplus \Delta_i \vdash u \subs{x}{v} : \iA_i$ such that $\meas{\pi_i'} = \meas{\pi^v_i} + \meas{\pi_i}$ ($1 \leq i \leq n$).
\end{itemize}
In both cases, \ih guarantees that $\pi^{\tlet}_i$ and $\pi'_i$ ($1 \leq i \leq n$) contain the same \paxs. The analysis of the flow is straightforward, as every path between the positions of $\iP$ in $\pi^{\tlet}$ can be recovered via the paths between the positions of $\iP_i$ in $\pi^{\tlet}_i$.
\item Case \ider. 
\[ 
\infer[\ilet]{\pi^{\tlet} \dem \Lambda, \Gamma \uplus \Delta \vdash \der u\esub{x}{v} : [\iA]}
{\pi^v \dem \Lambda, \Gamma \vdash v : \vnode0{\iP}  & 
	\infer[\ider]{\pi \dem \Lambda, \Delta, x : \vnode1{\iP} \vdash \der u : \iA }
     {\pi_1 \dem \Lambda, \Delta, x : \vnode2{\iP} \vdash u : [\iA] }  }
\specflowedge{node0}{node1}{to path={ .. controls +(.5,-.5) and +(-.5,-0.5) .. (\tikztotarget) } }
\flowedge{node1}{node2}
\]
\[ \rightsquigarrow \]
\[
{ \infer[\ilet]{\pi_1' \dem \Lambda, \Gamma \uplus \Delta \vdash u\esub{x}{v} : [\iA]}
  { \Lambda, \Gamma \vdash v : \vnode0{\iP} & \Lambda, \Delta, x : \vnode1{\iP} \vdash u : [\iA] }  }
\specflowedge{node0}{node1}{to path={ .. controls +(.5,-.5) and +(-.5,-0.5) .. (\tikztotarget) } }
\]
\[ \rightsquigarrow \]
\[
\infer[\ider]{\pi' \dem \Lambda, \Gamma \uplus \Delta \vdash \der u\subs{x}{v} : \iA}
{ \infer=[\ih]{\pi_1' \dem \Lambda, \Gamma \uplus \Delta \vdash u\subs{x}{v} : [\iA]}
	{ }}
\]
By \ih there exists $\pi_1' \dem \Lambda, \Gamma \uplus \Delta \vdash u \subs{x}{v} : [\iA]$ such that $\meas{\pi_1'} = \meas{\pi^v} + \meas{\pi_1}$; note that $\meas{\pi_1} = \meas{\pi}$ and $\meas{\pi_1'} = \meas{\pi'}$. 
Moreover, \ih guarantees that the \paxs contained in $\pi^{\tlet}_1$ and $\pi'_1$ are the same, and that both cycles and path between positions of the conclusion are preserved when transforming the former derivation into the latter.
\end{itemize}
\end{proof}


\begin{lemma}[Subject reduction] 
If $\pi \dem \Lambda, \Gamma \vdash t : \iA$ and $t \mapsto t'$, then $\pi' \dem \Lambda, \Gamma \vdash t' : \iA$ . Moreover:
\begin{enumerate}
\item $\meas{\pi} > \meas{\pi'}$;
\item $\pi$ and $\pi'$ contain the same \paxs, modulo renaming of free variables;
\item if $\flow{\pi}$ contains a directed cycle, then $\flow{\pi'}$ contains a directed cycle;
\item if there is a directed path between two positions $i$ and $j$ in the conclusion of $\pi$, then there is a path between $i$ and $j$ in the conclusion of $\pi'$. 
\end{enumerate}
\end{lemma}

\begin{proof}
The proof is by induction on the reduction context $E$ in which the reduction takes place. For the sake of space we focus on the base case $E = \eshole{\cdot}$, the inductive cases being easy to check using the \ih.
There is one case per rewriting rule.
\begin{itemize}
\item Rule $\dB$. 
In this case we have $t = (\lambda x.s)\eslist u$ and $t' = \eshole{s \esub{x}{u}}\eslist $. Thus we proceed by induction on $\eslist$.
\begin{itemize}
\item Case $\eslist = \shole{\cdot}$. The derivation $\pi$ has shape:
\begin{prooftree}
\AxiomC{$\Lambda, \Gamma_1, x : \vnode0{\iP} \vdash s : \iA$}
\RightLabel{\iabs}
\UnaryInfC{$\Lambda, \Gamma_1 \vdash \lambda x.s : \vnode1{\iP} \arrow \iA $}
	\AxiomC{$\Lambda, \Gamma_2 \vdash u : \vnode2{\iP} $}
\RightLabel{\iapp}
\BinaryInfC{$\pi \dem \Lambda, \Gamma \vdash (\lambda x.s)u : \iA $}
$
\flowedge{node0}{node1}
\specflowedge{node1}{node2}{to path={ .. controls +(.5,-.5) and +(-.5,-0.5) .. (\tikztotarget) } }
$
\end{prooftree}

We conclude with 
\begin{prooftree}
	\AxiomC{$\Lambda, \Gamma_2 \vdash u : \vnode0{\iP} $}
	\AxiomC{$\Lambda, \Gamma_1, x : \vnode1{\iP} \vdash s : \iA $}
\RightLabel{\ilet}
\BinaryInfC{$\pi' \dem \Lambda, \Gamma \vdash s \esub{x}{u} : \iA $}
$
\specflowedge{node0}{node1}{to path={ .. controls +(.5,-.5) and +(-.5,-0.5) .. (\tikztotarget) } }
$
\end{prooftree}

Observe that $\meas{\pi} = \meas{\pi'} + 1$. 
Clearly the rewriting rule does not affect \paxs nor the paths between positions in the conclusion; 
cyclic paths are trivially preserved too.
\item Case $\eslist = \eslist'\esub{y}{r}$. The derivation $\pi$ is: 
\begin{prooftree}
\AxiomC{$\pi_1 \dem \Lambda, \Gamma_1 \vdash r : \vnode0{\iP} $}
	\AxiomC{$ \pi_2 \dem \Lambda, \Gamma_2, y : \vnode1{\iP} \vdash \eshole{\lambda x.s}\eslist' : \vnode2{\iQ} \arrow \iA $}
\RightLabel{\ilet}
\BinaryInfC{$\Lambda, \Gamma_1 \uplus \Gamma_2 \vdash \eshole{\lambda x.s}\eslist : \vnode3{\iQ} \arrow \iA $}
	\AxiomC{$\pi_3 \dem \Lambda, \Gamma_3 \vdash u : \vnode4{\iQ} $}
\RightLabel{\iapp}
\BinaryInfC{$\pi \dem \Lambda, \Gamma \vdash \eshole{\lambda x.s}\eslist u : \iA$}
$
\specflowedge{node0}{node1}{to path={ .. controls +(.4,-.5) and +(-.4,-.5) .. (\tikztotarget) } }
\flowtwoedge{node2}{node3}
\specflowtwoedge{node3}{node4}{to path={ .. controls +(.6,-.5) and +(-.6,-.5) .. (\tikztotarget) } }
$
\end{prooftree}

Looking at $\pi$, we can safely assume that $y \not \in \dom(\Gamma_3)$; then one can build 
\begin{prooftree}
\AxiomC{$\pi_2 \dem \Lambda, \Gamma_2, y : \vnode2{\iP} \vdash \eshole{\lambda x.s}\eslist' : \vnode0{\iQ} \arrow \iA$}
	\AxiomC{$\pi_3 \dem \Lambda, \Gamma_3 \vdash u : \vnode1{\iQ} $}
\RightLabel{\iapp}
\BinaryInfC{$\rho \dem \Lambda, \Gamma_2 \uplus \Gamma_3, y : \vnode3{\iP} \vdash \eshole{\lambda x.s}\eslist' u : \iA$}
$
\specflowtwoedge{node0}{node1}{to path={ .. controls +(.6,-.5) and +(-.6,-.5) .. (\tikztotarget) } }
\flowedge{node2}{node3}
$
\end{prooftree}

By \ih there exists $\rho' \dem \Lambda, \Gamma_2 \uplus \Gamma_3, y : \iP \vdash \eshole{s \esub{x}{u}}\eslist' : \iA $ preserving \paxs and such that $\meas{\rho} > \meas{\rho'}$; hence one can build 
\begin{prooftree}
\AxiomC{$\pi_1 \dem \Lambda, \Gamma_1 \vdash r : \vnode0{\iP}$}
	\AxiomC{$\rho' \dem \Lambda, \Gamma_2 \uplus \Gamma_3, y : \vnode1{\iP} \vdash \eshole{s \esub{x}{u}}\eslist' : \iA$}
\RightLabel{\ilet}
\BinaryInfC{$\pi' \dem \Lambda, \Gamma \vdash \eshole{s \esub{x}{u}}\eslist : \iA $}
$
\specflowedge{node0}{node1}{to path={ .. controls +(.5,-.5) and +(-.5,-.5) .. (\tikztotarget) } }
$
\end{prooftree}
Note that $\meas{\pi} = \meas{\rho} + \meas{\pi_1} + 1 > \meas{\rho'} + \meas{\pi_1} + 1 = \meas{\pi'}$.
We now examine the flow in detail. Any (possibly cyclic) path $\gamma$ in $\pi$ can be partitoned into paths inside $\pi_1,\pi_2,\pi_3$, those connecting $\pi_2$ and $\pi_3$ via $\iQ$,  the edges connecting $\pi_1$ and $\pi_2$ via $\iP$, and the trivial edges between the border of $\pi$ and the borders of the three subderivations. Any path in $\pi$ which does not use $\iP$ nor $\Gamma_1$ appears in $\rho$ and consequently, by \ih, in $\rho'$. 
Therefore if $\pi$ has a a cycle or a path between two positions $i$ and $j$ in its conclusion, we are able to recover a cycle or a path between  positions $i$ and $j$  also in $\pi'$, possibly by using the edges between $\iP$ in $\pi'$.
\end{itemize}
\item Rule $\sval$.
Then $t = s \esub{x}{\eshole{v}\eslist}$ and $t' = \eshole{s\subs{x}{v}}\eslist$. Again, we proceed by induction on $\eslist$.
\begin{itemize}
\item $\eslist = \shole{\cdot}$. Then $\pi$ is:
\begin{prooftree}
\AxiomC{$\pi_1 \dem \Lambda, \Gamma_1 \vdash v : \iQ$}	
	\AxiomC{$\pi_2 \dem \Lambda, \Gamma_2, x : \iQ \vdash s : \iA$}
\RightLabel{\ilet}
\BinaryInfC{$\pi \dem \Lambda, \Gamma \vdash s\esub{x}{v} : \iA$}
\end{prooftree}
and we conclude by Lemma \ref{lemma:substitution}, assuring there exists $\pi' \dem \Lambda, \Gamma \vdash s\subs{x}{v} : \iA$ such that $\meas{\pi'} = \meas{\pi_1} + \meas{\pi_2}$. The same Lemma guarantees that \paxs, paths between positions in the conclusion and cycles are preserved.
\item $\eslist = \eslist'\esub{y}{r}$. The derivation $\pi$ ha shape:
\begin{prooftree}
\AxiomC{$\pi_1 \dem \Lambda, \Gamma_1 \vdash r : \vnode0{\iP} $}
    \AxiomC{$\pi_2 \dem \Lambda, \Gamma_2, y : \vnode1{\iP} \vdash \eshole{v}\eslist' : \vnode2{\iQ} $}
\RightLabel{\ilet}
\BinaryInfC{$\Lambda, \Gamma_1 \uplus \Gamma_2 \vdash \eshole{v}\eslist : \vnode3{\iQ} $}	
	\AxiomC{$\pi_3 \dem \Lambda, \Gamma_3, x : \vnode4{\iQ} \vdash s : \iA$}
\RightLabel{\ilet}
\BinaryInfC{$\pi \dem \Gamma \vdash s\esub{x}{\eshole{v}\eslist} : \iA$}
$
\specflowedge{node0}{node1}{to path={ .. controls +(.5,-.5) and +(-.5,-.5) .. (\tikztotarget) } }
\flowtwoedge{node2}{node3}
\specflowtwoedge{node3}{node4}{to path={ .. controls +(.5,-.5) and +(-.5,-.5) .. (\tikztotarget) } }
$
\end{prooftree}
We start by building 
\begin{prooftree}
    \AxiomC{$\pi_2 \dem \Lambda, \Gamma_2, y : \vnode0{\iP} \vdash \eshole{v}\eslist' : \vnode2{\iQ} $}
\AxiomC{$\pi_3 \dem \Lambda, \Gamma_3, x : \vnode3{\iQ} \vdash s : \iA$}
\RightLabel{\ilet}
\BinaryInfC{$\rho \dem \Lambda, \Gamma_2 \uplus \Gamma_3, y: \vnode1{\iP} \vdash s\esub{x}{\eshole{v}\eslist'} : \iA$}
$
\flowedge{node0}{node1}
\specflowtwoedge{node2}{node3}{to path={ .. controls +(.5,-.5) and +(-.5,-.5) .. (\tikztotarget) } }
$
\end{prooftree}
By \ih there exists $\rho' \dem \Lambda, \Gamma_2 \uplus \Gamma_3 , y : \iP \vdash \eshole{s\isub{x}{v}}\eslist' : \iA$ preserving \paxs and such that $\meas{\rho} > \meas{\rho'}$. Therefore we can conclude with 
\begin{prooftree}
	    \AxiomC{$ \pi_1 \dem \Lambda, \Gamma_1 \vdash r : \vnode0{\iP} $}
\AxiomC{$\rho' \dem \Lambda, \Gamma_2 \uplus \Gamma_3 , y : \vnode1{\iP} \vdash \eshole{s\isub{x}{v}}\eslist' : \iA$}
\RightLabel{\ilet}
\BinaryInfC{$\pi' \dem \Lambda, \Gamma \vdash \eshole{s \subs{x}{v}}\eslist : \iA$}
$\specflowedge{node0}{node1}{to path={ .. controls +(.5,-.5) and +(-.5,-.5) .. (\tikztotarget) } }$
\end{prooftree}
It is easy to check that $\meas{\pi} = \meas{\rho} + \meas{\pi_1} + 1 > \meas{\rho'} + \meas{\pi_1} + 1 = \meas{\pi'}$.
The reasoning on the flow is similar to case $\dB$.
\end{itemize}
\item Rule $\dbang$.
In this case $ t = \der \oc s$ and $t' = s$. We simply have:
\begin{prooftree}
\AxiomC{$\pi' \dem \Gamma \vdash s : \iA$}
\RightLabel{\ibang}
\UnaryInfC{$\Gamma \vdash \oc s : [\iA]$}
\RightLabel{\ider}
\UnaryInfC{$\pi \dem \Gamma \vdash \der (\oc s) : \iA$}
\end{prooftree}
Clearly $\meas{\pi} = \meas{\pi'} + 1$, and the \paxs contained in $\pi$ and $\pi'$ are the same. Moreover, since cyclic paths can only be in $\pi'$, they are trivially preserved; a similar argument applies to the paths between positions in the conclusion.
\item Rule $\dpair$.
In this case $t$ is $\letp{\pair{x_1}{x_2}}{\pair{v_1}{v_2}}{s}$ and $t'$ is $s\esub{x_1}{v_1}\esub{x_2}{v_2}$.
The derivation $\pi$ is
\begin{prooftree}
\AxiomC{$\Lambda \vdash v_1 : \iL_1 $ }
	\AxiomC{$\Lambda \vdash v_2 : \iL_2 $ }
\RightLabel{\ipair}
\BinaryInfC{$\Lambda \vdash \pair{v_1}{v_2} : \iL_1 \otimes \iL_2 $}
	\AxiomC{$\Lambda, \Gamma, x_1 : \iL_1, x_2 : \iL_2 \vdash s : \iA $}
\RightLabel{\iletp}
\BinaryInfC{$\pi \dem \Lambda, \Gamma \vdash \letp{\pair{x_1}{x_2}}{\pair{v_1}{v_2}}{s} : \iA$}
\end{prooftree}
and it is easy to build
\begin{prooftree}
\AxiomC{$\Lambda \vdash v_2 : \iL_2$}
	\AxiomC{$\Lambda \vdash v_1 : \iL_1$ }
		\AxiomC{$\Lambda, \Gamma, x_1 : \iL_1, x_2 : \iL_2 \vdash s : \iA$}
	\RightLabel{\ilet}
	\BinaryInfC{$\Lambda, \Gamma, x_2 : \iL_2 \vdash s\esub{x_1}{v_1} : \iA$}
\RightLabel{\ilet}
\BinaryInfC{$\pi' \dem \Lambda, \Gamma \vdash s\esub{x_1}{v_1}\esub{x_2}{v_2} : \iA$}
\end{prooftree}
Observe that $\meas{\pi} = \meas{\pi'} + 1$. Since we only need to rearrange the subderivations, it is easy to verify that \paxs, the paths between positions in the conclusion, and cyclic paths, are indeed preserved.
\end{itemize}
\end{proof}


\section{Subject Expansion} \label{sec:subject-expansion}
We strengthen the statement of Subject Expansion, in order to have the elements we need to prove also \Cref{prop:uniqueness}.

\begin{lemma}[Anti-split] \label{lemma:anti-split}
Let $\pi_i \dem \Lambda, \Gamma_i \vdash v : \mult_i$ $(1 \leq i \leq n)$. Then there exists $\pi \dem \Lambda, \Gamma \vdash v : \mult$ such that $\Gamma = \biguplus_{i=1}^n \Gamma_i$ and $\mult = \biguplus_{i=1}^n \mult_i$. Moreover, if all $\pi_i$ $(1 \leq i \leq n)$ are unique, then $\pi$ is unique.
\end{lemma}

\begin{proof} We distinguish two cases:
\begin{itemize}
\item $v = x$. Then $\pi_i \dem \Lambda, x: \mult_i \vdash x : \mult_i$, and there is only one way to build $\pi \dem \Lambda, x : \mult \vdash x : \mult$, by using rule $\tvar$. Observe that $\Gamma$ is empty.
\item $v = \oc t$. Let $\mult_i = [A_{i1},\dots,A_{ik_i}]$; then $\pi_i$ is obtained from $k_i$ subderivations of shape $\pi_{ij}  \dem \Lambda, \Gamma_{ij} \vdash t : A_{ij}$ $(1 \leq j \leq k_i)$, and it is possible to build $\pi$ by combining all of the $\pi_{ij}$ via rule \ibang. Therefore the uniqueness of all $\pi_i$ implies the uniqueness of $\pi$.
\end{itemize}
\end{proof}

\begin{lemma}[Anti-substitution]\label{lemma:anti-substitution} 
If $\pi \dem \Lambda, \Gamma\uplus\Delta \vdash t\subs{x}{v} : \iA$, then there exist $\pi^t \dem \Lambda, \Gamma, x : \iP \vdash t : \iA$ and $\pi^v \dem \Lambda, \Delta \vdash v : \iP$. Moreover, if $\pi$ is unique, both $\pi^t$ and $\pi^v$ are unique.
\end{lemma}

\begin{proof}
By induction on $t$.
\begin{enumerate}
\item Case $t = x$. \label{anti-subs-x} Then $\iP = \iA$, the derivation $\pi^t$ is necessarily $\pi^t \dem \Lambda, x : \iA \vdash x : \iA$ (observe that $\Gamma$ is empty), and $\pi^v = \pi$.
\item \label{anti-subs-y} Case $t = y \not = x$.  Then $\pi^t$ is necessarily $\pi^t \dem \Lambda', y : \iA \vdash y : \iA$ where $\Lambda = \Lambda', y : \iA$ if $\iA$ is ground, $\Lambda = \Lambda'$ otherwise. We distinguish two subcases: 
\begin{itemize}
\item $\Lambda' = \Lambda'', x : \iP$, i.e. $\iP$ is ground. Then by \Cref{lemma:basic_val} $\Delta$ is empty; observe that there is a unique way of constructing $\pi^v \dem \Lambda \vdash v : \iP$, using rules \ivar and \ipair only.
\item otherwise $\iP = []$, therefore $v = \oc u$ and $\pi^v \dem \Lambda \vdash v : \iP$ can only be obtained by rule \ibang posing $n=0$.
\end{itemize}
\item Case $t = \sample{d}$. Similar to case (\ref{anti-subs-y}).
\item Case $t = \CPT{y_1,\dots,y_n}$ where $x = y_i$ ($1 \leq i \leq n$). This scenario is similar to case (\ref{anti-subs-x}), with $\pi^t \dem \Lambda, x : \YY_i \vdash \CPT{y_1,\dots,y_n} : \XX$ and $\pi^v \dem \Lambda \vdash v : \YY_i$, where $v = z$ is a variable. 
\item Case $t = \CPT{y_1,\dots,y_n}$ where $x \not = y_i$ ($1 \leq i \leq n$). Similar to to case (\ref{anti-subs-y}).
\item Case $t = \obsb{x}$. Similar to case (\ref{anti-subs-x}): necessarily $\pi^t \dem \Lambda, x : \nX^\tb \vdash \obsb{x} : \nX^\tb$ and $\pi^v \dem \Lambda \vdash v : \nX^\tb$, where $v = z$ is a variable.
\item Case $t = \obsb{y}$ where $y \not = x$.  Necessarily $\pi^t \dem \Lambda, x : \iP \vdash \obsb{y} : \nX^\tb$, and for $\pi^v \dem \Lambda \vdash v : \iP$ we distinguish the subcases $\iP$ ground or $\iP = []$, as in case (\ref{anti-subs-y}).
\item Case $t = \pair{w_1}{w_2}$. Then $\pi$ has shape:
\[
\infer[\ipair]{\pi \dem \Lambda \vdash \pair{w_1}{w_2} \subs{x}{v} : \iA}{\Lambda \vdash w_1\subs{x}{v} : \iL_1  &  \Lambda \vdash w_2\subs{x}{v} : \iL_2 }
\] 
Observe that since $\iA = \iL_1 \otimes \iL_2$, all exponential contexts are empty (by \Cref{lemma:basic_val}).
By \ih there are $\pi^{w_1} \dem \Lambda, x : \iP \vdash w_1 : \iL_1$ and $\pi^v_1 \dem \Lambda \vdash v : \iP$, together with $\pi^{w_2} \dem \Lambda, x : \iP \vdash w_2 : \iL_2$ and $\pi^v_2 \dem \Lambda \vdash v : \iP$; remark that $P$ is ground (again, by \Cref{lemma:basic_val}). Since all the aforementioned derivations are built from rules \ivar and \ipair only, they are necessarily unique. Clearly $\pi^t$ is obtained from $\pi^{w_1}$ and $\pi^{w_2}$ by rule \ipair, and $\pi^v = \pi^v_1 = \pi^v_2$.
\item \label{anti-subs-abs} Case $t = \lambda y.u$. The derivation $\pi$ has shape: 
\[
\infer[\iabs]{\pi \dem \Lambda, \Gamma\uplus\Delta \vdash \lambda y.u\subs{x}{v} : \iA}{\Lambda, \Gamma\uplus\Delta, y : \iQ \vdash u\subs{x}{v} : \iB}
\] 
where $\iA = \iQ \arrow \iB$. Note that, by definition of substitution, $y$ cannot occur free in $v$; therefore the \ih assures there exist $\pi^u \dem \Lambda, \Gamma, y : \iQ, x : \iP \vdash u : \iB$ and $\pi^v \dem \Lambda, \Delta \vdash v : \iP$. One can easily construct $\pi^t$ starting from $\pi^u$, via rule \iabs. Lastly, note that $\pi$ unique implies its premise is unique; in turn, by \ih, this implies $\pi^u$ (and consequently $\pi^t$) and $\pi^v$ are unique.
\item \label{anti-subs-app} Case $t = su$. Then $\pi$ is:
\[
\infer[\iapp]{\pi \dem \Lambda, \Gamma\uplus\Delta \vdash su \subs{x}{v} : \iA}{\Lambda, \Gamma_1\uplus\Delta_1 \vdash s\subs{x}{v} : \iQ \arrow \iA  &  \Lambda, \Gamma_2\uplus\Delta_2 \vdash u\subs{x}{v} : \iQ }
\] 
By \ih there are $\pi^s \dem \Lambda, \Gamma_1, x : \iP_1 \vdash s : \iQ \arrow \iA$ and $\pi^v_1 \dem \Lambda, \Delta_1 \vdash v : \iP_1$, together with $\pi^u \dem \Lambda, \Gamma_2, x : \iP_2 \vdash u : \iQ$ and $\pi^v_2 \dem \Lambda, \Delta_2 \vdash v : \iP_2$.
Clearly we can build the $\pi^t$ from $\pi^s$ and $\pi^u$ via rule \iapp; for $\pi^v$ we distinguish two cases:
\begin{itemize}
\item If $\iP = \iP_1 = \iP_2$ is ground, then $\Delta$ is empty; consequently $\pi^v = \pi^v_1 = \pi^v_2$.
\item Otherwise $\iP = \iP_1 \uplus \iP_2$ is a multiset; then by \Cref{lemma:anti-split} we can build $\pi^v \dem \Lambda,\Delta \vdash v : \iP$, where $\Delta = \Delta_1\uplus\Delta_2$, starting from $\pi^v_1$ and $\pi^v_2$.
\end{itemize} 
Remark that $\pi$ unique means that both its premises are unique; in turn, by \ih, this implies $\pi^s$, $\pi^u$ (and consequently $\pi^t$), $\pi^v_1$, $\pi^v_2$ (and consequently $\pi^v$) are unique. 
\item \label{anti-subs-let} Case $t = \letin{y}{u}{s}$. The derivation $\pi$ has shape:
\[
\infer[\ilet]{\pi \dem \Lambda, \Gamma\uplus\Delta \vdash (\letin{y}{u}{s})\subs{x}{v} : \iA}{\Lambda, \Gamma_1\uplus\Delta_1 \vdash u\subs{x}{v} : \iQ & \Lambda, \Gamma_2\uplus\Delta_2, y : \iQ \vdash s\subs{x}{v} : \iA}
\]
Observe that $y$ cannot occur free in $v$, by definition of substitution. Then by \ih there are $\pi^u \dem \Lambda, \Gamma_1, x : \iP_1 \vdash u : \iQ$ and $\pi^v_1 \dem \Lambda, \Delta_1 \vdash v : \iP_1$, together with $\pi^s \dem \Lambda, \Gamma_2, y : \iQ, x : \iP_2 \vdash s : \iA$ and $\pi^v_2 \dem \Lambda, \Delta_2 \vdash v : \iP_2$. Again, $\pi^t$ is obtained from $\pi^u$ and $\pi^s$ via rule \ilet; for what concerns $\pi^v$ and the claims about uniqueness, the same considerations we made in case (\ref{anti-subs-app}) apply.
\item Case $t = \letp{\pair{y_1}{y_2}}{w}{s}$. The derivation $\pi$ is:
\[
\infer[\iletp]{\pi \dem \Lambda, \Gamma \uplus \Delta \vdash (\letp{\pair{y_1}{y_2}}{w}{s})\subs{x}{v} : \iA}{\Lambda \vdash w\subs{x}{v} : \iL_1 \otimes \iL_2 & \Lambda, \Gamma \uplus \Delta, y_1 : \iL_1, y_2 : \iL_2 \vdash s\subs{x}{v} : \iA}
\]
By definition of substitution, $y_1$ and $y_2$ cannot occur free in $v$. Then by \ih there are $\pi^w \dem \Lambda, x : \iP_1 \vdash w : \iL_1 \otimes \iL_2$ and $\pi^v_1 \dem \Lambda \vdash v : \iP_1$, together with $\pi^s \dem \Lambda, \Gamma, y_1 : \iL_1, y_2 : \iL_2, x : \iP_2 \vdash s : \iA$ and $\pi^v_2 \dem \Lambda, \Delta \vdash v : \iP_2$. Clearly $\pi^t$ is obtained from $\pi^w$ and $\pi^s$ via rule \iletp. Regarding $\pi^v$ and uniqueness, the reasoning is again similar to case (\ref{anti-subs-app}), the only relevant difference being that either $\iP = \iP_1 = \iP_2$ is ground or $\iP_1 = []$ and $\iP = \iP_2$.
\item Case $t = \oc u$. The derivation $\pi$ is:
\[
\infer[\ibang]{\pi \dem \Lambda, \Gamma\uplus\Delta \vdash (\oc u)\subs{x}{v} : \iA}{ \big( \ \Lambda, \Gamma_i\uplus\Delta_i \vdash u\subs{x}{v} : \iA_i \ \big)_{i=1}^n }
\]
where $\Gamma\uplus\Delta = \biguplus_{i=1}^n(\Gamma_i\uplus\Delta_i)$ and $A = [\iA_1,\dots,\iA_n]$. By \ih we have $\pi^u_i \dem \Lambda, \Gamma_i, x : \iP_i \vdash u : \iA_i$ and $\pi^v_i \dem \Lambda, \Delta_i \vdash v : \iP_i$ for $1 \leq i \leq n$. Clearly $\pi^t$ can be obtained from all $\pi^u_i$ $(1 \leq i \leq n)$ by rule \ibang, while for $\pi^v$ we follow the same reasoning of case (\ref{anti-subs-app}): 
\begin{itemize}
\item If $\iP = \iP_i$, i.e. $\iP$ is ground, then $\Delta$ is empty; consequently $\pi^v = \pi^v_i$ for all $i$.
\item Otherwise $\iP = \biguplus_{i=1}^n \iP_i$ is a multiset; then by \Cref{lemma:anti-split} we can build $\pi^v \dem \Lambda,\Delta \vdash v : \iP$, where $\Delta = \biguplus_{i=1}^n \Delta_i$, starting from all $\pi^v_i$ $(1 \leq i \leq n)$.
\end{itemize} 
The result on uniqueness follows by generalizing the argument presented in case (\ref{anti-subs-app}) to $n$ premises.
In the special case $n=0$, the derivation $\pi^t \dem \Lambda, x : \iP \vdash \oc u : []$ must be obtained by using no premises; therefore both $\pi^t$ and $\pi^v \dem \Lambda \vdash v : \iP$ (observe that $\iP$ is necessarily ground) are uniquely determined.
\item Case $t = \der u$. Then $\pi$ has shape:
\[
\infer[\ider]{\pi \dem \Lambda, \Gamma\uplus\Delta \vdash (\der u)\subs{x}{v} : \iA}{\Lambda, \Gamma\uplus\Delta \vdash u\subs{x}{v} : [\iA]}
\]
The result immediately follows by \ih : indeed, there exist $\pi^u \dem \Lambda, \Gamma, x : \iP \vdash u : \iA$, from which we obtain $\pi^t$ via rule \ider, and $\pi^v \dem \Lambda, \Delta \vdash v : \iP$. The considerations on uniqueness are analogous to case (\ref{anti-subs-abs}).
\end{enumerate}
\end{proof}


\begin{theorem}[Subject Expansion]\label{thm:subject-expansion}
If $\pi' \dem \Lambda, \Gamma \vdash t' : \iA$ and $t \mapsto t'$, then there exists $\pi \dem \Lambda, \Gamma \vdash t : \iA$. Moreover, if $\pi'$ is unique, then $\pi$ is unique.
\end{theorem}

\begin{proof}
The proof is by induction on the reduction context $E$ in which the reduction takes place. Here we examine the base case, in which the reduction context is empty; the inductive cases follow by \ih. There is one subcase per rewriting rule.
\begin{itemize}
\item Rule $\dB$.
We have $t = (\lambda x.s)\eslist u$ and $t' = \eshole{s \esub{x}{u}}\eslist $; we proceed by induction on $\eslist$.
\begin{itemize}
\item Case $\eslist = \shole{\cdot}$. The derivation $\pi'$ has shape:
\begin{prooftree}
	\AxiomC{$\Lambda, \Gamma_2 \vdash u : \iP $}
	\AxiomC{$\Lambda, \Gamma_1, x : \iP \vdash s : \iA $}
\RightLabel{\ilet}
\BinaryInfC{$\pi' \dem \Lambda, \Gamma \vdash s \esub{x}{u} : \iA $}
\end{prooftree}
clearly it is possible to build $\pi$ as follows:
\begin{prooftree}
\AxiomC{$\Lambda, \Gamma_1, x : \iP \vdash s : \iA$}
\RightLabel{\iabs}
\UnaryInfC{$\Lambda, \Gamma_1 \vdash \lambda x.s : \iP \arrow \iA $}
	\AxiomC{$\Lambda, \Gamma_2 \vdash u : \iP $}
\RightLabel{\iapp}
\BinaryInfC{$\pi \dem \Lambda, \Gamma \vdash (\lambda x.s)u : \iA $}
\end{prooftree}
Observe that $\pi'$ unique means that the two premises of the \ilet rule are unique; from this we immediately get that $\pi$ is unique too.

\item Case $\eslist = \eslist'\esub{y}{r}$. The derivation $\pi'$ has shape: 
\begin{prooftree}
\AxiomC{$ \Lambda, \Gamma_1 \vdash r : \iP$}
	\AxiomC{$\rho' \dem \Lambda, \Gamma_2 \uplus \Gamma_3, y : \iP \vdash \eshole{s \esub{x}{u}}\eslist' : \iA$}
\RightLabel{\ilet}
\BinaryInfC{$\pi' \dem \Lambda, \Gamma \vdash \eshole{s \esub{x}{u}}\eslist : \iA $}
\end{prooftree}
By \ih there exists a derivation $\rho$ of shape: 
\begin{prooftree}
\AxiomC{$\Lambda, \Gamma_2, y : \iP \vdash \eshole{\lambda x.s}\eslist' : \iQ \arrow \iA$}
	\AxiomC{$\Lambda, \Gamma_3 \vdash u : \iQ$}
\RightLabel{\iapp}
\BinaryInfC{$\rho \dem \Lambda, \Gamma_2 \uplus \Gamma_3, y : \iP \vdash \eshole{\lambda x.s}\eslist' u : \iA$}
\end{prooftree}
where we can safely assume that $y \not \in \dom(\Gamma_3)$ by hypothesis of rule $\dB$. Therefore one can construct $\pi$ as follows:
\begin{prooftree}
\AxiomC{$ \Lambda, \Gamma_1 \vdash r : \iP $}
	\AxiomC{$ \Lambda, \Gamma_2, y : \iP \vdash \eshole{\lambda x.s}\eslist' : \iQ \arrow \iA $}
\RightLabel{\ilet}
\BinaryInfC{$\Lambda, \Gamma_1 \uplus \Gamma_2 \vdash \eshole{\lambda x.s}\eslist : \iQ \arrow \iA $}
	\AxiomC{$\Lambda, \Gamma_3 \vdash u : \iQ $}
\RightLabel{\iapp}
\BinaryInfC{$\pi \dem \Lambda, \Gamma \vdash \eshole{\lambda x.s}\eslist u : \iA$}
\end{prooftree}
Observe that $\pi'$ unique means that its two premises, one of which is $\rho'$, are unique; this by \ih implies that $\rho$ is unique, and the uniqueness of $\pi$ follows.
\item Rule $\sval$.
Then $t = s \esub{x}{\eshole{v}\eslist}$ and $t' = \eshole{s\subs{x}{v}}\eslist$. Again, we proceed by induction on $\eslist$.
\begin{itemize}
\item $\eslist = \shole{\cdot}$. Starting from $\pi' \dem \Lambda, \Gamma \vdash s\subs{x}{v} : \iA$, \Cref{lemma:anti-substitution} guarantees there exist $\pi_1$ and $\pi_2$ from which we can construct $\pi$:
\begin{prooftree}
\AxiomC{$\pi_1 \dem \Lambda, \Gamma_1 \vdash v : \iQ$}	
	\AxiomC{$\pi_2 \dem \Lambda, \Gamma_2, x : \iQ \vdash s : \iA$}
\RightLabel{\ilet}
\BinaryInfC{$\pi \dem \Lambda, \Gamma \vdash s\esub{x}{v} : \iA$}
\end{prooftree}
The same Lemma provides the result about uniqueness.
\item $\eslist = \eslist'\esub{y}{r}$. Then $\pi'$ is:
\begin{prooftree}
	    \AxiomC{$ \Lambda, \Gamma_1 \vdash r : \iP $}
\AxiomC{$\rho' \dem \Lambda, \Gamma_2 \uplus \Gamma_3 , y : \iP \vdash \eshole{s\isub{x}{v}}\eslist' : \iA$}
\RightLabel{\ilet}
\BinaryInfC{$\pi' \dem \Lambda, \Gamma \vdash \eshole{s \subs{x}{v}}\eslist : \iA$}
\end{prooftree}
and by \ih there exists $\rho$:
\begin{prooftree}
    \AxiomC{$\Lambda, \Gamma_2, y : \iP \vdash \eshole{v}\eslist' : \iQ $}
\AxiomC{$\Lambda, \Gamma_3, x : \iQ \vdash s : \iA$}
\RightLabel{\ilet}
\BinaryInfC{$\rho \dem \Lambda, \Gamma_2 \uplus \Gamma_3, y: \iP \vdash s\esub{x}{\eshole{v}\eslist'} : \iA$}
\end{prooftree}
where $y \not \in \dom(\Gamma_3)$ by hypothesis of rule $\sval$. Hence we conclude by building:
\begin{prooftree}
\AxiomC{$ \Lambda, \Gamma_1 \vdash r : \iP $}
    \AxiomC{$\Lambda, \Gamma_2, y : \iP \vdash \eshole{v}\eslist' : \iQ $}
\RightLabel{\ilet}
\BinaryInfC{$\Lambda, \Gamma_1 \uplus \Gamma_2 \vdash \eshole{v}\eslist : \iQ $}	
	\AxiomC{$\Lambda, \Gamma_3, x : \iQ \vdash s : \iA$}
\RightLabel{\ilet}
\BinaryInfC{$\pi \dem \Gamma \vdash s\esub{x}{\eshole{v}\eslist} : \iA$}
\end{prooftree}
The reasoning about uniqueness is similar to case $\dB$.
\end{itemize}
\item Rule $\dbang$.
In this case $ t = \der \oc s$ and $t' = s$.
Starting from $\pi' \dem \Lambda, \Gamma \vdash s : \iA$, we can easily build:
\begin{prooftree}
\AxiomC{$\pi' \dem \Lambda, \Gamma \vdash s : \iA$}
\RightLabel{\ibang}
\UnaryInfC{$\Lambda, \Gamma \vdash \oc s : [\iA]$}
\RightLabel{\ider}
\UnaryInfC{$\pi \dem \Lambda, \Gamma \vdash \der (\oc s) : \iA$}
\end{prooftree}
One immediately sees that $\pi'$ unique implies $\pi$ unique.
\item Rule $\dpair$. 
In this case $t$ is $\letin{\pair{x_1}{x_2}}{\pair{v_1}{v_2}}{s}$ and $t'$ is $s\esub{x_1}{v_1}\esub{x_2}{v_2}$.
The derivation $\pi'$ has shape:
\begin{prooftree}
\AxiomC{$\Lambda \vdash v_2 : \iL_2$}
	\AxiomC{$\Lambda \vdash v_1 : \iL_1$ }
		\AxiomC{$\Lambda, \Gamma, x_1 : \iL_1, x_2 : \iL_2 \vdash s : \iA$}
	\RightLabel{\ilet}
	\BinaryInfC{$\Lambda, \Gamma, x_2 : \iL_2 \vdash s\esub{x_1}{v_1} : \iA$}
\RightLabel{\ilet}
\BinaryInfC{$\pi' \dem \Lambda, \Gamma \vdash s\esub{x_1}{v_1}\esub{x_2}{v_2} : \iA$}
\end{prooftree}
By using the very same subderivations, and only changing the rules, it is easy to obtain $\pi$:
\begin{prooftree}
\AxiomC{$\Lambda \vdash v_1 : \iL_1 $ }
	\AxiomC{$\Lambda \vdash v_2 : \iL_2 $ }
\RightLabel{\ipair}
\BinaryInfC{$\Lambda \vdash \pair{v_1}{v_2} : \iL_1 \otimes \iL_2 $}
	\AxiomC{$\Lambda, \Gamma, x_1 : \iL_1, x_2 : \iL_2 \vdash s : \iA $}
\RightLabel{\iletp}
\BinaryInfC{$\pi \dem \Lambda, \Gamma \vdash \letp{\pair{x_1}{x_2}}{\pair{v_1}{v_2}}{s} : \iA$}
\end{prooftree}
Again, it easy to check that $\pi'$ unique implies $\pi$ unique.
\end{itemize}
\end{itemize}
\end{proof}

\end{document}